\documentclass{llncs}

\newcommand{\OMIT}[1]{}

\makeatletter
\RequirePackage[bookmarks,unicode,colorlinks=true]{hyperref}%
   \def\@citecolor{blue}%
   \def\@urlcolor{blue}%
   \def\@linkcolor{blue}%

\def\orcidID#1{\smash{\href{http://orcid.org/#1}{\protect\raisebox{-1.25pt}{\protect\includegraphics{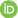}}}}}
\makeatother

\usepackage{defs}

\begin{document}
\allowdisplaybreaks

\title{Intel PMDK Transactions: Specification, Validation and Concurrency (Extended Version)}

\author{Azalea Raad\inst{1}\orcidID{0000-0002-2319-3242} \and Ori Lahav\inst{2}\orcidID{0000-0003-4305-6998} \and  John Wickerson\inst{1}\orcidID{0000-0001-6735-5533} \and \\ Piotr Balcer\inst{3} \and Brijesh Dongol\inst{4}\orcidID{0000-0003-0446-3507}}
\institute{
  Imperial College London, London, UK \\
  \and Tel Aviv University, Tel Aviv, Israel \\
  \and Intel \\
  \and University of Surrey, Guildford, UK}


\maketitle 

\begin{abstract}
  Software Transactional Memory (STM) is an extensively studied
  paradigm that provides an easy-to-use mechanism for thread safety
  and concurrency control. 
  With the recent advent of byte-addressable persistent memory, a
  natural question to ask is whether STM systems can be adapted to support
  {\em failure atomicity}. 
  In this article, we answer this question by showing how STM can be
  easily integrated with Intel's Persistent Memory Development Kit
  (PMDK) transactional library (which we refer to as \PMDKTX) to
  obtain STM systems that are both concurrent and persistent. We
  demonstrate this approach using known STM systems, \TML and \NOREC,
  which when combined with \PMDKTX result in persistent STM systems,
  referred to as \PMDKT and \PMDKN, respectively. However, it turns
  out that existing correctness criteria are insufficient for
  specifying the behaviour of \PMDKTX and our concurrent
  extensions. We therefore develop a new correctness criterion, {\em
    dynamic durable opacity}, that extends the previously defined
  notion of durable opacity with dynamic memory allocation. We provide
  a model of \PMDKTX, then show that this model satisfies dynamic
  durable opacity. Moreover, dynamic durable opacity supports
  concurrent transactions, thus we also use it to show correctness of
  both \PMDKT and \PMDKN.
\end{abstract}

\section{Introduction}
\label{sec:intro}

Persistent memory technologies (aka non-volatile memory, NVM) such
as 
Memory-Semantic SSD~\cite{samsungpressrelease} and 
XL-FLASH~\cite{techinsightsblog}, combine the durability of hard
drives with the fast and fine-grained accesses of DRAMs, with the
potential to radically change how we build fault-tolerant systems.
However, NVM also raise fundamental questions about semantics and the
applicability of standard programming models.

Among the most widely
used collections of libraries for persistent programming is Intel's
Persistent Memory Development Kit (\PMDK), which was first released in
2015~\cite{PMDK}.  One important component of \PMDK is its
transactional library, which we refer to as \PMDKTX, and which
supports generic {\em failure-atomic} programming. A programmer can
use \PMDKTX to protect against full system crashes by starting a
transaction, performing transactional reads and writes, then
committing the transaction. If a crash occurs during a transaction,
but before the commit, then upon recovery, any writes performed by the
transaction will be rolled back. If a crash occurs during the commit,
the transaction will either be rolled back or be committed
successfully, depending on how much of the commit operation has been
executed. If a crash occurs after committing, the effect of the
transaction is guaranteed to have been persisted. 

Most software transactional memory (STM) algorithms leave memory allocation
implicit, since they are generally safe under standard allocation
techniques (\eg{} {\code{malloc}). Memory that is allocated as part
of a transaction can be deallocated if the transaction is
aborted. However, in the context of persistency, memory allocation is
more subtle since transactions may be interrupted by a crash.

\begin{figure}[t!]
\centering 
\ \ \ \
\begin{minipage}[t]{0.9\columnwidth}
\begin{lstlisting}
struct loc {
  pmem::obj::p<int> value;                 (*\label{line:locval}*)
  pmem::obj::persistent_ptr<loc> next; }; (*\label{line:locnext}*)
  
struct root { pmem::obj::persistent_ptr<loc> head = nullptr; }; (*\label{line:root}*)

void post_crash(...) {(*\label{line:rec}*)
 auto pop = pmem::obj::pool<root>::open("file",...); (*\label{line:rec-pool}*)
 auto root = pop.root(); (*\label{line:rec-makeroot}*)
 pmem::obj::transaction::run(pop, [&]{              (*\label{line:rec-txbegin}*)
    auto xvalue = root->head->value;   (*\label{line:rec-getxvalue}*)
 }); }      (*\label{line:rec-txend}*)

int main(...) {
  auto pop = pmem::obj::pool<root>::open("file",...); (*\label{line:pool}*)
  auto root = pop.root(); (*\label{line:makeroot}*)
  pmem::obj::transaction::run(pop, [&]{              (*\label{line:txbegin}*)
    auto x = pmem::obj::make_persistent<loc>(); (*\label{line:x}*)
    x->value = 42;   (*\label{line:xvalue}*)
    x->next = nullptr;   (*\label{line:xnext}*)
    root->head = x;  (*\label{line:headtox}*)
  }); }                (*\label{line:txend}*)
\end{lstlisting}   
\end{minipage}
\vspace{-10pt}
\caption{C++ snippet for allocating in persistent memory  using \PMDKTX~\cite{Scargall2020}}
\label{fig:alloc}
\end{figure}

For example, consider the program in \cref{fig:alloc}.  Persistent
memory is allocated, accessed and maintained via \emph{memory pools}~\cite{Scargall2020}
(files that are memory mapped into the process address space) of a
certain type (\eg of type \code{loc} in \cref{fig:alloc}).  Due to
address space layout randomization (ASLR) in most operating systems,
the location of the pool can differ between executions and across
crashes.  As such, 
every pool has a root object used as an entry
point from which all other objects in the pool can be found. That is,
to avoid memory leaks, all objects in the pool must be reachable from
the root.  An application locates the root object using a \emph{pool
  object pointer} (POP) that is to be created with every program
invocation (\eg line \ref{line:pool}).  After locating the pool root
(line \ref{line:makeroot}), we  use a \PMDKTX transaction (lines
\ref{line:txbegin}--\ref{line:txend}) to allocate a persistent
\code{loc} object \code{x} (line \ref{line:x}) with value 42 (line
\ref{line:xvalue}) and add it to the pool (line \ref{line:headtox}).

Consider the scenario where the execution of this transaction crashes.
After recovery from the crash,  
we then execute \code{post\_crash} (line \ref{line:rec}). As before, we open the pool (line \ref{line:rec-pool}) and locate its root (line \ref{line:rec-makeroot}). 
We then use a \PMDKTX transaction to read from the \code{loc} object allocated and added at the pool head prior to the crash (line \ref{line:rec-getxvalue}). 
There are then three cases to consider: the crash may have occurred
\begin{enumerate*}[label=(\bfseries\arabic*)]
\item before the transaction started the commit process,
\item after the transaction successfully committed, or
\item while the transaction was in the process of committing.
\end{enumerate*}

In case (1), the execution of the two transactions can be depicted as
follows, where the \code{PBegin} events capture commencing the
transactions (lines \ref{line:txbegin} and \ref{line:rec-txbegin}),
\code{PAlloc(x)} denotes the persistent allocation of \code{x}
(line \ref{line:x}); \code{PWrite(x->value,42)} captures writing to
\code{x} (line \ref{line:xvalue}); and \code{PRead(root->head):x}
denotes reading from \code{x->value} and returning the value \code{x} (first part of line \ref{line:rec-getxvalue}). 
As the first transaction never reached the commit stage, its effects
(\ie allocating \code{x} and writing to it) should be invisible (\ie
rolled back), and thus the read of the second transaction effectively
reads from unallocated memory, leading to an error such as a
segmentation fault.
%
%
%
%
%
%
%
%
%

\noindent \ \ 
  \scalebox{0.85}{
  \begin{tikzpicture}
  \draw [dashed,->] (-2.5,0) -- (11.2,0);
  \draw [Bar-Bar, thick, red] (-2.4,0) -- node[above]{\tt \footnotesize PBegin} (-1.5,0);
  \draw [Bar-Bar, thick, red] (-1.3,0) -- node[above]{\tt \footnotesize PAlloc(x)} (-0.2,0);
  \draw [Bar-Bar,  thick, red] (0.2,0) -- node[below,pos=.9]{\tt \ \ \ \ \footnotesize \intab{PWrite\\[-2pt](x->value,42)}} (1,0);
  \draw [Bar-Bar,  thick, red] (1.2,0) --
  node[above]{\tt \footnotesize \intab{PWrite \\[-2pt] (x->next,...)}} (3,0);
  \draw [Bar-Bar,  thick, red] (3.4,0) --
  node[above]{\tt \footnotesize \ \intab{PWrite \\[-2pt] (root->head,x)}} (5.2,0);
  \node at (5.7,0) {\Huge\Lightning};
  \draw [Bar-Bar,  thick, blue] (6.2,0) -- node[below]{\tt \footnotesize PBegin} (6.8,0);
  \draw [Bar-Bar,  thick, blue] (7.1,0) -- node[above]{\tt \footnotesize \intab{PRead\\[-2pt](root->head):x}} (8.8,0);
  \draw [Bar-,  thick, blue] (9,0) -- node[below]{\tt \footnotesize \ \ \ \ \intab{PRead\\[-2pt](x->value)}} (10,0);
  \node at (10.5,0.2) {\small SegFault};
  \end{tikzpicture}}

In case (2), the execution of the transactions is as follows, where 
the \code{PCommit} events capture the end (successful commit) of the transactions (lines \ref{line:txend} and \ref{line:rec-txend}), the effects of the first transaction fully persist upon successful commit, and thus the read in the second transaction does not fault.

\noindent \ \ 
    \scalebox{0.85}{
  \begin{tikzpicture}
  \draw [dashed,->] (0,0) -- (13.5,0);
  \draw [Bar-Bar,  thick, red] (0.3,0) -- node[above]{\tt \footnotesize PBegin} (1.4,0);
  \draw [Bar-Bar,  thick, red] (1.6,0) -- node[above]{\tt \footnotesize PAlloc(x)} (2.8,0);
  \draw [Bar-Bar,  thick, red] (3.2,0) -- node[above]{\tt  ...} (4,0);
  \draw [Bar-Bar,  thick, red] (4.2,0) -- node[above]{\tt \footnotesize \phantom(PCommit\phantom(} (5.2,0);
  \node at (5.6,0) {\Huge\Lightning};
  \draw [Bar-Bar,  thick, blue] (6.0,0) -- node[above]{\tt \footnotesize PBegin} (7.0,0);
  \draw [Bar-Bar,  thick, blue] (7.2,0) -- node[above]{\tt \footnotesize
    \intab{PRead\\[-2pt](root->head):x}} (9.2,0);
  \draw [Bar-Bar,  thick, blue] (9.6,0) -- node[above]{\tt  \footnotesize \intab{PRead\\(x->value):42}} (11.6,0);
  \draw [Bar-Bar,  thick, blue] (12,0) -- node[above]{\tt \footnotesize PCommit} (13,0);
  \end{tikzpicture}  }

\noindent Finally, in case (3), either of the two behaviours depicted above is
possible (\ie the second transaction may either cause a segmentation
fault or read from \code{x}).

Efficient and correct memory allocation in a persistent memory setting
is challenging~(\cite[Chapter 16]{Scargall2020} and
\cite{PMDK-Alloc}). In addition to the ASLR issue mentioned above, the
allocator must guarantee failure atomicity of heap operations on
several internal data structures managed by
\PMDK. 
Therefore, \PMDK provides its own allocator that is
designed specifically to work with
\PMDKTX. 

We identify two key drawbacks of \PMDKTX as follows. In this article, we
take significant steps towards addressing both of these drawbacks.
\begin{description}[style=unboxed,leftmargin=0cm,itemsep=5pt]
\item [A)] {\bf Lack of concurrency support.} 
Unlike existing STM systems in the persistent setting \cite{kaminotx,timestone} that provide \emph{both failure atomicity} (ensuring that a transaction either commits fully or not at all in case of a crash) \emph{and isolation} (as defined by ACID properties, ensuring that the effects of incomplete transactions are invisible to concurrently executing transactions), 
\PMDKTX only provides failure atomicity and does not offer isolation in concurrent settings. In particular, naively implemented applications with racy \PMDK transactions lead to memory inconsistencies. 
This is against the spirit of STM: the primary function of STM systems is providing a concurrency control mechanism that ensures isolation. 
The current \PMDKTX implementation provides two solutions: threads either execute {\em concurrent} transactions over {\em disjoint} parts of the memory~\cite[Chapter 7]{Scargall2020}, or use user-defined {\em fine-grained locks} within a transaction to ensure memory isolation~\cite[Chapter 14]{Scargall2020}. 
However, both solutions are sub-optimal: the former enforces serial execution when transactions operate over the same part of the memory, and the latter expects too much of the user. 

\item [B)] {\bf Lack of a suitable correctness criterion.} There is no
  formal specification describing the desired behaviour of
  \PMDKTX, and hence no rigorous description or correctness proof 
  of its implementation. This undermines the
  utility of \PMDKTX in safety-critical settings and makes it impossible to
  develop formally verified applications that use \PMDKTX.
  Indeed, there is currently no correctness criterion, be it in the volatile or the persistent setting, for STM systems that provide dynamic memory allocation (a large category that includes all realistic implementations).
\end{description}

\subsection{Concurrency for \PMDKTX}


Integrating concurrency with PMDK transactions is an important end
goal for PMDK developers. The existing approach requires integration
of locks with \PMDKTX, which introduces overhead for programmers. Our
article shows that STM and PMDK can be easily combined, improving
programmability. Many other works have aimed to develop failure-atomic
and concurrent transactions (\eg{}
OneFile~\cite{DBLP:conf/dsn/RamalheteCFC19} and
Romulus~\cite{DBLP:conf/spaa/CorreiaFR18}), but none use off-the-shelf
commercially available libraries. Moreover, these other works have not
addressed correctness with the level of rigour that our article
does. In other work, popular key-value stores Memcached and Redis have
been ported to use
PMDK~\cite{DBLP:conf/asplos/0001SWWKK20,DBLP:conf/asplos/0001WZKK19};
our work paves the way for concurrent version of these applications to
be developed.  Another example is the work of Chajed et
al~\cite{DBLP:conf/osdi/ChajedTT0KZ21}, who provide a simulation-based
technique for \emph{verifying} refinement of durable filesystems,
where concurrency is handled by durable transactions.

We tackle the first drawback (A) mentioned above by developing, specifying, and validating two thread-safe versions of \PMDKTX.
Note that our primary focus here is \emph{correctness} rather than performance. 
Specifically, we aim to confirm that \PMDKTX can be combined with off-the-shelf thread-safe STM systems, and to provide a correct, thread-safe (\ie concurrent) and failure-atomic (\ie persistent) STM system.

\begin{description}[style=unboxed,leftmargin=0cm,itemsep=2pt]
\item [Contribution A: Making \PMDKTX thread-safe.]
We combine \PMDKTX with two off-the-shelf (thread-safe) STM systems, \TML~\cite{DBLP:conf/europar/DalessandroDSSS10} and
\NOREC~\cite{DBLP:conf/ppopp/DalessandroSS10}, to obtain two
new implementations, \PMDKT and \PMDKN, that support
{\em concurrent} failure-atomic transactions with dynamic memory
allocation. 
In particular, we reuse the existing concurrency control mechanisms provided by these STM systems to
ensure atomicity of write-backs, thus obtaining memory isolation even
in a multi-threaded setting. We show that it is possible to integrate
these mechanisms with \PMDKTX to additionally achieve failure
atomicity. Our approach is modular, with a clear separation of
concerns between the isolation required due to concurrency and the
atomicity required due to the possibility of system crashes. 
This shows that concurrency and failure atomicity are two
orthogonal concerns, highlighting a pathway towards a
mix-and-match approach to combining (concurrent) STM and failure-atomic transactions.
Finally, in order to provide the same interface as \PMDK, we extend both \TML
and \NOREC with an explicit operation for memory
allocation. Note that \TML and \NOREC are fundamentally different
designs. \TML performs {\em eager write-backs}: writes occur in
place, \ie when the write operation is executed. By contrast,
\NOREC performs {\em lazy write-backs}: writes are recorded in
transaction-local write-sets and write-backs occur when the
transaction commits. The application of \PMDKTX to two different STM
designs demonstrates the \emph{modularity} of our approach.
Specifically, we used both \NOREC and \TML as they are without making any adjustments to accommodate \PMDKTX. This suggests that \PMDK could similarly be combined with other existing STM systems \emph{off the shelf}. 

\end{description}

\definecolor{colorrefines}{HTML}{1b9e77}
\definecolor{coloruses}{HTML}{e7298a}
\definecolor{colorextends}{HTML}{e7298a}
\definecolor{colorfaithful}{HTML}{d95f02}
\definecolor{colorspec}{HTML}{e6f5c9}
\definecolor{colorOPspec}{HTML}{fff2ae}
\definecolor{colorcontrib}{HTML}{bbbbbb}

\tikzset{
  spec/.style={draw=black, align=center, rounded corners, fill=colorspec},
  specOP/.style={draw=black, align=center, rounded corners, fill=colorOPspec},
  impl/.style={draw=black, dashed, align=center},
  exec/.style={draw=black, align=center, rounded corners},
  refines/.style={-latex, colorrefines, every node/.style={color=colorrefines}},
  extends/.style={-latex, dashed, thick, colorextends, every node/.style={color=colorextends}},
  faithful/.style={-latex, colorfaithful, dotted, ultra thick, every node/.style={color=colorfaithful}},
  uses/.style={-latex, dashed, thick, coloruses, every node/.style={color=coloruses}},
}

\begin{figure}[t]
\scalebox{0.75}{
\begin{tikzpicture}[yscale=-1]
\node[spec] (ddOpacity) at (0,-0.5) {
dynamic 
durable opacity \\
\DDO \\
(Contribution B2) \\
\cref{sec:volatile_framework}};

\node[spec] (dOpacity) at (-5.7,0.3) {
durable opacity 
\cite{DBLP:conf/forte/BilaDDDSW20}};

\node[spec] (opacity) at (-5.7,-1.3) {
opacity 
\cite{DBLP:series/synthesis/2010Guerraoui}};

\node[specOP] (ddTMS) at (0,2) {
\DDTMS \\
(Contribution B3) \\
\cref{sec:DTMS2-full}};

\node[specOP] (dTMS2) at (-5.7,2) {
\DTMS 
\cite{DBLP:conf/forte/BilaDDDSW20}};

\node[exec] (PMDKmodel) at (0,4) {
PMDK \\
(Contribution B1) \\
\cref{sec:pmdk-implementation}};

\node[impl] (PMDKimpl) at (-5.7,4) {
implementation \\ of PMDK~\cite{PMDK}};

\node[exec] (PMDKTMLmodel) at (3.5,4) {
\PMDKT \\
(Contribution A) \\
\cref{sec:pmdkt}};

\node[exec] (PMDKNORECmodel) at (7,4) {
\PMDKN \\
(Contribution A) \\
\cref{sec:pmdkn}};


\draw[colorcontrib, very thick, rounded corners] (-3.5,-1.6) rectangle (8.7,5.1);
\node[anchor=north east, color=colorcontrib] at (8.7,-1.6) {\bf contributions of this article};

\draw[extends] (opacity) to[auto] node {extends} (dOpacity);
\draw[extends] (dOpacity) to[auto, pos=0.63] node {extends} (ddOpacity.west |- dOpacity);
\draw[refines] (dTMS2) to[auto, swap] node {refines} (dOpacity);
\draw[extends] (dTMS2) to[auto, pos=0.6] node {extends} (ddTMS);
\draw[refines] (ddTMS) to[auto, swap] node {refines (\cref{thm:soundness-ddtms-1})} (ddOpacity);
\draw[faithful] (PMDKimpl) to[auto, pos=0.6] node[align=center] {
faithful \\ model} (PMDKmodel);
\draw[refines] (PMDKmodel) to[auto, swap] node {refines (\FDR)} (ddTMS);
\draw[refines] (PMDKTMLmodel) to[bend left, auto, swap] node[inner sep=0pt] {refines (\FDR)} (ddTMS);
\draw[refines] (PMDKNORECmodel) to[bend left, auto, swap, pos=0.2] node[inner sep=0pt] {refines (\FDR)} (ddTMS);
\draw[uses] (PMDKTMLmodel) to[auto, swap] node {uses} (PMDKmodel);
\draw[uses] (PMDKNORECmodel) to[bend right=20,above, pos=0.1] node {uses} (PMDKmodel);

\node[spec, anchor=east] at (4,-0.5) {\phantom{A}};
\node[anchor=west] at (4,-0.5) {= declarative specification};
\node[specOP, anchor=east] at (4,0) {\phantom{A}};
\node[anchor=west] at (4,0) {= operational specification};
\node[exec, anchor=east] at (4,0.5) {\phantom{A}};
\node[anchor=west] at (4,0.5) {= executable abstraction};
\node[color=colorrefines,anchor=east] at (4,1) {$\rightarrow$};
\node[anchor=west] at (4,1) {= contribution B4};
\end{tikzpicture}
}
\caption{The contributions of this article and their relationships to prior work}
\label{fig:contributions}
\end{figure}
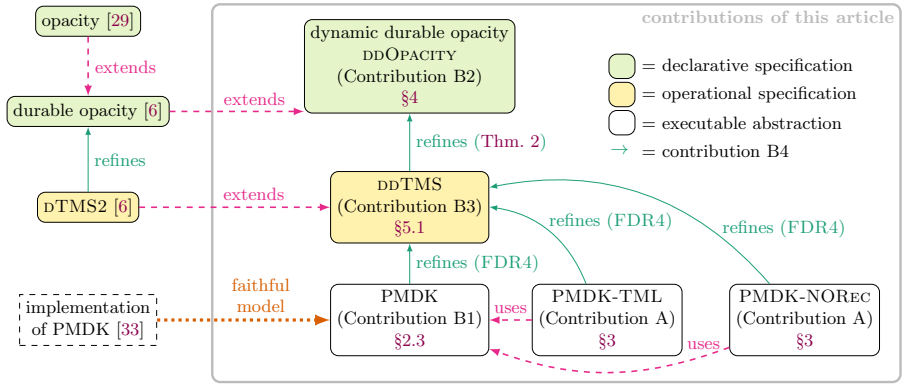
\subsection{Specifying and Validating \PMDKTX, \PMDKT and \PMDKN}

To tackle drawback (B) above, our contributions come in four
parts. Together, they provide the first formal (and rigorous)
specification of \PMDKTX and validation of its implementation.

\begin{description}[style=unboxed,leftmargin=0cm,itemsep=5pt]
\item [Contribution B1: A model of \PMDKTX]
We provide a formal specification of \PMDKTX as an abstract transition system.  
Our formal specification models almost all key components of \PMDKTX (including its redo and undo logs, as well as the interaction of these components with system crashes), with the exception of memory deallocation within \PMDKTX transactions. 
\item [Contribution B2: A correctness criterion for transactions with dynamic allocation]
  Although the literature includes several correctness criterion for transactional memory (TM), none can adequately capture \PMDKTX in that they do not account for dynamic memory allocation. 
Specifically, opacity is a standard correctness criterion for 
transactions~\cite{DBLP:series/synthesis/2010Guerraoui}, but does not account for durability guarantees (\eg failure atomicity) in case of crashes. 
Durable opacity~\cite{DBLP:conf/forte/BilaDDDSW20} extends opacity to handle crashes; 
nevertheless, durable opacity does not account for dynamic allocation
and is thus inadequate for specifying \PMDKTX. 
We develop a new correctness
condition, {\em dynamic durable opacity} (denoted \DDO), by extending durable opacity
to account for dynamic allocation. \DDO
supports not only sequential transactions such as \PMDKTX, but also concurrent
ones. To demonstrate the suitability of \DDO for concurrent and persistent (durable) transactions, later we validate our two concurrent \PMDKTX implementations (\PMDKN and \PMDKT) against \DDO.
Note that although \DDO is a straightforward extension of durable opacity, validating our implementations against \DDO is far from simple, in part due to the additional proof obligations induced by dynamic allocation. 
  
\item [Contribution B3: An operational characterisation of our
  correctness criterion]
Our aim is to show that \PMDKTX conforms to \DDO, or more precisely,
that our model of \PMDKTX refines our model of \DDO. In order to
demonstrate this refinement relation, we use a new intermediate model
called \DDTMS. While \DDO is defined declaratively (that is, a set of
axioms over execution traces that are consistent), \DDTMS is
defined operationally, which makes it conceptually closer to
our model of the \PMDKTX implementation. We prove that \DDTMS is a sound
model of \DDO (\ie every trace of \DDTMS satisfies
\DDO). \DDTMS builds upon \DTMS~\cite{DBLP:conf/forte/BilaDDDSW20},
which is an operational model for durable opacity.

\item [Contribution B4: Validation of \PMDKTX, \PMDKT and
  \PMDKN in \FDR]
  We mechanise our implementations (\PMDKTX, \PMDKT and \PMDKN) and
  specification (\DDTMS) using the CSP modelling language and
  use the \FDR model checker~\cite{fdr} to show the
  implementations are refinements of \DDTMS over both the persistent
  SC (\PSC)~\cite{DBLP:journals/pacmpl/KhyzhaL21} and persistent TSO
  (\spxes)~\cite{DBLP:journals/pacmpl/RaadWV19} memory models.  For
  \spxes, we use an equivalent formulation called \PTSOS developed by
  Khyzha and Lahav~\cite{DBLP:journals/pacmpl/KhyzhaL21}. The proof itself is fully
  automatic, requiring no user input outside of the encodings of the
  models themselves.  Additionally, we develop a sequential lower
  bound (\DDTMS-Seq), derived from \DDTMS, and show that this lower
  bound refines \PMDKTX (and hence that \PMDKTX is not vacuously strong).
  Our approach is based on an earlier technique for proving durable
  opacity~\cite{DBLP:conf/sefm/DongolL21}, but incorporates much more
  sophisticated examples and memory models, representing a significant
  challenge in our work.  Our models are reusable and can be adapted
  to other algorithms.
\end{description}

\paragraph{Outline} \cref{fig:contributions} gives an overview of the
  different components that we have developed in this article and their
  relationships to each other and to prior work.  The declaratively
  defined notion of \DDO forms the basis of our approach. Our abstract
  operational model \DDTMS is proved sound with respect to \DDO, then
  \PMDK is shown to refine \DDTMS. Our concurrent extensions \PMDKT and \PMDKN (which
  use \PMDKTX) are also shown to refine \DDTMS.
  We structure our article by presenting the components of
  \cref{fig:contributions} roughly from the bottom up. In
  \cref{sec:pmdk}, we present the abstract \PMDKTX model, and in
  \cref{sec:making-pmdk-trans} we describe its integration with STM
  to provide concurrency support via \PMDKT and \PMDKN. In
  \cref{sec:volatile_framework} we present \DDO; in
  \cref{sec:oper-prov-dynam} we present \DDTMS; and in
  \cref{sec:modell-verify-corr} we describe our \FDR encodings and
  bounded proofs of refinement.
  In \cref{sec:related} we present a summary of our caveats, discuss related and future work and conclude.

  \paragraph{Additional Material} We provide our \FDR development as
  supplementary material to this submission. The proofs of all
  theorems stated in the article are given in the accompanying
  appendix.
  

 

\# Contributions so far:
\begin{enumerate}
\item First full specification of (sequential) \PMDK on paper -
  eyeball accuracy, approved by Piotr
\item Mechanised specification of (sequential) \PMDK in \FDR
\item \begin{enumerate}
  \item First specification of \DDTMS = \DTMS + allocs on paper
  \item Mechanised specification 3a in \FDR
  \item Verification of 2 against 3b in \FDR (with certain bounds)
    \[
      \DDTMSSeq \leq \PMDK \leq \DDTMS
    \]
     (upper and lower bounds)
   \item ??? On paper verification of 2 against 3b (unbounded) :( 
  \end{enumerate}
\item  Two concurrent \PMDK reference implementations:

  \begin{enumerate}
  \item \PMDKN = \NOREC + \PMDK -- using \NOREC and \PMDK off the
    shelf 
  \item \PMDKT = \TML + \PMDK -- integrating \TML and \PMDK off the
  \end{enumerate}
  Write backs in \PMDKN are lazy, but in \PMDKT writes are eager

\item Mechanising both \PMDKN and \PMDKT in \FDR and proving
  \begin{gather*}
    \DDTMSSeq \leq \PMDKN \le \DDTMS
    \\
    \DDTMSSeq \leq \PMDKT \le \DDTMS
  \end{gather*}
\item 
6c) ??? On paper verification of 5b against 3b (unbounded) :( 

6d) ??? On paper verification of 5d against 3b (unbounded) :(







\end{enumerate}

\# Next steps

i) Adapt 5a, 5b (from PSC) to Px86

ii) Prototype implementation of i on real hardware + NVM (Sergey? + Gregory?)

iii) Universal construction for data structures: wrap sequential methods in concurrent PMDK transactions -> correct concurrent and persistent implementation

What does correct mean in this context? Durably linearisable? Buffered durably linearisable? (Ori says not buffered)

\# Modularity

Assume a transactional interface with Begin, Commit, Abort, and read/write inside the transaction

Assume a: an  atomic concurrency control mechanism, e.g. NOrec (no ownership records)
	   b: a sequential persistency control mechanism, e.g. PMDK

What are the requirements on a and b to make a+b a correct (strict durable opacity) transactional implementation? 

NB: strict = if an operation is incomplete before a crash, it is either completed right after recovery (at the beginning of the next era) or never.



\section{Intel \PMDK transactions}
\label{sec:pmdk}

We describe the abstract interface \PMDKTX provides
to clients (\cref{sec:pmdk-interface}), our assumptions about the memory model over which \PMDKTX is run (\cref{sec:memory-models}) and the operations of
\PMDKTX (\cref{sec:pmdk-implementation}). We present our \PMDK abstraction in \cref{sec:pmdk-implementation}.  

\subsection{\PMDK Interface}
\label{sec:pmdk-interface}
\begin{figure}[t]
  \centering 
  \begin{minipage}[t]{0.95\columnwidth}
    \begin{lstlisting}
struct queue_node {
    pmem::obj::p<int> value;                 (*\label{node1}*)
    pmem::obj::persistent_ptr<queue_node> next; }; (*\label{node2}*)
  
struct queue { private:
    pmem::obj::persistent_ptr<queue_node> head = nullptr;    (*\label{queue1}*)
    pmem::obj::persistent_ptr<queue_node> tail = nullptr; }; (*\label{queue2}*)
  
void push(pmem::obj::pool_base &pmem_op, int value) {
   pmem::obj::transaction::run(pmem_op, [&]{              (*\label{push1}*)
      auto node = pmem::obj::make_persistent<queue_node>(); (*\label{push2}*)
      node->value = value;   (*\label{push3}*)
      node->next = nullptr;  (*\label{push4}*)
      if (head == nullptr) { (*\label{push4}*)
         head = tail = node; (*\label{push5}*)
      } else {               (*\label{push6}*)
         tail->next = node;  (*\label{push7}*)
         tail = node; }       (*\label{push8}*)
    }); }                (*\label{push9}*)\end{lstlisting}
\end{minipage}
\vspace{-10pt}
\caption{C++ persistent \code{push} operation  using \PMDKTX
  (\cite[pg 131]{Scargall2020})}
\label{fig:persist-push}
\end{figure}
\noindent\PMDK provides an extensive suite of libraries for simplifying
persistent programming.  The \PMDK transactional library (\PMDKTX) has been
designed to support failure-atomicity by providing operations
for tracking memory locations that are to be made persistent, as well
allocating and accessing (reading and writing) persistent memory within an atomic
block.

In \cref{fig:persist-push} we present an example client code that uses
\PMDKTX. The code (due to \cite[p.~131]{Scargall2020})
implements the \code{push} operation for a persistent linked-list
queue.
  The implementation wraps a typical (non-persistent) \code{push}
  operation within a transaction using a C++ lambda \code{[\&]}
  expression (\lineno{push1}). The transaction is invoked using
  \code{transaction::run}, which operates over the memory pool
  \code{pmem\_op}.  
  The node structure (lines \ref{node1} and \ref{node2}), the queue
  structure (lines \ref{queue1} and \ref{queue2}), and any new node
  declaration (\lineno{push2}) are to be tracked by a PMDK
  transaction. 
  Additionally, the \code{push}
  operation takes as input the {\em persistent memory object pool},
  \code{pmem\_op}, which is a memory pool 
  on which the transaction is to be executed. This argument is needed
  because the application memory may map files from different file
  systems. On \lineno{queue2} we use
  \code{make\_persistent} 
  to perform a transactional allocation on persistent memory that is
  linked to the object pool \code{pmem\_op} (see \cite{Scargall2020}
  for details). 
  The remainder of the operation (lines~\ref{push3}--\ref{push8})
  corresponds to an implementation of a standard \code{push} operation
  with (transactional) reads and writes on the indicated locations. At
  \lineno{push9}, the C++ lambda and the transaction is closed,
  signalling that the transaction should be
  committed.  

If the system crashes while \code{push} is executing, but before
\lineno{push9} is executed, then upon recovery, the entire \code{push}
operation will be rolled back so that the effect of the incomplete
operation is not observed, and the queue remains a valid linked
list. After \lineno{push9}, 
the corresponding transaction executes a
commit operation. If the system crashes during commit, depending on
how much of the commit operation has been executed, the \code{push}
operation will either be rolled back, or committed successfully. Note
that roll-back in all cases ensures that the allocation
at \lineno{push2} is undone. 

While \PMDK's use of transactions provides a convenient and programmer-friendly way of
achieving failure-atomicity, there is no formal description (specification) or correctness proof  of its underlying implementation, nor
is there any concurrency support over a single object pool. In what
follows, we address these shortcomings. We proceed with an abstract
description of \PMDKTX.

\subsection{Memory Models}
\label{sec:memory-models}

We consider the execution of our implementations over two different
memory models: \PSC and
\PTSOS~\cite{DBLP:journals/pacmpl/KhyzhaL21}. Both models include a
\code{flush\;x} instruction 
to persist the contents of the given location \code{x} to
memory. \PTSOS aims for fidelity to the Intel x86 architecture. In a
race-free setting (as is the case for single-threaded \PMDKTX
transactions) it is sound to use the simpler \PSC model, though we
conduct all of our experiments in both models.

\PSC 
is a simple model that considers persistency effects and their
interaction with sequential consistency. Writes are propagated
directly to per-location persistence buffers, and are subsequently
flushed to non-volatile memory, either due to a system action, or the
execution of a \code{flush} instruction.  A read from \texttt{x} first
attempts to fetch its value from the 
persistence buffer and
if this fails, fetches its value from non-volatile memory.

Under Intel-x86, the memory models are further complicated by the
interaction between \emph{total store ordering} (TSO)
effects~\cite{DBLP:conf/tphol/OwensSS09} and persistency. Due to the
abstract nature of our models (see \cref{fig:pmdk}) it is sufficient
for us to focus on the simpler \spxes
model~\cite{DBLP:journals/pacmpl/RaadWV19} since we do not use any of
the advanced features~\cite{DBLP:journals/pacmpl/RaadMV22,DBLP:journals/pacmpl/RaadWNV20,DBLP:journals/pacmpl/RaadWV19}. We
introduce a further simplification via \PTSOS that is {\em
  observationally equivalent} to
\spxes~\cite{DBLP:journals/pacmpl/KhyzhaL21}. Unlike \spxes, which
uses a single (global) persistence buffer, \PTSOS uses per-location
buffers simplifying the resulting \FDR models
(\cref{sec:modell-verify-corr}).

In \PTSOS, writes are propagated from the store buffer in FIFO order
to a per-location FIFO persistency buffer. Writes in the persistency
buffer are later persisted to the non-volatile memory.
A read from location \texttt{x} first attempts to fetch
the latest write to \texttt{x} from the store
buffer. If this fails (\ie no writes to \texttt{x} exists in the
store buffer), it attempts to fetch the latest write from the
persistence buffer of \texttt{x}, and if this fails, it
fetches the value of \texttt{x} from non-volatile memory.

\begin{figure}[!t]
  \noindent\, \footnotesize
  \begin{minipage}[t]{0.98\linewidth}
\begin{lstlisting}
// Each location is persistent; there is no explicitly volatile memory.
mem : loc -> {
  val : int; // the contents of this location
  metadata : bool; } // false = not allocated, true = allocated
freeList : loc list // transient list of free locations

// Redo logs -- tredo is transient; pRedo is persistent.
tRedo$_\txid$, pRedo$_\txid$ : {undoValid:bool;  checksum:int;  allocs:loc set;}
undo$_\txid$ : loc -> int  // undo log recording the original val of each loc
undoValid : bool  // undoValid global flag, initially true 
\end{lstlisting}
\end{minipage}

\quad\!\!\!\!
  \begin{minipage}[t]{0.44\columnwidth}
  \begin{lstlisting}[firstnumber=11]
PBegin$_\txid$ $\eqdef$     
   tRedo$_\txid$ := (true, -1, {})  (*\label{begin1}*) 
   pRedo$_\txid$ := (true, -1, {})  (*\label{begin3}*)
   undo$_\txid$ := {}               (*\label{begin2}*) 
   undoValid$_\txid$ := true        (*\label{begin3}*)

PAlloc$_\txid$ $\eqdef$      
   x$_\txid$ := freeList.take                                   (*\label{alloc1}*)
   tRedo$_\txid$.allocs :=
     tRedo$_\txid$.allocs $\cup$ {x$_\txid$} (*\label{alloc3}*)
   return x$_\txid$                                             (*\label{alloc4}*)

PRead$_\txid$(x) $\eqdef$
   return mem[x].val                          (*\label{read1}*)
   
PWrite$_\txid$(x,v) $\eqdef$    
   if x $\notin$ dom(undo$_\txid$) then                (*\label{write2}*)
      w$_\txid$ := mem[x].val                          (*\label{write3}*)
      undo$_\txid$ := undo$_\txid$ $\cup$ {x $\mapsto$ w$_\txid$}    (*\label{write4}*)
      flush undo$_\txid$                                   (*\label{write5}*)
   mem[x].val := v                             (*\label{write6}*)
   
PCommit$_\txid$ $\eqdef$  
   persist_writes$_\txid$                        (*\label{commit1}*)
   tRedo$_\txid$.undoValid := false                 (*\label{commit2}*)
   tRedo$_\txid$.checksum :=                    (*\label{commit3}*)(*\Suppressnumber*)
         calc_checksum(tRedo$_\txid$)           (*\Reactivatenumber*)
   pRedo$_\txid$ := tRedo$_\txid$                    (*\label{commit4}*)
   flush pRedo$_\txid$                               (*\label{commit5}*) 
   apply_pRedo$_\txid$                          (*\label{commit6}*)
   pRedo$_\txid$.checksum := -1               (*\label{commit7}*)
   flush pRedo$_\txid$.checksum                               (*\label{commit8}*)\end{lstlisting}
\end{minipage}
\hfill
\begin{minipage}[t]{0.49\columnwidth}
\begin{lstlisting}[firstnumber=42]
apply_pRedo$_\txid$ $\eqdef$ 
   foreach x $\in$ pRedo$_\txid$.allocs: (*\label{AR1}*)
      mem[x].metadata := true     (*\label{AR2}*)
      flush mem[x].metadata    (*\label{AR3}*)
   if $\neg$pRedo$_\txid$.undoValid  then   (*\label{AR4}*)
      undoValid$_\txid$ := false         (*\label{AR5}*)
      flush undoValid$_\txid$         (*\label{AR6}*)

persist_writes$_\txid$ $\eqdef$
   foreach x $\in$ dom(undo$_\txid$):(*\label{PW1}*) flush x              (*\label{PW2}*)

roll_back$_\txid$ $\eqdef$
   foreach (x $\mapsto$ v) $\in$ undo$_\txid$: (*\label{abort1}*)
      mem[x].val := v                  (*\label{abort2}*)
   persist_writes$_\txid$                      (*\label{abort3}*)
   
PAbort$_\txid$ $\eqdef$       
   roll_back$_\txid$                           (*\label{abort1}*)
   undoValid$_\txid$ := false                 (*\label{abort2}*)
   flush undoValid$_\txid$                 (*\label{abort3}*)
   foreach x $\in$ tRedo$_\txid$.allocs:      (*\label{abort4}*)
      freeList.add(x)                 (*\label{abort5}*)
  
PRecovery$_\txid$ $\eqdef$
   if calc_checksum(pRedo$_\txid$)
      = pRedo$_\txid$.checksum             (*\label{rec4}*)
   then apply_pRedo$_\txid$                        (*\label{rec5}*)
   if undoValid$_\txid$ then                     (*\label{rec6}*)
      roll_back$_\txid$                           (*\label{rec7}*)
   foreach x $\in$ dom(mem):              (*\label{rec1}*)
      if $\neg$mem[x].metadata then      (*\label{rec2}*) 
         freeList.add(x)                 (*\label{rec3}*)\end{lstlisting}
  \end{minipage}

\caption{\PMDK global variables and pseudo-code} 
  \label{fig:pmdk} 
  \label{fig:pmdk-vars}
\end{figure}


\subsection{\PMDK Implementation}
\label{sec:pmdk-implementation}





We present the pseudo-code of our \PMDKTX abstraction in \cref{fig:pmdk-vars}.
We model all features of \PMDKTX (including its redo and and undo logs as well as its recovery mechanism in case of a crash) except memory deallocation within a \PMDKTX transaction.
%
%
We use \code{mem} to model the memory, mapping each location (in \code{loc}) to a value-metadata pair. We
model a value (in \code{val}) as an integers, and \code{metadata} as a boolean indicating whether the location is allocated. 
As we see below, the list of free (unallocated) locations, \code{freeList}, is calculated during recovery using \code{metadata}. 

Each \PMDK transaction maintains redo logs and an undo log. The redo
logs record the locations allocated by the transaction
so that if a crash occurs while committing, the allocated
locations can be reallocated, allowing the transaction to commit upon
recovery. Specifically, \PMDKTX uses two {\em distinct} redo logs:
\code{tRedo} and \code{pRedo}. 
Both are associated with fields \code{undoValid} (which is unset when
the log is invalidated), \code{checksum} (used
to indicate whether the log is valid), and \code{allocs}
(which contains the set of locations allocated by the
transaction). Note that \PMDKTX explicitly sets and unsets
\code{undoValid}, whereas \code{checksum} is calculated (\eg at
\lineno{commit3}) and may be invalidated by crashes corrupting a
partially completed write.  The undo log records the original (overwritten) value of
each location written to by the transaction, and is
consulted if the transaction is to be rolled back. We model it as a
map from locations to values (of type \code{int}). A separate variable
\code{undoValid} (distinct from \code{undoValid} in
\code{tRedo} and \code{pRedo}) is used to determine whether this undo
log is valid.


Each component in \cref{fig:pmdk-vars} have both a
volatile and persistent copy, although some components, \eg
\code{tRedo} and \code{freeList}, are {\em transient},
\ie their persistent versions are never
used. Likewise, the persistent redo log, \code{pRedo}, is only used in
a persistent fashion and its volatile copy is never
used. 

We now discuss the operations of \cref{fig:pmdk} in detail. We
assume the operations are executed by a transaction with id
$\txid$. This id is not useful in the sequential
setting in which \PMDKTX is used; however, in our concurrent extension (\cref{sec:making-pmdk-trans}) the transaction id is critical.  Note that the \PMDKTX implementation uses
instructions such as {\em optimised flush}, \code{sfence}, and {\em
  non-temporal stores} \cite{DBLP:journals/pacmpl/RaadMV22} to
further optimise the execution. However, we do not model these
here since our aim is to describe the main aspects of the \PMDKTX design
at a higher level of abstraction. 
%
%
\begin{description}[style=unboxed,leftmargin=0cm,itemsep=2pt]
\item 
[\bf\em PBegin.] The begin operation simply sets all local variables
to their initial values.

\item [\bf\em PAlloc.] Allocation chooses and removes a free
location, say {\tt x}, from the free
list, 
adds {\tt x} to the transient redo log (\lineno{alloc3}) and returns
{\tt x}. Removing {\tt x} from \code{freeList} ensures it is not allocated twice, while the transient redo log is used together with the persistent redo log to ensure allocated
locations are properly reallocated upon a system crash. Note that this
is an abstraction and the real allocator is much more complex (see
\cite[p.~323]{Scargall2020} and \cite{PMDK-Alloc} for details).

\quad 
When the transaction commits, the transient redo log
is copied to the persistent one (\lineno{commit4}), and the
effect of the persistent log is applied at \lineno{commit6} via
\code{apply\_pRedo}. (Note that \code{apply\_pRedo} is also
called by \code{PRecovery} on \lineno{rec5}.)  The behaviour of this
call depends on how much of the in-flight transaction was executed
before the crash leading to the recovery.  If a crash occurred after the
transaction executed (\lineno{commit4}) and the corresponding
write persisted (either due to a system flush or the
execution of \lineno{commit5}), then executing
\code{apply\_pRedo} via \code{PRecovery} has the same effect as the
executing \lineno{commit6}, \ie the effect of the redo log will be
applied. This (persistently) sets the \code{metadata} field of each
location in the redo log to indicate that it is allocated
(\linenos{AR1}{AR3}), and then invalidates the undo log
(\linenos{AR4}{AR6}) so that the transaction is not rolled
back. 

\item [\bf\em PRead.]
  A read from \texttt{x} simply returns its in-memory value (\lineno{read1}). 
  Note that location \texttt{x} may not be allocated; \PMDKTX delegates the responsibility of checking whether it is allocated to the client.

\item [\bf\em PWrite.]A write to {\tt x} first checks (\lineno{write2}) if the
current transaction has already written to {\tt x} (via a previously
executed \code{PWrite}). If not, it logs the current value by reading
the in-memory value of {\tt x} (\lineno{write3}) and records it in the undo log (\lineno{write4}). 
The updated undo log is then made persistent (\lineno{write5}). Once the
current value of {\tt x} is backed up in the undo log (either by
the current write or by the previous write to {\tt x}), the value of
{\tt x} in memory is updated to the new value {\tt v}
(\lineno{write6}). As with the read, location \texttt{x} may
not have been allocated; \PMDKTX delegates this check to the client.

\item [\bf\em PCommit.] The main idea behind the commit operation is to
ensure all writes are persisted, and that the persistent redo and undo
logs are cleared in the correct order, as follows. 
\begin{enumerate*}[label=\bfseries(\arabic*)]
	\item On \lineno{commit1} all
writes written by the transaction are persisted. 
\item Next, the transient
redo log is invalidated (\lineno{commit2}) and the checksum for the
log is calculated (\lineno{commit3}). This updated transient log is
then set to be the persistent redo log (\lineno{commit4}), which is
then made persistent (\lineno{commit5}). Note that after executing \lineno{commit5}, we can be assured that
the transaction has committed; if a crash occurs after this point, the
recovery will {\em redo} and persist the allocation and the undo log
will be cleared.
	\item The operation then calls \code{apply\_pRedo} at
\lineno{commit6}, which makes the allocation persistent and clears the
undo log. 
	\item Finally, at \lineno{commit7}, the \code{pRedo} checksum is
invalidated since \code{apply\_pRedo} has already been executed. If a
crash occurs after \lineno{commit7} has been executed, then the
recovery checks at \lineno{rec4} and \lineno{rec6} will fail, \ie
recovery will 
calculate the free list.
\end{enumerate*}

\item [\bf\em PAbort.] A \PMDK transaction is aborted by a \code{PRead}/\code{PWrite} that attempts to access (read/write) an unallocated
location. When a transaction is aborted, all of its observable effects
must be rolled back. First, the memory effects are rolled back
(\lineno{abort1}), then the undo log is invalidated (\lineno{abort2})
and made persistent (\lineno{abort3}), preventing undo from being
replayed in case a crash occurs. Finally, all of the locations
allocated by the executing transaction are freed
(\linenos{abort4}{abort5}).

Note that if a crash occurs during an abort, the effect
of the abort will be replayed. \code{PRecovery} reconstructs the free
list at \linenos{rec1}{rec3}, which effectively replays the loop at
\linenos{abort4}{abort5} of \code{PAbort}. Additionally, if a crash
occurs before the write at \lineno{abort2} has persisted, then the
effect of undoing the operation will be explicitly replayed by the
roll-back executed by \code{PRecovery} since \code{undoValid}
holds. If the crash occurs after the write at \lineno{abort2} has
persisted, then 
no roll-back is necessary.
\item [\bf\em PRecovery.] The recovery operation is executed immediately
after a crash, and before any other operation is executed. The
recovery proceeds in three phases:
\begin{enumerate*}[label=\bfseries(\arabic*)]
\item The checksum of the persistent redo log is recalculated
  (\lineno{rec4}) and if it matches the stored checksum (
  \code{pRedo.checksum}) the \code{apply\_pRedo} operation is
  executed. As discussed, \code{apply\_pRedo} sets and persists
  the metadata of each location in the redo log, and then invalidates the
  undo log.
\item The transaction is rolled back if \code{apply\_pRedo}
  in step 1 fails; otherwise, no roll-back
  is performed.
\item The free list is
reconstructed by inserting each location whose metadata is set to
false into \code{freeList} (\linenos{rec1}{rec3}).%
\end{enumerate*}
\end{description}


\paragraph{Correctness} As discussed in \cref{sec:pmdk-interface},
\PMDKTX is designed to be failure-atomic. This means that
correctness criteria such as
opacity~\cite{DBLP:series/synthesis/2010Guerraoui,DBLP:conf/wdag/AttiyaGHR14}
and TMS1/TMS2~\cite{DBLP:journals/fac/DohertyGLM13} (restricted to
sequential transactions) are inadequate since they do not accommodate
crashes and recovery. This points to conditions such as durable
opacity~\cite{DBLP:conf/forte/BilaDDDSW20}, which extends opacity with
a persistency model. 
However, durable opacity (restricted to sequential transactions)
is also insufficient since it does not define correctness of
allocations and assumes totally ordered histories. In
\cref{sec:volatile_framework} we develop a generalisation of durable
opacity, called {\em dynamic durable opacity} (\DDO) that addresses
both of these issues. As with durable opacity, \DDO defines correctness
for concurrent transactions. We develop concurrent extensions of \PMDK
transactions in \cref{sec:making-pmdk-trans}, which we show to be
correct against (\ie refinements of) \DDO.

\paragraph{Thread Safety}
As discussed, \PMDK
transactions are not thread-safe; \eg concurrent calls to
\code{PRead} and \code{PWrite} on the same location create a data
race causing \code{PRead} to return an undefined value (see the example in \cref{sec:intro}). We discuss techniques for mitigating against such races in
\cref{sec:making-pmdk-trans}.

Nevertheless, some \PMDK transactional operations are naturally
thread-safe. In particular, \code{PAlloc} is designed to be thread-safe via an built-in {\em arena} mechanism: a memory pool split into
disjoint arenas with each thread allocating from its own
arena. Moreover, each thread uses locks for each arena to publish
allocated memory to the shared pool~\cite{PMDK-Alloc}. 

\section{Making PMDK Transactions Concurrent}
\label{sec:making-pmdk-trans}
We develop two algorithms that combine two existing STM systems with
\PMDK, leveraging the concurrency control mechanisms provided by
STM and the failure atomicity guarantees of \PMDK in a
modular manner. The first algorithm (\cref{sec:pmdkt}) is based on
TML~\cite{DBLP:conf/europar/DalessandroDSSS10}, which uses {\em
  pessimistic concurrency control} via an eager write-back
scheme. Writing transactions effectively take a lock and perform
the writes in place.  The second algorithm (\cref{sec:pmdkn}) is based
on \NOREC~\cite{DBLP:conf/ppopp/DalessandroSS10}, which utilises {\em
  optimistic concurrency control} via a lazy write-back scheme. In
particular, transactional writes are {\em collected} in a local write
set and {\em written back} when the transaction commits.

It turns out that \PMDK can be incorporated within both
algorithms straightforwardly. This is a strength of
our approach and points towards a generic technique for extending existing STM systems with
failure atomicity. 
Given the
challenges of persistent allocation, we reuse \PMDK's allocation
mechanisms to provide an explicit allocation mechanism in both our
extensions~\cite{Scargall2020}.




\begin{figure}[t]
  \centering
  \begin{minipage}[t]{0.45\columnwidth}
{\tt {\bf Init}. glb = 0}
\begin{lstlisting}[escapeinside={(*}{*)}]
TxBegin$_\txid$ $\eqdef$                         
   do loc$_\txid$ := glb         (*\label{TMLB1}*)             
   until even(loc$_\txid$)       (*\label{TMLB2}*)
   (*\color{red}PBegin$_\txid$*) (*\label{TMLB3}*)

TxAlloc$_\txid$ $\eqdef$
   return (*\color{red} PAlloc$_\txid$*)         (*\label{TMLA1}*)

TxWrite$_\txid$(x, v) $\eqdef$
   if even(loc$_\txid$) then                     (*\label{TMLW1}*)
      if $\neg$cas(glb, loc$_\txid$, loc$_\txid$+1)           (*\label{TMLW2}*)
      then (*\color{red}\tt PAbort$_\txid$*); return abort  (*\label{TMLW3}*) 
      else loc$_\txid$++                         (*\label{TMLW4}*)
   (*\color{red}PWrite$_\txid$(x,v)*)       (*\label{TMLW5}*)\end{lstlisting}             
\end{minipage}
\qquad 
\begin{minipage}[t]{0.45\columnwidth}
\begin{lstlisting}[firstnumber=15]
TxRead$_\txid$(x) $\eqdef$
  v$_\txid$ := (*\color{red}\tt PRead$_\txid$(x)*)  (*\label{TMLR1}*)                
  if even(loc$_\txid$) then                (*\label{TMLR2}*)
     if glb = loc$_\txid$ then             (*\label{TMLR3}*)
        return v$_\txid$                 (*\label{TMLR4}*)
     else (*\color{red}\tt PAbort$_\txid$*); return abort                     (*\label{NORECV6}*)  (*\label{TMLR5}*) 
  else return v$_\txid$                  (*\label{TMLR7}*)

TxCommit$_\txid$  $\eqdef$             
   (*\color{red}PCommit$_\txid$*)      (*\label{TMLC1}*)
   if odd(loc$_\txid$) then            (*\label{TMLC2}*)
      glb := loc$_\txid$+1             (*\label{TMLC3}*)

Recovery $\eqdef$
   foreach $\txid \in$ $\TXIDs$:
      (*{\color{red}\tt PRecovery$_\txid$}*)  (*\label{TMLRec1}*)
   glb := 0                     (*\label{TMLRec2}*) \end{lstlisting}
\end{minipage}
\caption{Pseudo-code for \PMDKT with our additions made w.r.t. TML {\color{red}highlighted red} 
}
\label{fig:pmdk-tml}
\end{figure}

\paragraph{\PMDKT}
\label{sec:pmdkt}
We present the pseudo-code for \PMDKT (combining TML and \PMDKTX) in
\cref{fig:pmdk-tml}, where we highlight the calls to \PMDKTX
operations. These calls are the only changes we have made to the TML
algorithm. TML is based on a single global counter, \code{glb}, whose
value is read and stored within a local variable \code{loc$_\txid$}
when transaction $\txid$ begins (\code{TxBegin}). There is an
in-flight {\em writing} transaction iff \code{glb} is odd. TML is
designed for read-heavy workloads, and thus allows multiple concurrent
read-only transactions. A writing transaction causes all other
concurrent transactions 
to abort.

\PMDKT proposes a modular combination of \PMDK with the \TML algorithm
by nesting a \PMDK transaction inside a \TML transaction; \ie each
transaction additionally starts a \PMDK transaction. All reads and
writes to memory are replaced by \PMDKTX read and write
operations. Moreover, when a transaction aborts or commits, the
operation calls a \PMDKTX abort or commit, respectively. Finally, \PMDKT
includes allocation and recovery operations, which call \PMDKTX
allocation and recovery, respectively. The recovery operation
additionally sets \code{glb} to 0.

A read-only transaction $\txid$ may call \code{PRead$_\txid$} at
\lineno{TMLR1} when another transaction $\txidb$ is executing
\code{PWrite$_{\txidb}$} at \lineno{TMLW5} on the same location. Since
\PMDKTX does not guarantee thread safety for these calls, the value
returned by \code{PRead$_\txid$} should not be passed back to the
client. This is indeed what occurs. First, note that if transaction
$\txid$ is read-only, then \code{loc$_\txid$} is even. Moreover, a
read-only transaction only returns the value returned by
\code{PRead$_\txid$} (\lineno{TMLR4}) if no other transaction has
acquired the lock since $\txid$ executed \code{TxBegin$_\txid$}. In
the scenario described above, $\txidb$ must have incremented
\code{glb} by successfully executing the CAS at \lineno{TMLW2} as part
of the first write operation executed by $\txidb$, changing the
value of \code{glb}. This means that $\txid$ would abort since the test at \lineno{TMLR3} would fail.

\paragraph{\PMDKN}\label{sec:pmdkn}
We present \PMDKN (combining \NOREC and \PMDK) in
\cref{fig:pmdk-norec}, where we highlight the calls to \PMDKTX. These
calls are the only changes we have made to the \NOREC algorithm. As with
TML, \NOREC is based on a single global counter, \code{glb}, whose
value is read and stored within a transaction-local variable
\code{loc} when a transaction begins (\code{TxBegin}). There is an
in-flight writing transaction iff \code{glb} is odd. Unlike
TML, \NOREC performs lazy write-back, and hence utilises
transaction-local read and write sets. A transaction only performs the
write-back at commit time once it ``acquires'' the \code{glb} lock. 
Prior to write-back and read response, it ensures that the
read sets are consistent using a per-location validate operation. We
eschew detailed discussion of the \NOREC synchronisation mechanisms
and refer the interested reader to the original
paper~\cite{DBLP:conf/ppopp/DalessandroSS10}.

\begin{figure}[t]
  \centering
\begin{minipage}[t]{0.44\columnwidth}
{\tt {\bf Init:}  glb = 0}
\begin{lstlisting}[escapeinside={(*}{*)}]
TxBegin$_\txid$ $\eqdef$                       
   do loc$_\txid$ := glb
   until even(loc$_\txid$)
   (*\color{red}PBegin$_\txid$*)

TxAlloc$_\txid$ $\eqdef$
   return (*\color{red}\tt PAlloc$_\txid$*)
   
TxRead$_\txid$(x) $\eqdef$   
   if x $\in$ dom(wrSet$_\txid$) then                  (*\label{NORECR1}*)
      return wrSet$_\txid$(x)                          (*\label{NORECR2}*)
   v$_\txid$ := (*\color{red}\tt PRead$_\txid$(x)*)     (*\label{NORECR3}*)
   while loc$_\txid$ $\neq$ glb                        (*\label{NORECR4}*)
      loc$_\txid$ := Validate                          (*\label{NORECR5}*)
      v$_\txid$ := (*\color{red}\tt PRead$_\txid$(x)*) (*\label{NORECR6}*)
   rdSet$_\txid$ := rdSet$_\txid$ $\cup$ {x $\mapsto$ v$_\txid$} (*\label{NORECR7}*)
   return v$_\txid$                                    (*\label{NORECR8}*)

Recovery $\eqdef$
   foreach $\txid \in$ $\TXIDs$:
      (*{\color{red}\tt PRecovery$_\txid$}*)
   glb := 0 \end{lstlisting} 
\end{minipage}
\qquad 
\begin{minipage}[t]{0.48\columnwidth}
\begin{lstlisting}[firstnumber=23]
TxWrite$_\txid$(x,v) $\eqdef$  
   wrSet$_\txid$ := wrSet$_\txid$ $\cup$ {x $\mapsto$ v}    
   
Validate$_\txid$ $\eqdef$
   while true                                  (*\label{NORECV1}*)
      time$_\txid$ := glb                              (*\label{NORECV2}*)
      if odd(time$_\txid$) then goto (*\ref{NORECV2}*) (*\label{NORECV3}*)
      foreach x $\mapsto$ v $\in$ rdSet$_\txid$:(*\label{NORECV4}*)
         if (*\color{red}\tt PRead$_\txid$(x)*) $\neq$ v                   (*\label{NORECV5}*)
         then (*\color{red}\tt PAbort$_\txid$*); return abort                     (*\label{NORECV6}*)
      if time$_\txid$ = glb                            (*\label{NORECV7}*)
      then return time$_\txid$                         (*\label{NORECV8}*)
  
TxCommit$_\txid$ $\eqdef$   
    if wrSet$_\txid$.isEmpty                            (*\label{NORECC1}*)
       then (*\color{red}\tt PCommit$_\txid$*)          (*\label{NORECC2}*)
            return                                     (*\label{NORECC3}*)
    while $\neg$cas(glb, loc$_\txid$, loc$_\txid$ + 1)  (*\label{NORECC4}*)  
       loc$_\txid$ := Validate$_\txid$                  (*\label{NORECC5}*)
    foreach x $\mapsto$ v $\in$ wrSet$_\txid$:          (*\label{NORECC6}*)
       (*\color{red}\tt PWrite$_\txid$(x, v)*)          (*\label{NORECC7}*)
    (*\color{red}\tt PCommit$_\txid$*)                  (*\label{NORECC8}*)
    glb := loc$_\txid$ + 2                              (*\label{NORECC9}*)
    return \end{lstlisting} 
\end{minipage}
\caption{Pseudo-code for \PMDKN, with our additions made w.r.t. \NOREC\ {\color{red} highlighted red}}
\label{fig:pmdk-norec}
\end{figure}

The transformation from \PMDKTX to \PMDKN is similar to \PMDKT. We
ensure that a \PMDK transaction is started when a \PMDKN transaction
begins, and that this \PMDK transaction is either aborted or committed
before the \PMDKN transaction completes. We introduce \code{TxAlloc}
and \code{Recovery} operations that are identical to \PMDKT, and
replace all calls to read and write from memory by \code{PRead} and
\code{PWrite} operations, respectively.

As with \PMDKT, a \code{PRead} executed by a transaction (at
\lineno{NORECR3}, \lineno{NORECR6} or \lineno{NORECV5}) may race with
a \code{PWrite} (at \lineno{NORECC7}) executed by another
transaction. However, since \code{PWrite} operations are only executed
after a transaction takes the \code{glb} lock (at \lineno{NORECC4}),
any transaction with a racy \code{PRead} is revalidated. If validation
fails, the associated transaction is aborted.







\section{Declarative Correctness Criteria for Transactional Memory}
\label{sec:volatile_framework}

We present a declarative correctness criteria for TM
implementations. Unlike prior definitions such as (durable) opacity, TMS1/2 \etc that are defined in terms of histories
of invocations and responses, we define dynamic durable opacity (\DDO)
in terms of execution graphs, as is standard model for weak
memory setting. Our models are inspired by prior work on declarative
specifications for transactional memory, which focussed on specifying
relaxed
transactions~\cite{hw-transactions-dongol,DBLP:conf/pldi/ChongSW18}. However,
these prior works do not describe crashes or allocation.

%
\smallskip\emph{\textbf{Executions and Events.}}
The traces of memory accesses generated by a program are commonly represented as a set of \emph{executions}, where each execution $G$ is a graph comprising: 
\begin{enumerate*}
	\item a set of events (graph nodes); and 
	\item a number of relations on events (graph edges). 
\end{enumerate*}	
Each event $e$ corresponds to the execution of either a transactional event (\eg marking the beginning of a transaction) or a memory access (read/write) within a transaction.


\begin{definition}[Events]
\label{def:events}
An \emph{event} is a tuple $a=\tup{n,\tid, \txid, \lb}$, where $n \in \Nats$ is an event identifier, 
$\tid \in\TIDs$ is a \emph{thread identifier},
$\txid \in \TXIDs$ is a \emph{transaction identifier}
and $\lb \in \Labs$ is an event \emph{label}.

A label may be $\tbeginL$ to mark the beginning of a transaction; 
$\abortL$ to denote a transactional abort; 
$\allocL x$ to denote a memory allocation yielding $x$ initialised with value $0$; 
$\readL x v$ to denote reading value $v$ from location $x$; 
$\writeL x v$ to denote writing $v$ to $x$; 
$\commitL$ to mark the beginning of the transactional commit process; or
$\succL$ to denote a successful commit.
\end{definition}

%
The functions $\lTHRD$, $\lTX$ and $\lLAB$ respectively project the thread identifier, transaction identifier and the label of an event. 
The functions $\lLOC$, $\lVALR$ and $\lVALW$ respectively project the location, the read value and the written value of a label, where applicable, and are lifted to events by defining \eg $\lLOC(a) = \lLOC(\lLAB(a))$. 

\smallskip\emph{\textbf{Notation.}}
Given a relation $\makerel r$ and a set $A$, we write $\refC{\makerel{r}}$, $\transC{\makerel{r}}$ and $\reftransC{\makerel{r}}$ for the  reflexive, transitive and reflexive-transitive closures of $\makerel{r}$, respectively. 
We write $\inv{\makerel r}$ for the inverse of $\makerel r$;
$\coerce{\makerel r}{A}$ for $\makerel r \cap (A \times A)$; 
$[A]$ for the identity relation on $A$, \ie $\set{(a, a) \mid a \!\in\! A}$;
$\irr{r}$ for $\nexsts{a}\! (a, a) \!\in\! r$;
and $\acyc{\makerel r}$ for $\irr{\transC{\makerel r}}$.
We write $\makerel r_1; \makerel r_2$ for the relational composition of $r_1$ and $r_2$,
\ie $\{(a, b) \mid \exsts{c}\! (a, c) \!\in\! \makerel r_1 \land (c, b) \!\in\! \makerel r_2\}$.
%
When $A$ is a set of events, we write
$A_{\x}$ for $\set{a \!\in\! A \mid \loc{a} {=} \x}$,
and $\makerel r_{\x}$ for $\coerce{\makerel r}{A_\x}$.
Analogously, we write $A_{\txid}$ for $\set{a \!\in\! A \mid \tx{a} {=} \txid}$.
The \emph{`same-transaction' relation}, $\st \suq \Events \times \Events$, 
is the equivalence relation
$
		\st \eqdef 
		\setcomp{
			(a, b) \in \Events \times \Events
		}{
			\tx{a} {=} \tx{b}
		}	
$.

\begin{definition}[Executions]
\label{def:volatile_executions}
An \emph{execution}, $G \in \Execs$, is a tuple
$(\Events, \po, \clo, \rf, \mo)$, where:
\begin{itemize}
	\item $\Events$ denotes a set of \emph{events}.
          The set of \emph{read} events in $\Events$ is:
          $$\Reads \eqdef \setcomp{e \in \Events}{\labl{e} {=} \readL - -}.$$
	The sets of \emph{memory allocations} ($\Allocs$), 
	\emph{writes} ($\Writes$),
	\emph{aborts} ($\Aborts$),
	\emph{transactional begins} ($\Begins$),
	\emph{transactional commits} ($\Commits$) and 
	\emph{commit successes} ($\Succs$) 
	are defined analogously.

%

%
%
	\item $\po \suq \Events \times \Events$ denotes the \emph{`program-order' relation},  defined as a disjoint
	union of strict total orders, each ordering the events of one thread.

%
%
	\item $\clo \suq \Events \times \Events$ denotes the \emph{`client-order' relation},
	which is a strict partial order between transactions ($\st ; \clo ; \st \suq \clo \setminus \st$)
	that extends the program order between transactions ($\po\setminus \st \suq \clo$).
	Intuitively, this is the order the client may observe on transactions by using other means than
	the transaction system (\eg writing to a shared flag after some transaction successfully commits
	and reading from that flag before starting another transaction on a different thread).
	\item $\rf \suq (\Allocs \cup \Writes) \times \Reads$ denotes the \emph{`reads-from' relation} between events of the same location with matching  values; \ie $(a, b) \in \rf \Rightarrow \loc a {=} \loc b \land \wval a {=} \rval b$.
	Moreover, $\rf$ is total and functional on its range, 
	\ie every read is related to exactly one write.
	\item $\mo \suq \Events \times \Events$ is the \emph{`modification-order'}, defined 
	 as the disjoint union of relations $\{\mo_\x\}_{\x \in \Locs}$, such that each $\mo_\x$ is a strict total order on $\Allocs_\x \cup \Writes_\x$.
%
%
%
%
\end{itemize}
\end{definition}
Given a relation $\makerel r \subseteq \Events \times \Events$,
we write $\tlift{\makerel r}$ for lifting $\makerel r$ to transaction classes:
$
	\tlift{\makerel r} \eqdef \st ; (\makerel r \setminus \st); \st
$.
For instance, when $(w, r) \in \rf$, $w$ is a transaction $\txid_1$ event and $r$ is a transaction $\txid_2$ event, then all events in $\txid_1$ are $\rft$-related to all events in $\txid_2$. 
We write $\tin{\makerel r}$ to restrict $\makerel r$ to its \emph{intra-transactional} edges (within a transaction): $\tin{\makerel r} \eqdef \makerel r \cap \st$;
and write $\tex{\makerel r}$ to restrict $\makerel r$ to its \emph{extra-transactional} edges (outside a transaction):  $\tex{\makerel r} \eqdef \makerel r \setminus \st$. 
Analogously, we write $\internal{\makerel r}$ to restrict $\makerel r$ to its \emph{intra-thread edges}: 
$\internal{\makerel r} \eqdef \setcomp{(a, b) \in \makerel r}{\thrd a {=} \thrd b}$; 
and write $\external{\makerel r}$ to restrict $\makerel r$ to its \emph{extra-thread edges}: 
$\external{\makerel r} \!\eqdef\! \makerel r \setminus \internal{\makerel r}$.

In the context of an execution $G$ (we use the ``$G.$'' prefix to make this explicit),
the \emph{reads-before} relation is $\rb \eqdef (\inv{\rf}; \mo)$.
Intuitively, $\rb$ relates a read $r$ to all writes/allocations $w$ that are $\mo$-after the write $r$ reads from;
\ie when $(w', r) \in \rf$ and $(w', w) \in \mo$, then $(r, w) \in \rb$. 
%

Lastly, we write $\CSet$ for the events of \emph{committing} transactions, \ie those that have reached the commit stage: 
$\CSet \eqdef \dom(\st; [\Commits])$.
We define the sets of \emph{aborted} events, $\ASet$, and
\emph{(commit)-successful} events, $\SSet$, analogously.
We define the set of \emph{commit-pending} events as $\CPSet \eqdef \CSet \setminus  (\ASet \cup \SSet)$, and the set of \emph{pending} events as $\PSet \eqdef \Events \setminus (\CPSet \cup \ASet \cup  \SSet)$.

Given an execution $G {=} (\Events, \po, \clo, \rf, \mo)$, we write $\coerce G A$ for 
$$(\Events \cap A,  \coerce \po {\Events \cap A}, \coerce \clo {\Events \cap A}, 
\coerce \rf {\Events \cap A}, \coerce \mo {\Events \cap A}).$$

We further impose certain ``well-formedness'' conditions on executions,
used to delimit transactions and restrict allocations.
For example, we require that events of the same transaction are by the same thread 
and the each $\txid$ contains exactly one begin event.
The full set of well-formedness condition is presented in \cref{app:wf}.
In particular, these conditions ensure that in the context of a well-formed execution $G$
we have 
\begin{enumerate*}
	\item $G.\SSet \suq G.\CSet$; 
	\item each $\txid$ contains at most a single abort or success ($\size{G.\Events_{\txid} \cap (\Aborts \cup \Succs)} \leq 1$) and thus 
$G.(\SSet \cap \ASet) {=} \emptyset$; and 
	\item $G.\Events = G.(\PSet \uplus \ASet \uplus \CPSet \uplus \SSet)$, \ie the sets $G.\PSet$, $G.\ASet$, $G.\CPSet$ and $G.\SSet$ are pair-wise disjoint. 
\end{enumerate*}

\smallskip\emph{\textbf{Execution Consistency.}}
The definition of (well-formed) executions above puts very few constraints on the $\rf$ and $\mo$ relations. Such restrictions and thus the permitted behaviours of a transactional program are determined by defining the set of \emph{consistent} executions, defined separately for each transactional consistency model. 
The existing literature includes several definitions of well-known consistency models, including \emph{serialisability} (\SER) \cite{ser}, \emph{snapshot isolation} (SI) \cite{SI-Cerone,SI-VMCAI} and \emph{parallel snapshot isolation} (PSI) \cite{PSI-Cerone,PSI-ESOP}.



\smallskip\emph{\textbf{Serialisability (\SER).}}
The \emph{serialisability} (\SER) consistency model \cite{ser} is one of the most well-known transactional consistency models, as it provides strong guarantees that are intuitive to understand and reason about. 
Specifically, under \SER, all concurrent transactions must appear to execute atomically one after another in a \emph{total sequential order}.
The existing declarative definitions of \SER \cite{SI-Cerone,PSI-Cerone,DBLP:journals/pacmpl/RaadWV19} are somewhat restrictive in that they only account for fully committed (complete) transactions, \ie they do not support pending or aborted transactions. 
Under the assumption that all transactions are complete, an execution $(\Events, \po, \clo, \rf, \mo)$ is deemed to be serialisable (\ie \SER-consistent) iff: 
\begin{itemize}
	\item $\rfI \cup \moI \cup \rbI \subseteq \po$%
	\labelAxiom{ser-int}{ax:comp_ser_int}
	\item $\clo \cup \rft \cup \mot \cup \rbt$ is acyclic.%
	\labelAxiom{ser-ext}{ax:comp_ser_ext}
\end{itemize}
The \ref{ax:comp_ser_int} axiom enforces \emph{intra-transactional} consistency, ensuring that \eg a transaction observes its own writes by requiring $\rfI \suq \po$ (\ie intra-transactional reads respect the program order).
Analogously, the \ref{ax:comp_ser_ext} axiom guarantees \emph{extra-transactional} consistency, ensuring the existence of a total sequential order in which all concurrent transactions appear to execute atomically one after another. This total order is obtained by an arbitrary extension
of the (partial) `happens-before’ relation which captures synchronisation resulting from
transactional orderings imposed by client order ($\clo$) or \emph{conflict} between transactions ($\rft \cup \mot \cup \rbt$).
Two transactions are conflicted if they both access (read or write) the same location $x$, and at least one of these accesses is a write. 
As such, the inclusion of $\rft \cup \mot \cup \rbt$ enforces conflict-freedom of serialisable transactions. For instance, if transactions $\txid_1$ and $\txid_2$ both write to $x$ via events $w_1$ and $w_2$ such that
$(w_1,w_2) \in \mo$, then $\txid_1$ must commit before $\txid_2$, and thus the entire effect of $\txid_1$ must be visible to $\txid_2$.

\smallskip\emph{\textbf{Opacity.}}
In our general setting we do not stipulate that all transactions commit successfully and allow for both aborted and pending transactions.
As such, we opt for the stronger notion of transactional correctness known as \emph{opacity}.
In what follows we describe our notion of opacity over executions (formalised in \cref{def:opacity}), and later relate it to the existing notion of opacity over histories \cite{DBLP:series/synthesis/2010Guerraoui} and prove that our characterisation of opacity is \emph{equivalent} to that of the existing one (see \cref{thm:orig-opacity}). 

Intuitively, an execution is opaque iff
\begin{enumerate*}
	\item it is serialisable (upto several relaxations for aborted/pending transactions); and additionally
	\item each external read (be it in a committed, commit-pending, pending or aborted transaction) must read from a write in a \emph{visible} transaction, where a transaction is visible if it has either successfully committed or reached the commit stage (\ie is commit-pending) and is read from by another transaction. 
	\label{item:opacity_rf}
\end{enumerate*}
Specifically, condition \ref{item:opacity_rf} is enforced by \ref{ax:vis_rf} in \cref{def:opacity} below, ensuring that all external reads are from visible transactions (in $\VSet$). 
To see this, consider the execution in \cref{subfig:vis_rf} comprising two transactions, $\txid$ and $\txid'$, where $\txid$ writes 1 to $x$ ($w{:} \writeL x 1$, where $w$ denotes an event with label $\writeL x 1$), and either aborts (with $\abortL$) or remains pending (without any other events).
Similarly, $\txid'$ reads 1 from $x$, and either remains pending (without any other events), aborts (with $\abortL$), commits and remains commit-pending (with $\commitL$) or successfully commits (with $\commitL \!\rightarrow\! \succL$). 
The solid $\rightarrow$ arrows in each transactions denote $\po$ edges. 
As $r$ in $\txid'$ reads from $w$ in $\txid$ which is aborted/pending, this execution violates \ref{ax:vis_rf} and is not opaque.

\begin{figure}[t]
\small
\begin{tabular}{|@{\hspace{2pt}} c @{\hspace{2pt}} | c @{\hspace{2pt}} | c  @{\hspace{2pt}} | @{}}
\hline
\!\begin{subfigure}[b]{0.26\textwidth}
\centering
\scalebox{0.85}{
\begin{tikzpicture}[yscale=0.85,xscale=0.77]\small
  \draw[draw=black,rounded corners,dotted,fill=blue!10] (-0.7,1.4) rectangle (1,-0.45);
  \node[draw=black,rounded corners,dotted,fill=blue!20] at (0.15,1.55) {$\txid$};
  \node (x1)  at (0.15,1) {$w{:}\,\writeL x 1$ };
  \node (a)  at (0.15,0.25) {$(\abortL)$ };
  \node (phantom)  at (0.35,-0.5) {}; 
  \draw[po] (x1) edge (a);
  \draw[draw=black,rounded corners,dotted,fill=yellow!10] (1.6,1.4) rectangle (3.6,-0.45);
  \node[draw=black,rounded corners,dotted,fill=yellow!20] at (2.6,1.55) {$\txid'$};
  \node (rx)  at (2.6,1) {$r {:}\,\readL x 1$ };
  \node (e)  at (2.6,0.25) {($\abortL$/$\commitL$/$\commitL \!\rightarrow\! \succL$)};
  \draw[po] (rx) edge (e);
  \draw[rf] (x1) edge node[below=2pt]{$\rf$} (rx);
\end{tikzpicture}}
\caption{\intabTC{Opaque: \xmark \\(not \ref{ax:vis_rf})}}\vspace{-3pt}
\label{subfig:vis_rf}
\end{subfigure}
&
\!\!\begin{subfigure}[b]{0.42\textwidth}
\centering
\!\scalebox{0.85}{\begin{tikzpicture}[yscale=0.85,xscale=0.77]\small
  \draw[draw=black,rounded corners,dotted,fill=blue!10] (-0.3,1.4) rectangle (1,-0.9);
  \node[draw=black,rounded corners,dotted,fill=blue!20] at (0.35,1.55) {$\txid_x$};
  \node (x1)  at (0.35,1) {$\writeL x 1$ };
  \node (c1)  at (0.35,0.25) {$\commitL$ };
  \node (s1)  at (0.35,-0.5) {$\succL$ };
  \draw[po] (x1) edge (c1);
  \draw[po] (c1) edge (s1);
  \draw[draw=black,rounded corners,dotted,fill=yellow!10] (1.7,1.4) rectangle (3,-0.9);
  \node[draw=black,rounded corners,dotted,fill=yellow!20] at (2.35,1.55) {$\txid$};
  \node (ry1)  at (2.35,1) {$\readL y 1$ };
  \node (wx2)  at (2.35,0.25) {$\writeL x 2$ };
  \node (a1)  at (2.35,-0.5) {($\abortL$/$\commitL$)};
  \draw[po] (ry1) edge (wx2);
  \draw[po] (wx2) edge (a1);
%
%
  \draw[draw=black,rounded corners,dotted,fill=red!10] (3.7,1.4) rectangle (5,-0.9);
  \node[draw=black,rounded corners,dotted,fill=red!20] at (4.35,1.55) {$\txid'$};
  \node (rx1) at (4.35,1) {$\readL x 1$};
  \node (wy2) at (4.35,0.25) {$\writeL y 2$ };
  \node (a2)  at (4.35,-0.5) {($\abortL$/$\commitL$)};
  \draw[po] (rx1) edge (wy2);
  \draw[po] (wy2) edge (a2);
%
%
%
  \draw[draw=black,rounded corners,dotted,fill=green!10] (5.7,1.4) rectangle (7,-0.9);
  \node[draw=black,rounded corners,dotted,fill=green!20] at (6.35,1.55) {$\txid_y$};
  \node (y1)  at (6.35, 1) {$\writeL y 1$ };
  \node (c2)  at (6.35,0.25) {$\commitL$ };
  \node (s2)  at (6.35,-0.5) {$\succL$ };
  \draw[po] (y1) edge (c2);
  \draw[po] (c2) edge (s2);
  \draw[rf,bend left=40] (x1) edge node[above]{$\rf$} (rx1);
  \draw[rf,bend right=40] (y1) edge node[above]{$\rf$} (ry1);
  \draw[fr] (ry1) edge node[above]{$\rb$} (wy2);
  \draw[fr] (rx1) edge node[below]{$\rb$} (wx2);
  \draw[mo] (x1) edge node[above]{$\mo$} (wx2);
  \draw[mo] (y1) edge node[above]{$\mo$} (wy2);
\end{tikzpicture}}\vspace{-5pt}
\caption{Opaque: \cmark} 
\vspace{-3pt}
\label{subfig:rb-ok}
\end{subfigure}\!\!
&
\!\!\begin{subfigure}[b]{0.32\textwidth}
\centering
\!\scalebox{0.85}{\begin{tikzpicture}[yscale=0.85,xscale=0.7]\small
  \draw[draw=black,rounded corners,dotted,fill=blue!10] (-0.3,1.4) rectangle (1,-0.9);
  \node[draw=black,rounded corners,dotted,fill=blue!20] at (0.35,1.55) {$\txid$};
  \node (x1)  at (0.35,1) {$\writeL x 1$ };
  \node (y1)  at (0.35,0.25) {$\writeL y 1$ };
  \node (c)  at (0.35,-0.5) {$\commitL$ };
  \draw[po] (x1) edge (y1);
  \draw[po] (y1) edge (c);
%
%
  \draw[draw=black,rounded corners,dotted,fill=yellow!10] (1.7,1.4) rectangle (3,-0.9);
  \node[draw=black,rounded corners,dotted,fill=yellow!20] at (2.35,1.55) {$\txid_r$};
  \node (rx)  at (2.35,1) {$\readL x 0$ };
  \node (ry)  at (2.35,0.25) {$\readL y 1$ };
  \node (a)  at (2.35,-0.5) {($\abortL$)};
  \draw[po] (rx) edge (ry);
  \draw[po] (ry) edge (a);
  \draw[fr] (rx) edge node[above]{$\rb$} (x1);
  \draw[rf] (y1) edge node[below]{$\rf$} (ry);
  \draw[draw=black,rounded corners,dotted,fill=red!10] (3.7,1.4) rectangle (5,-0.9);
  \node[draw=black,rounded corners,dotted,fill=red!20] at (4.35,1.55) {$\txid'$};
  \node (x0) at (4.35,1) {$\writeL x 0$};
  \node (c') at (4.35,0.25) {$\commitL$ };
  \node (s') at (4.35,-0.5) {$\succL$};
  \draw[po] (x0) edge (c');
  \draw[po] (c') edge (s');
  \draw[rf] (x0) edge node[above]{$\rf$} (rx);  
  \draw[mo,bend right=40] (x0) edge node[above]{$\mo$} (x1);  
\end{tikzpicture}}\vspace{-5pt}
\caption{Opaque: \xmark (not \ref{ax:ext})}\vspace{-3pt}
\label{subfig:rb}
\end{subfigure}\!\!
\\
\hline
\end{tabular}
\vspace{-5pt}
\caption{Several opaque (\cmark) and non-opaque (\xmark) executions, where ($\abortL$) \eg in $\txid$ of \subref{subfig:vis_rf} denotes that the same outcome holds whether $\txid$ has an abort event (\ie $\txid$ has aborted) or not (\ie $\txid$ is pending); similarly, ($\abortL$/$\commitL$/$\commitL \!\rightarrow\! \succL$) in \subref{subfig:vis_rf} denotes that the same outcome holds whether $\txid$ is pending, aborted, commit-pending or committed.}
\label{fig:opaque_executions}
\end{figure}


The remaining axioms in \cref{def:opacity} adapt the notion of serialisability to account for pending/aborted transactions. 
First, as with serialisability, opacity enforces intra-transactional consistency via the \ref{ax:int1} axiom (\cf\ref{ax:comp_ser_int}). 
%
However, the notion of extra-transactional consistency is too strong in the context of pending/aborted transactions. 
To see this, consider the execution in \cref{subfig:rb-ok} comprising four transactions, $\txid_x$ and $\txid_y$ which have successfully committed, and $\txid$ and $\txid'$ which are aborted/pending. 
As depicted, (the event labelled) $\readL x 1$ reads from $\writeL x 1$ which in turn is $\mo$-before $\writeL x 2$, then $\readL x 1$ is $\rb$-before $\writeL x 2$, and thus $\writeL y 2$ is $\rbt$-before $\writeL x 2$. 
Analogously, $\readL y 1$ is $\rb$-before $\writeL y 2$, and thus $\writeL x 2$ is $\rbt$-before $\writeL y 2$. 
As such, this execution contains an $\rbt$ cycle: $\writeL y 2 \relarrow{\rbt} \writeL x 2 \relarrow{\rbt} \writeL y 2$.
Intuitively, \ref{ax:comp_ser_ext} considers $\txid$ and $\txid'$ to be conflicted, and stipulates that the events in one be fully ordered before the other (\ie it precludes $\rbt$ cycles).
However, this notion of conflict should not extend to invisible transactions. 
Indeed, either one of $\txid_1$/$\txid_2$ could abort precisely because they are conflicted with one another.
As such, we only require that there be no $\rbt$ cycles into visible transactions.
That is, we weaken \ref{ax:comp_ser_ext} to \ref{ax:ext}, whereby we relax the $\rbt$ component to $\rbt; [\VSet]$. 
Therefore, as neither of $\txid, \txid'$ in \cref{subfig:rb-ok} are visible, there are no $\rbt; [\VSet]$ cycles in the execution, rendering it opaque.

Note that \ref{ax:ext} is strong enough to preclude bogus executions such as that in \cref{subfig:rb}.
Specifically, although $\txid_r$ observes the write on $y$ by $\txid$ ($\writeL y 1 \relarrow{\rf} \readL y 1$), it neglects the write on $x$ by $\txid$ and instead observes that of $\txid'$ ($\writeL x 0 \relarrow{\rf} \readL x 0$) which is stale ($\writeL x 0 \relarrow{\mo} \writeL x 1$). 
In other words, this executions is bogus as $\txid_r$ observes a partial effect of $\txid$.
Our notion of opacity thus rightly renders this execution not opaque since $\txid$ is visible (it is commit-pending and read from), and this execution contains a prohibited cycle: $\writeL y 1 \relarrow{\rft} \readL x 0 \relarrow{\rbt;[\VSet]} \writeL y 1$.

\begin{definition}[Opacity]
\label{def:opacity}
An execution $G = (\Events, \po, \clo, \rf, \mo)$ is \emph{opaque} iff:
\begin{itemize}
	\item $\dom(\rft) \suq \VSet$\labelAxiom{vis-rf}{ax:vis_rf}
	\item $\rfI \cup \moI \cup \rbI \subseteq \po $%
	\labelAxiom{int}{ax:int1}
%
	\item $(\clo \cup \rft \cup \mot \cup (\rbt;[\VSet]))$ is acyclic \labelAxiom{ext}{ax:ext}
\end{itemize}
where
$\VSet \eqdef \SSet \cup \CPRFSet$ with 
$\CPRFSet \eqdef \dom([\CPSet]; \rft)$.
\end{definition}

The existing definition of opacity \cite{DBLP:series/synthesis/2010Guerraoui} does not account for memory allocation and assume that all locations accessed (read/written) by a transaction are initialised with some value (typically 0).
In our setting, we make no such assumption and extend the notion of opacity to \emph{dynamic opacity} to account for memory allocation. 
More concretely, our goal is to ensure that accesses in \emph{visible} transactions are \emph{valid}, in that they are on locations that have been previously allocated in a visible transaction. 
Note that we only require accesses in visible transactions to be valid as an invisible transaction with invalid accesses may always abort. 
Specifically, it suffices to require that \emph{write} accesses in visible transactions be valid, as the validity of write accesses then follows from the validity of write accesses and the \ref{ax:vis_rf} axiom of opacity.
In particular, given an execution with $(a, r) \in \rf$ and $a \in \Allocs \cup \Writes$, 
either 
\begin{enumerate*}
	\item $a \in \Allocs$, in which case from \ref{ax:vis_rf} we know $a \in \VSet$ and thus $r$ is a valid access; or
	\item $a \in \Writes$, in which case so long as $a$ itself is valid, then $r$ would also be valid. 
\end{enumerate*}
To this end, we define an execution to be dynamically opaque (\cref{def:dynamic_opacity}) if its visible write accesses are valid, \ie are $\mo$-preceded by a visible allocation.

\begin{definition}[Dynamic opacity]
\label{def:dynamic_opacity}
An execution $G$ is \emph{dynamically opaque} iff it is opaque (\cref{def:opacity}) and
$
	G.(\Writes \cap \VSet) \suq \rng\big([\Allocs \cap \VSet]; G.\mo\big)
$.
\end{definition}

We next use the above definitions to define (dynamic durable) opacity over execution \emph{histories}.
In the context of persistent memory where executions may crash (\eg due to a power failure) and resume thereafter upon recovery, a history is a sequence of events (\cref{def:events})  partitioned into different \emph{eras} separated by \emph{crash markers} (recording a crash occurrence), provided that the threads in each era are distinct, \ie thread identifiers from previous eras are not reused after a crash.

\begin{definition}[Histories]
\label{def:histories}
A \emph{history}, $\hist \!\in\! \Hists$, is a pair $(\Events, \teo)$, where $\Events$ comprises events and \emph{crash markers}, $\Events \suq \EventsType \cup \CrashMarkers$ with $\CrashMarkers \eqdef \setcomp{(n, \crash)}{n \!\in\! \Nats}$, and $\teo$ is a total order on $\Events$, such that:
\begin{itemize}
	\item $(\Events, \internal\teo)$ is well-formed (\cref{def:wf_execution}); and
	\item events separated by crash markers have distinct threads: $([\Events]; \teo; [\CrashMarkers]; \teo; [E]) \cap \internal\teo = \emptyset$.
\end{itemize}
A history $(\Events'\!, \pteo)$ is a \emph{prefix} of history $(\Events, \teo)$ iff $E' \suq E$, $\pteo = \coerce \teo {E'}$ and $\dom(\teo; [\Events']) \suq \Events'$.
\end{definition}

\begin{definition}
The \emph{client order} induced by a history $\hist= (\Events, \teo)$,
denoted by $\clo(\hist)$, is the partial order on $\TXIDs$
defined by $\clo(\hist) \defeq [\Succs\cup \Aborts] ; \tlift{\teo} ; [\Begins]$.
\end{definition}

We define history opacity as a \emph{prefix-closed} property (\cf
\cite{DBLP:series/synthesis/2010Guerraoui}), designating a history
$\hist$ as opaque if \emph{every prefix} $(\Events, \pteo)$ of $\hist$
induces an opaque execution. 
 
\begin{definition}[History opacity]
\label{def:history_opacity___old}
A history $\hist$ is \emph{opaque} iff for each prefix $\hist_p = (\Events, \pteo)$ of $\hist$, there exist $\rf, \mo$ such that 
$(\Events, \internal\pteo, \clo(\hist_p), \rf, \mo)$ is opaque (\cref{def:opacity}). 
\end{definition}

The notion of dynamic opacity over histories is defined analogously.

\begin{definition}[History dynamic opacity]
\label{def:history_dyn_opacity}
A history $\hist$ is \emph{dynamically opaque} iff 
for each prefix $\hist_p {=} (\Events, \pteo)$ of $\hist$, 
there exist $\rf, \mo$ such that $(\Events, \internal\pteo, \clo(\hist_p), \rf, \mo)$ is dynamically opaque (\cref{def:dynamic_opacity}). 
\end{definition}

We define \emph{durable opacity} over histories: a history $\hist$ is durably
opaque iff the history obtained from $\hist$ by removing crash
markers is opaque. 
We define \emph{dynamic, durable
  opacity} analogously.

\begin{definition}[History durable opacity]
\label{def:history_dur_opacity}
A history $(\Events, \teo)$ is \emph{durably opaque} iff the history $(\Events \setminus \CrashMarkers, \coerce \teo {\Events \setminus \CrashMarkers})$ is opaque.
A history $(\Events, \teo)$ is \emph{dynamically and durably opaque} iff the history $(\Events \setminus \CrashMarkers, \coerce \teo {\Events \setminus \CrashMarkers})$ is dynamically opaque.
\end{definition}

Finally, we show that our definitions of history (durable) opacity are
equivalent to the original definitions in the literature. (See \cref{app:org_opacity}
for the proof.)
\begin{theorem}
  \label{thm:orig-opacity}
  History opacity as defined in \cref{def:history_opacity___old} is
  equivalent to the original notion of
  opacity~\cite{DBLP:series/synthesis/2010Guerraoui}. History durable
  opacity as defined in \cref{def:history_dur_opacity} is equivalent
  to the original notion of durable
  opacity~\cite{DBLP:conf/forte/BilaDDDSW20}.
\end{theorem}

\newcommand{\Pending}{\Pi}
\newcommand{\DoPending}{\Delta}
\newcommand{\silent}{\varepsilon}

\section{Operationally Proving Dynamic Durable Opacity}
\label{sec:oper-prov-dynam}
We develop an operational specification, \DDTMS (\cref{sec:DTMS2-full}), and prove it correct against \DDO (\cref{sec:soundness-ddtms}). In particular, we show that every history (\ie observable trace) of \DDTMS satisfies \DDO. As \DDTMS is a
concurrent operational specification, it serves as basis for validating
the correctness of \PMDKTX as well as our concurrent extensions \PMDKT and
\PMDKN via refinement.


\subsection{\DDTMS: The \DTMS Automaton Extended with Allocation}
\label{sec:DTMS2-full}

\DDTMS is based on \DTMS, which is an operational specification
that guarantees durable
opacity~\cite{DBLP:conf/forte/BilaDDDSW20}. \DTMS in turn is based on
TMS2 automaton~\cite{DBLP:journals/fac/DohertyGLM13}, which is known
to satisfy opacity~\cite{LLM12}. 
Furthermore, the  \DDTMS commit operation includes the
simplification described by~Armstrong et al~\cite{DBLP:conf/forte/ArmstrongDD17},
omitting a validity check when committing read-only
transactions. In what follows we present \DDTMS as a transition system. 
In \cref{app:DDTMS_automata} we present an equivalent model based on
input/output automata for those familiar with prior work.

\begin{figure}[t]
  \centering
  $\small
\begin{array}{@{} l @{}}
\begin{array}{@{} r @{\hspace{2pt}} l @{}}
	\mems \in \Mems \eqdef & \Seq{\Locs \to \Vals_\bot} \qquad  \Vals_\bot \eqdef \Vals \cup \{\bot\}, \text{where } \bot \notin \Vals \\
	\smap \in \SMaps \eqdef & \TXIDs \rightarrow \States \\
  \s \in \States \eqdef & \Nats \times (\Locs \pfun \Vals) \times (\Locs \pfun \Vals) \times \pset{\Locs} \\
  & \text{storing the local begin index, read set, write set and allocation set}  \\
	\pcmap \in \PCMaps \eqdef & \TXIDs \rightarrow \PCVals 
  \\
	\textsc{Invs} \eqdef &
	\setcomp{\tmbegin,
                \tmread(l),
		\tmwrite(l, v), \\
		\tmalloc, 
		\tmcommit
	}{
		l \in \Locs, 
		v \in \Vals
	}\\
	\textsc{Resps} \eqdef &
	\setcomp{\tmbegin,
                \tmread(l,v),
		\tmwrite(l, v), \\
		\tmalloc(l), 
		\tmcommit, \abort
	}{
		l \in \Locs, 
		v \in \Vals
	}\\
  \PCVals \eqdef &
	\setcomp{
		\pcNotStarted, 
		\pcReady, 
		\pcAborted, 
		\pcCommitted, \pcChaos, 
                \Pending(i),
		\DoPending(\tmcommit) 
	}{
		i \in \textsc{Invs}
	}\\
	\act \in \Actions \eqdef  &
	\begin{array}[t]{@{} l @{}}
          \setcomp{\invkw(i), \reskw(r), \silent, \crash
          }{
          i \in \textsc{Invs},
          r \in \textsc{Resps}
          }
        \end{array}
  \end{array}
  \\[5pt]
  \hline
  \\[-5pt]
  \text{Initially, } 
  \begin{array}[t]{@{} l @{}}
    \pcmap_0 \eqdef 
    \lambda\txid. \pcNotStarted \qquad
    \smap_0 \eqdef
    \lambda \txid. (0, \emptyset, \emptyset, \emptyset)
    \qquad \mems_0 \eqdef [\lambda x.\ \bot] 
  \end{array}\\
\end{array}$
\caption{\DDTMS state}
\label{fig:ddtms-state}
\end{figure}


\smallskip\emph{\textbf{\DDTMS state.}}
Formally, the state of \DDTMS is given by the variables in
\cref{fig:ddtms-state}.  
\DTMS keeps track of a sequence of memory stores, $\mems$, one for
each committed writing transaction since the last crash. This allows
us to determine whether reads are consistent with previously committed
write operations. Each committing transaction that contains at least one
write adds a new memory version to the end of the memory sequence. As
we shall see, $\mems$ also tracks allocated locations since it maps
every allocated location to a value different from $\bot$. 

Each transaction $\txid$ is associated with several variables: $\getPC$, $\getBIdx$, $\getRSet$, $\getWSet$ and $\getASet$. 
The $\getPC$ denotes the  program counter, ranging over a set of {\em program counter values} ensuring each transaction is well-formed and
that each transactional operation takes effect between its invocation
and response. 
The $\getBIdx \in \Nats$ denotes the {\em begin index}, 
set to the index of the most recent
memory version when the transaction begins. This is used to
ensure the real-time ordering property between transactions. 
The $\getRSet \in \Locs \pfun \Vals$ is the {\em read set} and $\getWSet \in \Locs \pfun \Vals$ is the {\em write set}, recording the values read and written by 
the transaction during its execution,
respectively. We use $S \pfun T$ to denote a partial function from $S$
to $T$. 
Finally, $\getASet \subseteq \Locs$ denotes the {\em allocation set}, containing the set of locations allocated by the transaction $\txid$.  
We use $\s.\bidx$, $\s.\rset$, $\s.\wset$ and $\s.\aset$ to refer to the begin index,
read set, write set and allocation set of a state $\s$, respectively.

The read set is used to determine whether the values read by the transaction are consistent with its version of memory (using $\isValidIdxName$). The write set, on the other hand, is required because writes are modelled using {\em deferred update} semantics:
writes are recorded in the transaction's write set and are not
published to any shared state until the transaction commits.

\begin{figure}
  \fontsize{8.6pt}{10.32pt}\selectfont
  \def \MathparLineskip {\lineskip=0.25cm}
  $\begin{array}[t]{@{}l@{}}
	\isValidIdx n {\s, \mems} \iffdef
	\begin{array}[t]{@{} l @{}}
		 \s.\bidx \leq n < \size{\mems} 
		 \land \s.\rset \subseteq \mems(n) \\
		 {} \land \s.\aset \subseteq \setcomp{l}{\mems(n)(l)=\bot}	 
	\end{array}
\end{array}$
\begin{mathpar} 
\infer[(\textsc{S})]{
  \inarr{\pcmap, \smap, \mems \redT{\act_{\txid}} \\
    \pcmap[\txid \mapsto \pc], \smap[\txid \mapsto \s], \mems'}
}{
  \inarrC{\pcmap(\txid), \smap(\txid), \mems \redT{\act} \pc, \s, \mems'\!
	\\ 
	 \pc {\ne} \pcChaos}
}
\quad
\infer[(\textsc{F})]{
  \inarr{\pcmap, \smap, \mems \redT{\segFault} \\ \pcmap', \smap[\txid \mapsto \s], \mems'}
}{
  \inarrC{\pcmap(\txid), \smap(\txid), \mems \redT{\segFault} \pcChaos, \s, \mems' \\
  \pcmap' {=} \lambda \txid. \pcChaos}
}
\and
\infer[(\textsc{X})]{
	\pcmap, \smap, \mems \redT{\crash} \pcmap', \smap, \langle last(\mems) \rangle
}{
	\pcmap'\! \!=\!
	\lambda \txid.\ 
	\inarrT{\sif\ \pcmap(\txid) \!\nin\! 
	\{\pcNotStarted, \pcCommitted, \pcChaos
        \} \\
	\;\sthen\; \pcAborted 
	\;\selse\; \pcmap(\txid)}
}
\quad
\infer[(\textsc{C})]{
	\pcmap, \smap, \mems \redT{\act_\txid} \pcmap, \smap, \mems
}{
	\pcmap = \lambda \txid. \pcChaos
}%
\end{mathpar}
\hrule\vspace{1pt}\hrule
\begin{mathpar}
\inferrule[\textsc{(IB)}]{
 \inarr{
   \pc = \pcNotStarted}
}{
  \pc, \s, \mems \redT{\invkw(\tmbegin)}  \\\\
    \DoPending(\tmbegin), \s', \mems
}
\quad
\inferrule[\textsc{(DB)}]{
    \pc {=} \DoPending(\tmbegin) 
    \\\\
    \s' {=} \s[\bidx \mapsto \size{\mems} {-} 1]
}{
  \pc, \s, \mems \redT{\reskw(\tmbegin)} \\\\
    \pcReady, \s, \mems}
\quad
\inferrule[\textsc{(IOp)}]{
  \pc = \pcReady \quad a \in {\it InvOps}
}
{
  \pc, \s, \mems \redT{\invkw(a)} \\\\ \DoPending(a), \s, \mems
}
\\
\inferrule[\textsc{(DR-E)}]{
	\pc \!=\! \DoPending(\tmread(l))
	\\\\ l \!\nin\! \s.\aset \cup\! \dom(\s.\wset) 
	\\\\ \isValidIdx n {\s, \mems} 
	\\\\ \mems(n)(l) = v \qquad v \neq \bot 
	\\\\ \mathit{rs} = \s.\rset \oplus \{l \!\mapsto\! v\}
}
{
  \pc, \s, \mems \redT{\reskw(\tmread(l, v))}
  \\\\
    \pcReady, \s[\rset \mapsto \mathit{rs}], \mems
}
\and
\inferrule[\textsc{(FR)}]
{
  \pc \!=\! \DoPending(\tmread(l))
  \\\\ l \!\nin\! \s.\aset \cup\! \dom(\s.\wset) 
  \\\\  \isValidIdx n {\s, \mems} 
  \\\\ \mems(n)(l) {=} \bot 
}
{
  \pc, \s, \mems \redT{\segFault} \\\\\pcChaos, \s, \mems
}
\and
\inferrule[\textsc{(RA)}]{
  \pc \not\in \inset{
		\pcNotStarted, \pcReady, \\
		\pcCommitted, \\
                \pcAborted, \pcChaos
	}}
    {
  \pc, \s, \mems \redT{\abortResp{}} \\\\ \pcAborted, \s, \mems
}
\and
\infer{
  \inarr{\pc, \s, \mems \redT{\reskw(\tmread(l, v))} \\
         \hfill \pcReady, \s, \mems}
}{
  \inarr{\textsc{(DR-I)}\\
  \begin{array}[t]{@{}c@{}}
	\pc = \DoPending(\tmread(l))
	\\ l \!\in\! \dom(\s.\wset) 
	\\   \s.\wset(l) = v
  \end{array}}
}
\quad
\infer{
  \inarr{\pc, \s, \mems \redT{\reskw(\tmread(l, 0))} \\ \hfill \pcReady, \s, \mems}
}{
  \inarr{\textsc{(DR-A)} \\
  \begin{array}[t]{@{}c@{}}
	\pc = \DoPending(\tmread(l))
	\\ l \nin \dom(\s.\wset) 
	\\ l \in \s.\aset
  \end{array}}
}
\quad 
\infer{
  \inarr{\pc, \s, \mems \redT{\reskw(\tmwrite(l,v))} \\
    \pcReady, \s[\wset \mapsto \mathit{ws}], \mems}
}{
  \inarr{\textsc{(DW)} \\
  \begin{array}[t]{@{}c@{}}
	\pc = \DoPending(\tmwrite(l, v))
	\\ l \in \s.\aset \lor \last{\mems}(l) \neq\bot
	\\ \mathit{ws} = \s.\wset \oplus \{l \!\mapsto\! v\}
  \end{array}}
}
\and
\infer{
  \inarr{\pc, \s, \mems \redT{\segFault} \\ \pcChaos, \s, \mems}
}{
  \inarr{
    \textsc{(FW)}
    \\
  \begin{array}[t]{@{}c@{}}
    \pc = \DoPending(\tmwrite(l, v))
    \\
    l \notin \s.\aset
    \\  \last{\mems}(l) = \bot
  \end{array}}
}
\and
\infer{
  \inarr{\pc, \s, \mems \redT{\reskw(\tmalloc(l))}
    \\\pcReady, \s[\aset \mapsto \mathit{as}], \mems}
      }{
        \inarr{\textsc{(DA)} \\
  \begin{array}[t]{@{}c@{}}
	\pc = \DoPending(\tmalloc) 
	\\ l \!\nin\! \s.\aset
	\\ \mathit{as} = \s.\aset \uplus \{l\}
  \end{array}}
}
\and
\infer{
  \inarr{\pc, \s, \mems \redT{\silent} \\\Pending(\tmcommit), \s, \mems}
      }{
        \inarr{\textsc{(DC-RO)} \\
  \begin{array}[t]{@{}c@{}}
	\pc = \DoPending(\tmcommit)
	\\ \s.\aset = \emptyset
	\\ \dom(\s.\wset) = \emptyset
  \end{array}}
}
\and
\infer{
  \inarrC{\pc, \s, \mems \redT{\reskw(\tmcommit)}  \\\pcCommitted, \s, \mems}
      }{
        \inarr{\textsc{(RC)} \\
  \begin{array}[t]{@{}c@{}}
    \pc = \Pending(\tmcommit)
  \end{array}}
}
\quad 
\infer{
	\pc, \s, \mems \redT{\silent} \Pending(\tmcommit), \s, \mems'
      }{
        \inarr{\textsc{(DC-W)} \\
	\begin{array}{@{} c @{}}
		\pc = \DoPending(\tmcommit)
		\\ \isValidIdx {\last{\mems}} {\s, \mems}
          \\
		\mems'\! \!=\! \mems \!\cat\! ((\last{\mems} \!\oplus\! \setcomp{l \!\mapsto\! 0}{l \!\in\! \s.\aset}) \!\oplus\! \s.\wset)
	\end{array}}
}
\end{mathpar}
\hrule
\caption{The \DDTMS global transitions (above) with its per-transaction transitions (below), where ${\it InvOps} \eqdef \{\tmwrite(l, v), \tmread(l) \mid l
  \in \Locs, v \in \Vals\} \cup \{ \tmalloc, \tmcommit\}$}
\label{fig:ddtms}
\end{figure}


\paragraph{\DDTMS Global Transitions}
\DDTMS is specified by the
transition system shown in \cref{fig:ddtms}, where the \DDTMS global
transitions are given at the top and the per-transaction transitions
are given at the bottom. The global transitions may either take a
per-transaction step (rule (\textsc{S})), match a transaction fault
(rule (\textsc{F})), crash (rule (\textsc{X})), or behave chaotically
due to a fault (rule (\textsc{C})).

Note that a {\em crash} transition models both a crash and a
recovery. It sets the program counter of every live transaction to
$\pcAborted$, which prevents these transactions from performing any
further actions after the crash. Since transaction identifiers are not
reused, the program counters of completed transactions need not be
modified. After restarting, it must not be possible for any new
transaction to interact with memory states prior to the crash. We
therefore reset the memory sequence to be a singleton sequence
containing the last memory state prior to the crash.

As already discussed, following the design of \PMDKTX (and our
concurrent extensions \PMDKT and \PMDKN) we do not check for reads and
writes to unallocated memory within the library and instead delegate such
checks to the client. An execution of \PMDKTX (as well as \PMDKT and
\PMDKN) that accesses unallocated memory is assumed to be faulty. In
particular, a read or write of unallocated memory induces a {\em
  fault} (rule (\textsc{F})). Once a fault is triggered, the program
counter of each transaction is set to ``$\pcChaos$'' and recovery is
impossible.  From a faulty state the system behaves chaotically, \ie
it is possible to generate any history using rule (\textsc{C}).

\paragraph{\DDTMS Per-Transaction Transitions}
The system contains externally visible transitions for {\em invoking}
an operation (rules \textsc{IB} and \textsc{IOp}), which set the
program counters to $\DoPending(a)$, where $a$ is the operation being
performed. This allows the histories of the system to contain
operation invocations without corresponding matching responses.

For the begin, allocation, read and write operations an invocation can be
followed by a single transition (rules \textsc{DB}, \textsc{DA},
\textsc{DR-E}, \textsc{DR-I}, \textsc{DR-E} and \textsc{DW})
that performs the operation combined with the corresponding {\em
  response}. Note that this differs from {\em canonical input/output
  automata}~\cite{DBLP:books/mk/Lynch96}, and consequently \DTMS,
where each operation has a separate internal ``do'' action followed by
an externally visible response. In the case of \DDTMS, the ``do''
and ``response'' for begin, allocation, read and write can be combined
into one action as the effect of these operations are not
externally visible to other transactions until a commit
occurs. Note that $\reskw(\tmread(l,v))$ in \textsc{DR-E},
\textsc{DR-I} and \textsc{DR-E} can be simplified to
$\reskw(\tmread(v))$, and $\reskw(\tmwrite(l,v))$ in \textsc{DW}
can be simplified to $\reskw(\tmwrite(v))$. However, the location
that each read/write operation is responding to is needed in
\cref{sec:soundness-ddtms}.

Following an invocation, the commit operation is split into internal
do actions ((\textsc{DC-RO}) and (\textsc{DC-W})) and an external
response (rule \textsc{RC}). Finally, after a read/write invocation, the
system may perform a {\em fault transition} for a read (rule
\textsc{FR}) or a write (rule \textsc{FW}).

We next discuss the \DDTMS behaviour in detail. The main change from
$\DTMS$ is the inclusion of an allocation procedure to model
dynamic allocation within a transaction. This procedure affects other operations. The design of \DDTMS allows the
executing transaction, $\txid$, to tentatively allocate a location $l$
within its transaction-local allocation set,
$\getASet$. 
Note that this allocation in
$\DDTMS$ is  
optimistic -- correctness of the allocation is only checked when
$\txid$ performs a read or
commits. 

Successful (non-faulty) read and write operations take allocations
into account as follows.
\begin{enumerate*}[label=\bfseries(\arabic*)]
\item A read operation of transaction $\txid$ reads from a prior write
  (rule \textsc{(DR-I)}) or allocation (rule \textsc{(DR-A)})
  performed by $\txid$ itself. In this case, the operation may only
  proceed if the location $l$ is either in the allocation or write set
  of $\txid$. The effect of the operation is to return the value of
  $l$ in the write set (if it exists) or $0$ if it only exists in the
  allocation set.
\item A read operation of transaction $\txid$ reads from a write or
  allocation performed by another transaction (rule
  \textsc{(DR-E)}). 
  Note that as with $\DTMS$ and TMS2, in $\DDTMS$ a read-only
  transaction may serialise with any memory index $n$ after
  $\getBIdx$. 
  Moreover, within $\isValidIdxName$, in
  addition to ensuring that $\txid$'s read set is consistent with the
  memory index $n$ (second conjunct), we must also ensure that
  $\txid$'s allocation set is consistent with memory index $n$ (third
  conjunct) by ensuring that none of the locations in the allocation
  set have been allocated at memory index $n$. 
\item A write of transaction $\txid$ successfuly performs its
  operation (rule (\textsc{DW})), which can only happen if the
  location $l$ being written has been allocated, either by $\txid$
  itself (first disjunct), or by a prior transaction (second
  disjunct). A writing transaction must serialise after the
  last memory index in $\mems$, thus the second disjunct checks
  allocation against the last memory index. 
\end{enumerate*}

A successful (non-faulty) transaction is split into two cases:
\begin{enumerate}[label=\bfseries(\arabic*)]
\item $\txid$ is a read-only transaction (rule (\textsc{DC-RO})),
  where both $\getASet$ and $\getWSet$ are empty for $\txid$. In this
  case, the transaction simply commits.
\item $\txid$ has performed an allocation or a write (rule
  (\textsc{DC-W})). In this case, we check that $\txid$ is valid
  with respect to the last memory in $\mems$ via the predicate
  $\isValidIdxName$. 
    Here,
    the commit introduces a new memory into the memory sequence
    $\mems$. The update also ensures that all pending allocations in
    $\getASet$ take effect prior to applying the writes from $\txid$'s
    write set.

\end{enumerate}


\subsection{Soundness of \DDTMS}
\label{sec:soundness-ddtms}

We state our main theorem relating \DDTMS to \DDO. 
As the models are inherently different, we need several
definitions to transform \DDTMS histories to those compatible with \DDO. 

An {\em execution} of a labelled transition system (LTS) is an
alternating sequence of states and actions, \ie a sequence of the
form
$s_0\, a_1\, s_2\, a_2 \dots s_{n-1}\,
a_{n}\, s_n$ such that for each $0 < i \le n$,
$s_{i-1} \redT{a_i} s_{i}$ and $s_0$ is an initial
state of the LTS.
Suppose $\sigma$ is an execution of $\DDTMS$. We let
${\it AH}_\sigma = a_1\, a_2 \dots a_n$ be the {\em action history}
corresponding to $\sigma$, and ${\it EH}_\sigma$ be the {\em external
  history} of $\sigma$, which is ${\it AH}_\sigma$ restricted to
non-$\epsilon$ actions. Let ${\it FF}_\sigma$ be the longest {\em
  fault-free prefix} of ${\it
  EH}_\sigma$.  
We generate the history (in the sense of \cref{def:histories})
corresponding to ${\it FF}_\sigma$ as follows. First, we construct the
{\em labelled history}, ${\it LH}_\sigma$ of $\sigma$ from
${\it FF}_\sigma$ by removing all invocation actions (leaving only
response and crash actions). Then, we replace each response
$a_i = \alpha_{\txid}$ by the event $(i, \txid, \txid, L(\alpha))$,
where we have:
\begin{align*}
  L(\reskw(\tmbegin)) & = \tbeginL
  & L(\reskw(\tmalloc(l))) & = \allocL l
  \\ L(\reskw(\tmread(l, v))) & = \readL l v
  & 
  L(\reskw(\tmwrite(l, v))) & = \writeL l v
  \\ L(\invkw(\tmcommit)) & = \commitL
  & 
  L(\reskw(\tmcommit)) & = \succL
  \\
  L(\reskw(\abort)) & = \abortL
\end{align*}
Similarly, we replace each crash action $a_i = \crash$ by the pair
$(i, \crash)$.  Note that in this construction, for simplicity, we
conflate threads and transactions, but this restriction is
straightforward to generalise.  Finally, we let the {\em ordered
  history} of $\sigma$, denoted ${\it OH}_\sigma$, be the total order
corresponding to ${\it LH}_\sigma$.

\begin{theorem}
  \label{thm:soundness-ddtms-1}
  For any execution $\sigma$ of \DDTMS, the ordered history
  ${\it OH}_\sigma$ satisfies \DDO.
\end{theorem}

The definitions of (dynamic) durable opacity 
can lifted to the level of systems in
the standard manner, providing a notion of correctness for 
implementations~\cite{DBLP:journals/toplas/HerlihyW90}.



\OMIT{Our discussion thus far has focussed on individual histories.  An LTS is said to satisfy durable opacity iff every history
of the LTS satisfies durable history. For dynamic durable history, we
weaken this requirement slightly, and say that an LTS satisfies \DDO
iff every fault-free history of the LTS satisfies \DDO. By
\cref{thm:soundness-ddtms-1}, we therefore have that \DDTMS satisfies
\DDO.

   


As we shall see, the refinement check in \FDR proves history inclusion
between the two systems. In particular, a concrete system $B$ is a
refinement of an abstract system $A$ iff for every execution $\sigma$
of $B$, there exists an execution $\tau$ of $A$ such that
${\it H}_\sigma = {\it H}_\tau$, where ${\it H}_\sigma$ is
${\it EH}_\sigma$ with all fault actions hidden. Thus,
${\it H}_\sigma$ contains invocation, response and crash events only.

We set up the models of our implementations so that each operation
starts with an invocation and ends with a response (possibly an
abort). The system may also crash at any time, generating crash
events. All other events (including $\segFault$) are hidden.  For such
a system $A$ (\eg \PMDKTX), if $A$ refines \DDTMS, then every history
of $A$ must either be faulty (with a hidden fault) or correspond to a
total order that satisfies \DDO. That is, if $A$ refines
\DDTMS, then $A$ satisfies \DDO.
}

\section{Modelling and Validating Correctness in \FDR}
\label{sec:modell-verify-corr}

\FDR~\cite{fdr} is a model checker for
CSP~\cite{DBLP:journals/cacm/Hoare83a} that has recently been used to
verify linearisability~\cite{DBLP:conf/birthday/Lowe17}, as well as
opacity and durable opacity~\cite{DBLP:conf/sefm/DongolL21}. We
similarly provide an \FDR development, which allows proofs of
refinement to be {\em automatically} checked up to certain bounds.
This is in contrast to manual methods of proving correctness of
concurrent
objects~\cite{DBLP:journals/csur/DongolD15,DBLP:journals/fac/DerrickDDSTW18},
which require a significant amount of manual human input (though such
manual proofs are unbounded).

\OMIT{
FDR3 uses a variety of underlying model checking paradigms and partial
order reduction
techniques~\cite{DBLP:conf/tacas/Gibson-RobinsonABR14}, depending on
the structure of the files to be verified. FDR4 builds on FDR3, but
the exact implementation details of FDR4 are not publicly available
since it is a commercial product (available for free academic
use). 

One of the most challenging aspects of the \FDR development is the
modelling work itself. Our algorithms execute over a shared memory,
but the CSP formalism is based on {\em communicating processes} with
no notion of shared states. Thus, for each location we must explicitly
define {\em handler} processes that communicate through {\em channels}
to update and return the values of components (\eg the addresses,
read/write sets) of each model. Moreover, the implementations
(\PMDKTX, \PMDKN and \PMDKT), the specification (\DDTMS) and
underlying memory models (\PSC and \PTSOS) we consider are
non-trivial, significantly increasing the challenge of the modelling
effort.  
Although constructing the models is challenging,
once the models have been developed, they can be combined in a modular
fashion. We have taken advantage of this feature to combine our
implementations with different memory models during development. The
combination of \PMDKT and TML/NOrec also takes advantage of this
modularity.

This modularity also means that our models are reusable. One could use
our models to check other developments, \eg those that use \PMDKTX to
implement other failure-atomic data structures, or verify redesigns of
\PMDKTX over different memory models.  Specifically, we use a
top-level CSP process (which may comprise an interleaved composition
of processes for each transaction) to model the most general
client. Each transaction process begins a transaction, and then calls
an unbounded number of reads, writes and allocations at
non-deterministically chosen locations and with non-deterministically
chosen values. An in-flight transaction process may also
non-deterministically choose to terminate by calling commit instead of
calling a read, write or allocation. Each operation call produces an
externally visible invocation event, and when the operation
terminates, an externally visible response is generated. Some
operations may respond with an abort, in which case the transaction
process itself terminates.

Additionally, there is an externally visible crash
event that synchronises with all processes. At the level of the
abstraction (\ie \DDTMS), this simply terminates all in-flight
transactions, and resets the memory sequence (as detailed by the rule
(\textsc{X}) in \cref{fig:ddtms}). At the level of the implementation,
all in-flight transactions are terminated and additionally, the store
and persistency buffers are cleared. This means that when execution
resumes, the value of each location is taken from NVM (as detailed in
\cref{sec:memory-models}). Immediately after a crash (and before any
other processes are started), the recovery process corresponding to
the algorithm is executed. Note that transaction identifiers are never
reused. 

We eschew further details of our \FDR models since they are provided
as supplementary material and also refer the interested reader to
other prior works
~\cite{DBLP:conf/birthday/Lowe17,
  DBLP:conf/sefm/DongolL21}.
}

 

\begin{figure}[t]
  \centering
  \begin{minipage}[b]{0.35\columnwidth}
  \scalebox{0.75}{
  \begin{tikzpicture}[
    node distance=8mm,
    title/.style={font=\small\color{black!50}},
    mod/.style={exec, font=\small, anchor=west}
    ]
    \node (decomp) [title] {Implementations};
    
    \node (di) [below=6mm of decomp.west, mod] { \PMDK };
    \node (dr) [below=of di.west, mod] { \PMDKT };
    \node (dnc) [below=of dr.west, mod] { \PMDKN };
    
    \node [draw=black!50, rounded corners, fit={(decomp) (di) (dr) (dnc)}] (Imp) {};
    
    \node (dep) [below right = -0.5cm and 1cm of decomp, title] { \intab{Memory \\ models} };
    
    \node (da) [below=10mm of dep.west, mod] { \PSC };
    \node (dr) [below=of da.west, mod] { \PTSOS };
    
    \node [draw=black!50, rounded corners, fit={(dep) (dr) (da)}] (MM) {};
    
    \node [draw, rounded corners, fit={(Imp) (MM)}] (Syst) {};
        
    \node[specOP,below= 4mm of Syst] (DDTMS) {
      \small\begin{tabular}{@{}c@{}}
        \DDTMS \\
              (concurrent 
              upper bound)
      \end{tabular}
    };
    
    \node[specOP,above= 4mm of Syst] (DDTMSS) {\small
      \begin{tabular}{@{}c@{}}
        \DDTMS-Seq \\
        (sequential 
        lower bound)
      \end{tabular}
    };
    
    \draw[color=white] (Syst) -- node[above,color=black,yshift=-5.5]{\footnotesize refines}(DDTMS) ;
    \draw[color=white] (DDTMSS) -- node[above,color=black,yshift=-5.5]{\footnotesize refines}(Syst) ;
    
    \draw[color=white] (Imp) -- node[above,color=black]{\small uses}(MM) ;

  \end{tikzpicture}}
  \caption{Overview of \FDR checks}
  \label{fig:FDR}
\end{minipage}
\hfill
\begin{minipage}[b]{0.61\columnwidth}
  \centering \footnotesize
  \scalebox{0.95}{
  \begin{tabular}[t]{|@{}l@{}|@{}c@{}|@{}c@{}|@{}c@{}|@{}c@{}||@{}r@{}|@{}r@{}|@{}r@{}|}
    \hline
    %

    Memory  & \scalebox{0.8}{\#txns} & \scalebox{0.8}{\#locs}  & \scalebox{0.8}{\#val} & \scalebox{0.8}{\#buff} & \ \scalebox{0.8}{\PMDKTX} &
                                                          \begin{tabular}{l}
                                                            PMDK-\\
                                                            TML
                                                          \end{tabular}
&                                                           \begin{tabular}{l}
                                                            PMDK-\\
                                                              \NOREC
                                                            \end{tabular}
 \\
   \hline
     \PSC      & 2      & 2       & 2  & 2         & 5.83s & 5.90s & 6.74s 
     \\
     \PSC      & 2      & 3       & 2  & 2         & 201.03s & 213.97s & 271.35s
     \\
     \PSC      & 2      & 2       & 3  & 2         & 21.65s & 23.47s & 27.40s
     \\
     \PSC      & 2      & 2       & 2  & 3         & 5.83s & 5.78s & 6.60s
    \\
    \hline
     \PTSOS & 2      & 1       & 2   & 2        & 0.61s & 3.96s  & 1.57s
     \\
     \PTSOS & 2      & 2       & 2   & 2        & 6.67s & 6.71s  & 7.73s 
     \\
     \PTSOS & 2      & 3       & 2   & 2        & 267.1s & 268.91s & 319.18s 
     \\
     \PTSOS & 2      & 2       & 3   & 2        & 24.10s & 25.53s  & 29.24s 
     \\
     \PTSOS & 2      & 2       & 2   & 3        & 14.37s & 14.19s  &  15.41s 
     \\
    \hline
\end{tabular}}
\caption{Summary of upper bounds checks (total time in seconds:
  compilation + model exploration). The time out (TO) is set to 1000
  seconds of compilation time. }
\label{tab:summary}
\end{minipage}
\end{figure}

An overview of our \FDR development is given in \cref{fig:FDR}. We
derive two specifications from \DDTMS. The first is an \FDR model of
\DDTMS itself, based on prior work~\cite{DBLP:conf/birthday/Lowe17,DBLP:conf/sefm/DongolL21}, but
contains the extensions described in \cref{sec:DTMS2-full}. The second
is \DDTMS-Seq, which restricts \DDTMS to a sequential crash-free
specification. We use \DDTMS-Seq to obtain (lower-bound) liveness-like
guarantees, which strengthens traditional deadlock or
divergence proofs of refinement. These lower-bound checks ensure our models contain at least the traces of \DDTMS-Seq.


\OMIT{
Each check in \FDR is split into two phases: 
\begin{enumerate*}
\item a compilation phase that builds the models; and 
\item a model exploration phase. 
\end{enumerate*}
The characteristics of the upper and lower bounds checks are
distinct. When naively checking the upper bound, compilation is almost
instantaneous but model exploration times can be significant; these
times are swapped for the lower bounds checks.

In general, lower-bounds take much longer to verify than the
upper-bounds since \FDR is optimised to verify abstract (low-detail)
specifications are refined by concrete (high-detail)
implementations. The lower bounds checks use the more complex models
as the specification, leading to the creation of very large
space-inefficient models, putting pressure on the available system
memory. However, the lower-bound checks for \PSC and \PTSOS are
superceded by the corresponding checks over \NVM, since the memory
models \PSC and \PTSOS are both supersets of \NVM. That is, any trace
over \NVM must also be a trace \PSC and \PTSOS. For two transactions,
two locations and two values, the checks for \PMDK, \PMDKT and \PMDKN
take 12.16, 17.36, and 56.02 seconds, respectively.
}

\cref{tab:summary} summarises our experiments on the upper bound
checks, where the times shown combine the compilation and model
exploration times. Each row represents an experiment that bounds
the number of transactions (\#txns), locations (\#locs),
values (\#val) and the size of the persistency and store
buffers (\#buff). The times reported are for an Apple M1 device with
16GB of memory. The first row depicts a set of experiments where the
implementations execute directly on NVM, without any buffers.
As we discuss below, these tests are sufficient for checking lower bounds. The baseline for our
checks sets the value of each parameter to two, and \cref{tab:summary}
allows us to see the cost of increasing each parameter. Note that all
models time out when increasing the number of transactions to three,
thus these times are not shown.  Also note that for \PMDKTX (which is
single-threaded), the checks for \PSC also cover \PTSOS, since \PTSOS
is equivalent to \PSC in the absence of
races~\cite{DBLP:journals/pacmpl/KhyzhaL21}. Nevertheless, it is
interesting to run the single-threaded experiments on the \PTSOS model
to understand the impact of the memory model on the checks.

In our experiments we use \FDR's built-in \emph{partial order
  reduction} features to make the upper bound checks feasible.  This
has a huge impact on the model checking speed; for instance, the check
for \PMDKT with two transactions, two locations, two values and buffer
size of two reduces from over 6000 seconds (1 hour and 40 minutes) to
under 7 seconds, which is almost a 1000-fold improvement! This
speed-up makes it feasible to use \FDR for rapid prototyping when
developing programs 
that use \PMDKTX, even
for the relatively complex \PTSOS memory model.

\section{Related Work}
\label{sec:related}

\emph{\textbf{Crash Consistency.}}
Research on persistent memory has progressed remarkably quickly, and
efforts such as Intel's \PMDK libraries have meant that development of
production-level code is now becoming feasible. Works on persistent
memory have leveraged the many existing approaches to
failure-resilient and fault-tolerant computing from areas such as
databases and concurrent/distributed systems. 
Several authors have defined notions of atomicity for {\em concurrent
  operations} that take persistency into account. Guerraoui and
Levy~\cite{DBLP:conf/icdcs/GuerraouiL04} define two notions, {\em
  persistent atomicity} (which allows threads to survive a crash, but
requires that every pending operation on a crashed thread either take
effect or abort before a subsequent operation of the same thread is
invoked), and {\em transient atomicity} (which relaxes persistent
atomicity by allowing an incomplete operation to take effect before a
subsequent write response of the same thread even if the operation is
interrupted by a crash). Berryhill
\etal~\cite{DBLP:conf/opodis/BerryhillGT15} define {\em recoverable
  linearisability}, which requires every pending operation on a thread
to take effect or abort before the thread linearises another
operation. Izraelevitz \etal~\cite{DBLP:conf/wdag/IzraelevitzMS16}
define two other notions: {\em buffered durable linearisability} and
{\em durable linearisability}, which generalise linearisability to
account for crashes. In particular, durable linearisability requires
that a history with crashes be linearisable when crashes are removed;
buffered durable linearisability allows one to `forget' operations (in
a manner consistent with happens-before) even if they have
completed. Although \DDO is inspired by durable opacity, none of these
conditions are suitable as they define consistency for concurrent
operations (\eg those of concurrent data structures) as opposed to
transactional memory.

Approaches and semantics to {\em crash-consistent} transactions
stretch back to the mid 1970s, which considered the problem in the
database
setting~\cite{DBLP:journals/cacm/EswarranGLT76,DBLP:conf/sigmod/LienW78}. Since
then, a myriad of definitions have been developed for particular
applications (\eg distributed systems, file systems, \etc). For plain
reads and writes, one of the first studies of persistency models
focussed on NVM is by Pelley
\etal~\cite{DBLP:journals/micro/PelleyCW15}. Since then, several
semantic models for real hardware (Intel and ARM) have been
developed~\cite{DBLP:journals/pacmpl/RaadWV19,DBLP:journals/pacmpl/RaadWNV20,DBLP:journals/pacmpl/KhyzhaL21,DBLP:conf/pldi/ChoLRK21,DBLP:journals/pacmpl/RaadMV22}. For
transactional memory, there are only a few notions that combine a
notion of crash consistency with ACID guarantees as required for
concurrent durable transactions. Raad
\etal~\cite{DBLP:journals/pacmpl/RaadWV19} define a {\em persistent
  serializability} under relaxed memory, which does not handle aborted
transactions. As we have already discussed, Bila
\etal~\cite{DBLP:conf/forte/BilaDDDSW20} define {\em durable opacity},
but this is defined in terms of (totally ordered) histories as opposed
to partially ordered graphs. Neither persistent serialisability nor
durable opacity handle memory allocation.

\emph{\textbf{Validating the \PMDKTX Implementation.}}
Even without a clear consistency condition, a range of papers have
explored correctness of the C/C++ implementation.  Bozdogan
\etal~\cite{SafePM} built a sanitiser for persistent memory and used
it to uncover memory-safety violations in \PMDKTX. Fu
\etal~\cite{witcher} have built a tool for testing persistent
key-value stores and uncovered consistency bugs in the \PMDK
libraries. Liu \etal~\cite{XFDetector} have built a tool for detecting
cross-failure races in persistent programs, and uncovered a bug in
\PMDK's libpmemobj library (see `Bug 4' in their paper). They are at a
different level of abstraction than ours since they focus on the code
itself and do not provide any description of the design principles
behind \PMDK.

Raad \etal~\cite{DBLP:journals/pacmpl/RaadLV20} and Bila
\etal~\cite{DBLP:conf/esop/BilaDLRW22} have developed logics for
reasoning about programs over the Px86-TSO model (which we recall is
equivalent to \PTSOS). However, these logics have thus far only been
applied to small examples. Extending these logics to cover a proof by
simulation and a full (manual) proof of correctness of \PMDK, \PMDKT
and \PMDKN would be a significant undertaking, but an interesting
avenue for future work.

\emph{\textbf{Transactional Memory (TM).}}  Several
works have studied the semantics of TM
\cite{john-transactions,hw-transactions-dongol,SI-VMCAI,PSI-ESOP}.
However, our works differ from those in that they do not account for
persistency guarantees and crash consistency.  However, while earlier
works~\cite{SI-VMCAI,PSI-ESOP} merely \emph{propose} a model for weak
isolation (\ie mixing transactional and non-transactional accesses),
\cite{john-transactions,hw-transactions-dongol} formalise the weak
isolation in various hardware and software TM platforms, albeit
without validating their semantics. Moreover, as expected, none of
these works provide a well-established model for nested transactions.

Several approaches to crash consistency have recently been
proposed. For a survey and comparison of techniques (in addition to
transactions) see~\cite{DBLP:journals/csur/BaldassinBCR22}.
OneFile~\cite{DBLP:conf/dsn/RamalheteCFC19},
Romulus~\cite{DBLP:conf/spaa/CorreiaFR18}, and Trinity and
Quadra~\cite{DBLP:conf/ppopp/RamalheteCF21} together describe a set of
algorithms that aim to improve the efficiency of \PMDKTX by reducing
the number of fence instructions. Liu
\etal~\cite{DBLP:conf/asplos/LiuZCQWZR17} present DudeTM, a persistent
TM design that uses a shadow copy of NVM in DRAM, which is is shared
amongst all transactions. Their approach comprises three key steps:
\begin{enumerate*}
	\item an STM or HTM operates over the DRAM
copy, and if the transaction commits, the updates are stored in a
volatile redo log; 
	\item contents of the redo log are made persistent by
copying them to a separate persistent redo log; and
	\item updates in the persistent redo log are copied over to
NVM. 
\end{enumerate*}
Zardoshti \etal~\cite{DBLP:conf/IEEEpact/ZardoshtiZLS19} present an
alternative technique for making STMs persistent by instrumenting STM
code with additional logging and flush instructions. However, none of
these works have defined any formal correctness guarantees, and hence
do not offer any proofs of correctness either. In particular, the role
of allocation and its interaction with reads and writes is generally
unclear.

As well as defining durable opacity, Bila
\etal~\cite{DBLP:conf/forte/BilaDDDSW20} develop a persistent version
of the transactional mutex
lock~\cite{DBLP:conf/europar/DalessandroDSSS10} STM by introducing
explicit undo logging and flush instructions. They then prove this to
be durably opaque using a theorem prover via the \DTMS
specification. More recently, Bila
\etal~\cite{DBLP:journals/corr/abs-2011-15013} have developed a
technique for transforming both an STM and its corresponding opacity
proof by delegating reads/writes to memory locations controlled by the
TM to an abstract library that is later refined to use volatile and
non-volatile memory.  Neither of these works use \PMDKTX, and are over
a sequentially consistent memory model.

\section{Conclusions and Future Work}

Our main contribution is
validating the correctness for \PMDKTX via the development of
declarative (\DDO) and operational (\DDTMS) consistency criteria. We
provide an abstraction of \PMDKTX and show that it satisfies \DDTMS
and hence \DDO by extension. By encoding the above in the \FDR model
checker, we obtain machine-checked bounded proofs for a well-supported
production-ready failure-atomic transactional library.

Additionally, we develop \PMDKT and \PMDKN as two concurrent
extensions of \PMDKTX that are based on existing STM designs, and show
that these also satisfy \DDTMS (and hence \DDO). These are also validated
correct under the \PSC and \PTSOS memory models using \FDR.  Our
philosophy for providing failure atomicity and concurrency is
different from other works. We show that it is straightforward to
achieve failure atomicity by simply reusing \PMDK's allocation, and
replacing reads and writes to locations controlled by the
transactional memory by calls to \PMDKTX read and write operations. This
shows concurrency and failure atomicity can be treated as two
orthogonal sets of concerns, providing a pathway towards a
mix-and-match approach to combining STMs (for concurrency control) and
crash-consistent transactions.

As with most accepted existing transactional models (be it with or
without persistency), we assume \emph{strong isolation}, where each
non-transactional access behaves like a singleton transaction (a
transaction with a single access). That is, even ignoring persistency,
there are no accepted definitions or models for mixing
non-transactional and transactional accesses, and all existing
transactional models (including opacity and serialisability) assume
strong isolation.  Indeed, \PMDK transactions are specifically
designed to be used in a purely transactional setting and are not
meant to be used in combination with non-transactional accesses; \ie
they would have undefined semantics otherwise.

Consequently, as we do not consider mixing transactional code with non-transactional code, RMW (read-modify-write) instructions are irrelevant in our setting. Specifically, as non-transactional access are treated as singleton transactions, RMW instructions are not needed or relevant since they behave as transactions and their atomicity would be guaranteed by the transactional semantics.

Moreover, we do not support nested transactions. Once again, \PMDK transactions are not designed to support nested transactions. Indeed, to the best of our knowledge, there is no accepted correctness notion of nested transactions (be they serialisable, opaque or otherwise).
Were we to consider nested transactions, we would have to invent a new transactional definition, but its practical relevance would be unclear as there are no implementations of it available.

One threat to validity of our work is
that the model checking results are on a small number of transactions,
locations, values, and buffer sizes (see \cref{tab:summary}). However,
we have found that these sizes have been adequate for validating all
of our examples, i.e., when errors are deliberately introduced, FDR
validation fails and counter-examples are automatically
generated. Currently, we do not know whether there is a small model
theorem for durable opacity in general. This is a separate line of
work and a general question that we believe is out of the scope of
this article. Specifically, our focus here is on making PMDK
transactions concurrent, providing a clear specification for PMDK (and
its concurrent variations) with dynamic allocation, and validating
correctness of the results under a realistic memory model.

Edge cases involving more than two concurrent events occur in the
presence of highly relaxed weak
memories~\cite{DBLP:conf/usenix/ChongSW19,hw-transactions-dongol} and
weak isolation. Under SC, Guerraoui \etal~\cite{DBLP:journals/dc/GuerraouiHS10}
verify a small model theorem, showing that only 2 transactions, 2
variables and 2 values are sufficient to model check all possible
conflicts. 
As Px86-TSO
is a small relaxation of persistent SC and we assume strong isolation,
our model checking results provide a high level of assurance.



\bibliography{references,biblio}

\begin{thebibliography}{10}
\providecommand{\url}[1]{\texttt{#1}}
\providecommand{\urlprefix}{URL }
\providecommand{\doi}[1]{https://doi.org/#1}

\bibitem{DBLP:conf/forte/ArmstrongDD17}
Armstrong, A., Dongol, B., Doherty, S.: Proving opacity via linearizability:
  {A} sound and complete method. In: Bouajjani, A., Silva, A. (eds.) FORTE.
  Lecture Notes in Computer Science, vol. 10321, pp. 50--66. Springer (2017).
  \doi{10.1007/978-3-319-60225-7\_4},
  \url{https://doi.org/10.1007/978-3-319-60225-7\_4}

\bibitem{DBLP:conf/wdag/AttiyaGHR14}
Attiya, H., Gotsman, A., Hans, S., Rinetzky, N.: Safety of live transactions in
  transactional memory: {TMS} is necessary and sufficient. In: Kuhn, F. (ed.)
  DISC. Lecture Notes in Computer Science, vol.~8784, pp. 376--390. Springer
  (2014). \doi{10.1007/978-3-662-45174-8\_26},
  \url{https://doi.org/10.1007/978-3-662-45174-8\_26}

\bibitem{DBLP:journals/csur/BaldassinBCR22}
Baldassin, A., Barreto, J., Castro, D., Romano, P.: Persistent memory: {A}
  survey of programming support and implementations. {ACM} Comput. Surv.
  \textbf{54}(7),  152:1--152:37 (2022). \doi{10.1145/3465402},
  \url{https://doi.org/10.1145/3465402}

\bibitem{DBLP:conf/opodis/BerryhillGT15}
Berryhill, R., Golab, W.M., Tripunitara, M.: Robust shared objects for
  non-volatile main memory. In: Anceaume, E., Cachin, C., Potop{-}Butucaru,
  M.G. (eds.) OPODIS. LIPIcs, vol.~46, pp. 20:1--20:17. Schloss Dagstuhl -
  Leibniz-Zentrum f{\"{u}}r Informatik (2015).
  \doi{10.4230/LIPIcs.OPODIS.2015.20},
  \url{https://doi.org/10.4230/LIPIcs.OPODIS.2015.20}

\bibitem{DBLP:journals/corr/abs-2011-15013}
Bila, E., Derrick, J., Doherty, S., Dongol, B., Schellhorn, G., Wehrheim, H.:
  Modularising verification of durable opacity. Log. Methods Comput. Sci.
  \textbf{18}(3) (2022). \doi{10.46298/LMCS-18(3:7)2022},
  \url{https://doi.org/10.46298/lmcs-18(3:7)2022}

\bibitem{DBLP:conf/forte/BilaDDDSW20}
Bila, E., Doherty, S., Dongol, B., Derrick, J., Schellhorn, G., Wehrheim, H.:
  Defining and verifying durable opacity: Correctness for persistent software
  transactional memory. In: Gotsman, A., Sokolova, A. (eds.) FORTE. Lecture
  Notes in Computer Science, vol. 12136, pp. 39--58. Springer (2020).
  \doi{10.1007/978-3-030-50086-3\_3},
  \url{https://doi.org/10.1007/978-3-030-50086-3\_3}

\bibitem{DBLP:conf/esop/BilaDLRW22}
Bila, E.V., Dongol, B., Lahav, O., Raad, A., Wickerson, J.: View-based
  owicki-gries reasoning for persistent {x86-TSO}. In: Sergey, I. (ed.) ESOP.
  Lecture Notes in Computer Science, vol. 13240, pp. 234--261. Springer (2022).
  \doi{10.1007/978-3-030-99336-8\_9},
  \url{https://doi.org/10.1007/978-3-030-99336-8\_9}

\bibitem{SafePM}
Bozdogan, K.K., Stavrakakis, D., Issa, S., Bhatotia, P.: Safepm: a sanitizer
  for persistent memory. In: Bromberg, Y., Kermarrec, A., Kozyrakis, C. (eds.)
  EuroSys. pp. 506--524. {ACM} (2022). \doi{10.1145/3492321.3519574},
  \url{https://doi.org/10.1145/3492321.3519574}

\bibitem{SI-Cerone}
Cerone, A., Gotsman, A.: Analysing snapshot isolation. In: Proceedings of the
  2016 ACM Symposium on Principles of Distributed Computing. pp. 55--64 (2016)

\bibitem{PSI-Cerone}
Cerone, A., Gotsman, A., Yang, H.: Transaction chopping for parallel snapshot
  isolation. In: DISC. vol.~9363, pp. 388--404 (2015)

\bibitem{DBLP:conf/osdi/ChajedTT0KZ21}
Chajed, T., Tassarotti, J., Theng, M., Jung, R., Kaashoek, M.F., Zeldovich, N.:
  Gojournal: a verified, concurrent, crash-safe journaling system. In: Brown,
  A.D., Lorch, J.R. (eds.) {OSDI}. pp. 423--439. {USENIX} Association (2021),
  \url{https://www.usenix.org/conference/osdi21/presentation/chajed}

\bibitem{DBLP:conf/pldi/ChoLRK21}
Cho, K., Lee, S., Raad, A., Kang, J.: Revamping hardware persistency models:
  view-based and axiomatic persistency models for {Intel-x86} and {Armv8}. In:
  Freund, S.N., Yahav, E. (eds.) {PLDI}. pp. 16--31. {ACM} (2021).
  \doi{10.1145/3453483.3454027}, \url{https://doi.org/10.1145/3453483.3454027}

\bibitem{techinsightsblog}
Choe, J.: Review and things to know: Flash memory summit 2022. TechInsights
  (August 2022),
  \url{https://www.techinsights.com/blog/review-and-things-know-flash-memory-summit-2022}

\bibitem{DBLP:conf/pldi/ChongSW18}
Chong, N., Sorensen, T., Wickerson, J.: The semantics of transactions and weak
  memory in x86, power, arm, and {C++}. In: Foster, J.S., Grossman, D. (eds.)
  Proceedings of the 39th {ACM} {SIGPLAN} Conference on Programming Language
  Design and Implementation, {PLDI} 2018, Philadelphia, PA, USA, June 18-22,
  2018. pp. 211--225. {ACM} (2018). \doi{10.1145/3192366.3192373},
  \url{https://doi.org/10.1145/3192366.3192373}

\bibitem{john-transactions}
Chong, N., Sorensen, T., Wickerson, J.: The semantics of transactions and weak
  memory in x86, power, arm, and {C++}. In: Foster, J.S., Grossman, D. (eds.)
  Proceedings of the 39th {ACM} {SIGPLAN} Conference on Programming Language
  Design and Implementation, {PLDI} 2018, Philadelphia, PA, USA, June 18-22,
  2018. pp. 211--225. {ACM} (2018). \doi{10.1145/3192366.3192373},
  \url{https://doi.org/10.1145/3192366.3192373}

\bibitem{DBLP:conf/usenix/ChongSW19}
Chong, N., Sorensen, T., Wickerson, J.: The semantics of transactions and weak
  memory in x86, power, arm, and {C++}. In: Malkhi, D., Tsafrir, D. (eds.) 2019
  {USENIX} Annual Technical Conference, {USENIX} {ATC} 2019, Renton, WA, USA,
  July 10-12, 2019. {USENIX} Association (2019),
  \url{https://www.usenix.org/conference/atc19/presentation/chong}

\bibitem{DBLP:conf/spaa/CorreiaFR18}
Correia, A., Felber, P., Ramalhete, P.: Romulus: Efficient algorithms for
  persistent transactional memory. In: Scheideler, C., Fineman, J.T. (eds.)
  Proceedings of the 30th on Symposium on Parallelism in Algorithms and
  Architectures, {SPAA} 2018, Vienna, Austria, July 16-18, 2018. pp. 271--282.
  {ACM} (2018). \doi{10.1145/3210377.3210392},
  \url{https://doi.org/10.1145/3210377.3210392}

\bibitem{DBLP:conf/europar/DalessandroDSSS10}
Dalessandro, L., Dice, D., Scott, M.L., Shavit, N., Spear, M.F.: Transactional
  mutex locks. In: D'Ambra, P., Guarracino, M.R., Talia, D. (eds.) Euro-Par.
  Lecture Notes in Computer Science, vol.~6272, pp. 2--13. Springer (2010).
  \doi{10.1007/978-3-642-15291-7\_2},
  \url{https://doi.org/10.1007/978-3-642-15291-7\_2}

\bibitem{DBLP:conf/ppopp/DalessandroSS10}
Dalessandro, L., Spear, M.F., Scott, M.L.: Norec: streamlining {STM} by
  abolishing ownership records. In: Govindarajan, R., Padua, D.A., Hall, M.W.
  (eds.) {PPoPP}. pp. 67--78. {ACM} (2010). \doi{10.1145/1693453.1693464},
  \url{https://doi.org/10.1145/1693453.1693464}

\bibitem{DBLP:journals/fac/DerrickDDSTW18}
Derrick, J., Doherty, S., Dongol, B., Schellhorn, G., Travkin, O., Wehrheim,
  H.: Mechanized proofs of opacity: a comparison of two techniques. Formal
  Aspects Comput.  \textbf{30}(5),  597--625 (2018).
  \doi{10.1007/s00165-017-0433-3},
  \url{https://doi.org/10.1007/s00165-017-0433-3}

\bibitem{DBLP:journals/fac/DohertyGLM13}
Doherty, S., Groves, L., Luchangco, V., Moir, M.: Towards formally specifying
  and verifying transactional memory. Formal Aspects Comput.  \textbf{25}(5),
  769--799 (2013). \doi{10.1007/s00165-012-0225-8},
  \url{https://doi.org/10.1007/s00165-012-0225-8}

\bibitem{DBLP:journals/csur/DongolD15}
Dongol, B., Derrick, J.: Verifying linearisability: {A} comparative survey.
  {ACM} Comput. Surv.  \textbf{48}(2),  19:1--19:43 (2015).
  \doi{10.1145/2796550}, \url{https://doi.org/10.1145/2796550}

\bibitem{hw-transactions-dongol}
Dongol, B., Jagadeesan, R., Riely, J.: Transactions in relaxed memory
  architectures. Proc. ACM Program. Lang.  \textbf{2}(POPL) (Dec 2018).
  \doi{10.1145/3158106}, \url{https://doi.org/10.1145/3158106}

\bibitem{DBLP:conf/sefm/DongolL21}
Dongol, B., Le{-}Papin, J.: Checking opacity and durable opacity with {FDR}.
  In: Calinescu, R., Pasareanu, C.S. (eds.) SEFM. Lecture Notes in Computer
  Science, vol. 13085, pp. 222--242. Springer (2021).
  \doi{10.1007/978-3-030-92124-8\_13},
  \url{https://doi.org/10.1007/978-3-030-92124-8\_13}

\bibitem{DBLP:journals/cacm/EswarranGLT76}
Eswaran, K.P., Gray, J., Lorie, R.A., Traiger, I.L.: The notions of consistency
  and predicate locks in a database system. Commun. {ACM}  \textbf{19}(11),
  624--633 (1976). \doi{10.1145/360363.360369},
  \url{https://doi.org/10.1145/360363.360369}

\bibitem{witcher}
Fu, X., Kim, W., Shreepathi, A.P., Ismail, M., Wadkar, S., Lee, D., Min, C.:
  Witcher: Systematic crash consistency testing for non-volatile memory
  key-value stores. In: van Renesse, R., Zeldovich, N. (eds.) {SOSP} '21: {ACM}
  {SIGOPS} 28th Symposium on Operating Systems Principles, Virtual Event /
  Koblenz, Germany, October 26-29, 2021. pp. 100--115. {ACM} (2021).
  \doi{10.1145/3477132.3483556}, \url{https://doi.org/10.1145/3477132.3483556}

\bibitem{fdr}
Gibson-Robinson, T., Armstrong, P., Boulgakov, A., Roscoe, A.: {FDR3 --- A
  Modern Refinement Checker for CSP}. In: Ábrahám, E., Havelund, K. (eds.)
  TACAS. Lecture Notes in Computer Science, vol.~8413, pp. 187--201 (2014)

\bibitem{DBLP:journals/dc/GuerraouiHS10}
Guerraoui, R., Henzinger, T.A., Singh, V.: Model checking transactional
  memories. Distributed Comput.  \textbf{22}(3),  129--145 (2010).
  \doi{10.1007/s00446-009-0092-6},
  \url{https://doi.org/10.1007/s00446-009-0092-6}

\bibitem{DBLP:series/synthesis/2010Guerraoui}
Guerraoui, R., Kapalka, M.: Principles of Transactional Memory. Synthesis
  Lectures on Distributed Computing Theory, Morgan {\&} Claypool Publishers
  (2010). \doi{10.2200/S00253ED1V01Y201009DCT004},
  \url{https://doi.org/10.2200/S00253ED1V01Y201009DCT004}

\bibitem{DBLP:conf/icdcs/GuerraouiL04}
Guerraoui, R., Levy, R.R.: Robust emulations of shared memory in a
  crash-recovery model. In: ICDCS. pp. 400--407. {IEEE} Computer Society
  (2004). \doi{10.1109/ICDCS.2004.1281605},
  \url{https://doi.org/10.1109/ICDCS.2004.1281605}

\bibitem{DBLP:journals/toplas/HerlihyW90}
Herlihy, M., Wing, J.M.: Linearizability: {A} correctness condition for
  concurrent objects. {ACM} Trans. Program. Lang. Syst.  \textbf{12}(3),
  463--492 (1990). \doi{10.1145/78969.78972},
  \url{https://doi.org/10.1145/78969.78972}

\bibitem{DBLP:journals/cacm/Hoare83a}
Hoare, C.A.R.: Communicating sequential processes (reprint). Commun. {ACM}
  \textbf{26}(1),  100--106 (1983). \doi{10.1145/357980.358021},
  \url{https://doi.org/10.1145/357980.358021}

\bibitem{PMDK}
Intel: Persistent memory development kit, {\tt libpmemobj} library (2022),
  \url{https://pmem.io/pmdk/libpmemobj/}

\bibitem{DBLP:conf/wdag/IzraelevitzMS16}
Izraelevitz, J., Mendes, H., Scott, M.L.: Linearizability of persistent memory
  objects under a full-system-crash failure model. In: Gavoille, C., Ilcinkas,
  D. (eds.) DISC. Lecture Notes in Computer Science, vol.~9888, pp. 313--327.
  Springer (2016). \doi{10.1007/978-3-662-53426-7\_23},
  \url{https://doi.org/10.1007/978-3-662-53426-7\_23}

\bibitem{DBLP:journals/pacmpl/KhyzhaL21}
Khyzha, A., Lahav, O.: Taming x86-tso persistency. Proc. {ACM} Program. Lang.
  \textbf{5}({POPL}),  1--29 (2021). \doi{10.1145/3434328},
  \url{https://doi.org/10.1145/3434328}

\bibitem{timestone}
Krishnan, R.M., Kim, J., Mathew, A., Fu, X., Demeri, A., Min, C., Kannan, S.:
  Durable transactional memory can scale with timestone. In: Proceedings of the
  Twenty-Fifth International Conference on Architectural Support for
  Programming Languages and Operating Systems. p. 335–349. ASPLOS '20,
  Association for Computing Machinery, New York, NY, USA (2020).
  \doi{10.1145/3373376.3378483}, \url{https://doi.org/10.1145/3373376.3378483}

\bibitem{LLM12}
Lesani, M., Luchangco, V., Moir, M.: Putting opacity in its place. In: Workshop
  on the Theory of Transactional Memory (2012)

\bibitem{DBLP:conf/sigmod/LienW78}
Lien, Y.E., Weinberger, P.J.: Consistency, concurrency and crash recovery. In:
  Lowenthal, E.I., Dale, N.B. (eds.) {ACM} {SIGMOD} International Conference on
  Management of Data. pp. 9--14. {ACM} (1978). \doi{10.1145/509252.509258},
  \url{https://doi.org/10.1145/509252.509258}

\bibitem{DBLP:conf/asplos/LiuZCQWZR17}
Liu, M., Zhang, M., Chen, K., Qian, X., Wu, Y., Zheng, W., Ren, J.: Dudetm:
  Building durable transactions with decoupling for persistent memory. In:
  Chen, Y., Temam, O., Carter, J. (eds.) Proceedings of the Twenty-Second
  International Conference on Architectural Support for Programming Languages
  and Operating Systems, {ASPLOS} 2017, Xi'an, China, April 8-12, 2017. pp.
  329--343. {ACM} (2017). \doi{10.1145/3037697.3037714},
  \url{https://doi.org/10.1145/3037697.3037714}

\bibitem{DBLP:conf/asplos/0001SWWKK20}
Liu, S., Seemakhupt, K., Wei, Y., Wenisch, T.F., Kolli, A., Khan, S.M.:
  Cross-failure bug detection in persistent memory programs. In: Larus, J.R.,
  Ceze, L., Strauss, K. (eds.) {ASPLOS} '20: Architectural Support for
  Programming Languages and Operating Systems, Lausanne, Switzerland, March
  16-20, 2020. pp. 1187--1202. {ACM} (2020). \doi{10.1145/3373376.3378452},
  \url{https://doi.org/10.1145/3373376.3378452}

\bibitem{XFDetector}
Liu, S., Seemakhupt, K., Wei, Y., Wenisch, T.F., Kolli, A., Khan, S.M.:
  Cross-failure bug detection in persistent memory programs. In: Larus, J.R.,
  Ceze, L., Strauss, K. (eds.) {ASPLOS} '20: Architectural Support for
  Programming Languages and Operating Systems, Lausanne, Switzerland, March
  16-20, 2020. pp. 1187--1202. {ACM} (2020). \doi{10.1145/3373376.3378452},
  \url{https://doi.org/10.1145/3373376.3378452}

\bibitem{DBLP:conf/asplos/0001WZKK19}
Liu, S., Wei, Y., Zhao, J., Kolli, A., Khan, S.M.: Pmtest: {A} fast and
  flexible testing framework for persistent memory programs. In: Bahar, I.,
  Herlihy, M., Witchel, E., Lebeck, A.R. (eds.) Proceedings of the
  Twenty-Fourth International Conference on Architectural Support for
  Programming Languages and Operating Systems, {ASPLOS} 2019, Providence, RI,
  USA, April 13-17, 2019. pp. 411--425. {ACM} (2019).
  \doi{10.1145/3297858.3304015}, \url{https://doi.org/10.1145/3297858.3304015}

\bibitem{DBLP:conf/birthday/Lowe17}
Lowe, G.: Analysing lock-free linearizable datatypes using {CSP}. In:
  Gibson{-}Robinson, T., Hopcroft, P.J., Lazic, R. (eds.) Concurrency,
  Security, and Puzzles - Essays Dedicated to Andrew William Roscoe on the
  Occasion of His 60th Birthday. Lecture Notes in Computer Science, vol. 10160,
  pp. 162--184. Springer (2017). \doi{10.1007/978-3-319-51046-0\_9},
  \url{https://doi.org/10.1007/978-3-319-51046-0\_9}

\bibitem{DBLP:books/mk/Lynch96}
Lynch, N.A.: Distributed Algorithms. Morgan Kaufmann (1996)

\bibitem{kaminotx}
Memaripour, A., Badam, A., Phanishayee, A., Zhou, Y., Alagappan, R., Strauss,
  K., Swanson, S.: Atomic in-place updates for non-volatile main memories with
  kamino-tx. In: Proceedings of the Twelfth European Conference on Computer
  Systems. p. 499–512. EuroSys '17, Association for Computing Machinery, New
  York, NY, USA (2017). \doi{10.1145/3064176.3064215},
  \url{https://doi.org/10.1145/3064176.3064215}

\bibitem{DBLP:conf/tphol/OwensSS09}
Owens, S., Sarkar, S., Sewell, P.: A better x86 memory model: x86-tso. In:
  Berghofer, S., Nipkow, T., Urban, C., Wenzel, M. (eds.) TPHOLs. Lecture Notes
  in Computer Science, vol.~5674, pp. 391--407. Springer (2009).
  \doi{10.1007/978-3-642-03359-9\_27},
  \url{https://doi.org/10.1007/978-3-642-03359-9\_27}

\bibitem{ser}
Papadimitriou, C.H.: The serializability of concurrent database updates. J. ACM
   \textbf{26}(4),  631–653 (oct 1979). \doi{10.1145/322154.322158},
  \url{https://doi.org/10.1145/322154.322158}

\bibitem{DBLP:journals/micro/PelleyCW15}
Pelley, S., Chen, P.M., Wenisch, T.F.: Memory persistency: Semantics for
  byte-addressable nonvolatile memory technologies. {IEEE} Micro
  \textbf{35}(3),  125--131 (2015). \doi{10.1109/MM.2015.46},
  \url{https://doi.org/10.1109/MM.2015.46}

\bibitem{PSI-ESOP}
Raad, A., Lahav, O., Vafeiadis, V.: On parallel snapshot isolation and
  release/acquire consistency. In: Ahmed, A. (ed.) Programming Languages and
  Systems. pp. 940--967. Springer International Publishing, Cham (2018)

\bibitem{SI-VMCAI}
Raad, A., Lahav, O., Vafeiadis, V.: On the semantics of snapshot isolation. In:
  Enea, C., Piskac, R. (eds.) Verification, Model Checking, and Abstract
  Interpretation. pp. 1--23. Springer International Publishing, Cham (2019)

\bibitem{DBLP:journals/pacmpl/RaadLV20}
Raad, A., Lahav, O., Vafeiadis, V.: Persistent owicki-gries reasoning: a
  program logic for reasoning about persistent programs on intel-x86. Proc.
  {ACM} Program. Lang.  \textbf{4}({OOPSLA}),  151:1--151:28 (2020).
  \doi{10.1145/3428219}, \url{https://doi.org/10.1145/3428219}

\bibitem{DBLP:journals/pacmpl/RaadMV22}
Raad, A., Maranget, L., Vafeiadis, V.: Extending intel-x86 consistency and
  persistency: formalising the semantics of intel-x86 memory types and
  non-temporal stores. Proc. {ACM} Program. Lang.  \textbf{6}({POPL}),  1--31
  (2022). \doi{10.1145/3498683}, \url{https://doi.org/10.1145/3498683}

\bibitem{DBLP:journals/pacmpl/RaadWNV20}
Raad, A., Wickerson, J., Neiger, G., Vafeiadis, V.: Persistency semantics of
  the intel-x86 architecture. Proc. {ACM} Program. Lang.  \textbf{4}({POPL}),
  11:1--11:31 (2020). \doi{10.1145/3371079},
  \url{https://doi.org/10.1145/3371079}

\bibitem{DBLP:journals/pacmpl/RaadWV19}
Raad, A., Wickerson, J., Vafeiadis, V.: Weak persistency semantics from the
  ground up: formalising the persistency semantics of armv8 and transactional
  models. Proc. {ACM} Program. Lang.  \textbf{3}({OOPSLA}),  135:1--135:27
  (2019). \doi{10.1145/3360561}, \url{https://doi.org/10.1145/3360561}

\bibitem{DBLP:conf/ppopp/RamalheteCF21}
Ramalhete, P., Correia, A., Felber, P.: Efficient algorithms for persistent
  transactional memory. In: Lee, J., Petrank, E. (eds.) PPoPP '21: 26th {ACM}
  {SIGPLAN} Symposium on Principles and Practice of Parallel Programming,
  Virtual Event, Republic of Korea, February 27- March 3, 2021. pp. 1--15.
  {ACM} (2021). \doi{10.1145/3437801.3441586},
  \url{https://doi.org/10.1145/3437801.3441586}

\bibitem{DBLP:conf/dsn/RamalheteCFC19}
Ramalhete, P., Correia, A., Felber, P., Cohen, N.: Onefile: {A} wait-free
  persistent transactional memory. In: 49th Annual {IEEE/IFIP} International
  Conference on Dependable Systems and Networks, {DSN} 2019, Portland, OR, USA,
  June 24-27, 2019. pp. 151--163. {IEEE} (2019). \doi{10.1109/DSN.2019.00028},
  \url{https://doi.org/10.1109/DSN.2019.00028}

\bibitem{samsungpressrelease}
{Samsung Electronics}: Samsung electronics unveils far-reaching,
  next-generation memory solutions at flash memory summit 2022 (August 2022),
  \url{https://news.samsung.com/global/samsung-electronics-unveils-far-reaching-next-generation-memory-solutions-at-flash-memory-summit-2022}

\bibitem{Scargall2020}
Scargall, S.: Programming Persistent Memory: A Comprehensive Guide for
  Developers. APress (2020). \doi{10.1007/978-1-4842-4932-1_8}

\bibitem{PMDK-Alloc}
Upadhyayula, U.: Introduction to persistent memory allocator and transactions
  (2020),
  \url{https://www.intel.com/content/www/us/en/developer/videos/introduction-to-persistent-memory-allocator-and-transactions.html}

\bibitem{DBLP:conf/IEEEpact/ZardoshtiZLS19}
Zardoshti, P., Zhou, T., Liu, Y., Spear, M.F.: Optimizing persistent memory
  transactions. In: 28th International Conference on Parallel Architectures and
  Compilation Techniques, {PACT} 2019, Seattle, WA, USA, September 23-26, 2019.
  pp. 219--231. {IEEE} (2019). \doi{10.1109/PACT.2019.00025},
  \url{https://doi.org/10.1109/PACT.2019.00025}

\end{thebibliography}

\appendix
\clearpage
\section{\DDTMS Automata: The \DTMS Automata Extended with Allocation}
\label{app:DDTMS_automata}
\begin{minipage}[t]{\columnwidth}
\small
{\bf State variables:}\\
$\mems : \Seq{\Locs \to (\Vals \cup \set{\bot})}$, initially satisfying $\dom(\mems) = \{0\}$ and $initMem(\mems(0))$\\
$\getPC : PCVal$, for each $\txid \in \TXIDs$, initially $\getPC=\pcNotStarted$ for all $\txid \in \TXIDs$\\
$\getBIdx : \nat$ for each $\txid \in \TXIDs$, unconstrained initially\\
$\getRSet: \Locs \pfun \Vals$  for each $\txid \in \TXIDs$, initially empty for all $\txid \in \TXIDs$\\
$\getWSet: \Locs \pfun \Vals$  for each $\txid \in \TXIDs$, initially empty for all $\txid \in \TXIDs$\\
$\getASet \subseteq \Locs$  for each $\txid \in \TXIDs$, initially empty for all $\txid \in \TXIDs$\\
\\
{\bf Transition relation:} \\
  {\em External actions} \\[0.5em]
  \bgroup 
\begin{tabular}{@{}l@{\qquad\qquad }l@{}}
\action{\beginInv{\txid}}
{\getPC = \pcNotStarted }
  {\getPC := \pcBeginPending\\&
  \getBIdx := |\mems| - 1}
&
\action{\abortResp{\txid}}
{\getPC \notin \{\pcNotStarted, \pcReady, \pcChaos  \\
  & \pcCommitResp, \pcCommitted, \pcAborted\}}
  {\getPC := \pcAborted}
 \\[\myrowsep] 
    \action{\beginResp{\txid}}
    {\getPC = \pcBeginPending}
    {\getPC := \pcReady}
  &
    \action{\segFault}
    {pc = \lambda \txid \in \TXIDs.\ \pcChaos}
    {pc := \lambda \txid \in \TXIDs.\ \pcDone}
  \\[\myrowsepa]
\multicolumn{2}{@{}l@{}}{
  \begin{tabular}[t]{@{}l@{}}
  \action{\crashed}
  {\true}
  {\getPC[] := \lambda \txid \in \TXIDs. \sif\ \getPC \notin \{\pcNotStarted, \pcCommitted, \pcChaos\}\  \sthen\ \pcAborted\  \selse \ \getPC \\
  & \mems := \langle last(\mems) \rangle}
\end{tabular}}
\end{tabular} 
\\[1em]
{\it Internal actions} \\[0.5em]
\begin{tabular}{@{}l@{\qquad}l@{}}
\action{\doCommitReadOnly{\txid}}
{\getPC = \pcDoCommit\\&
 \shade{$\getASet = \emptyset$}\\&
 \dom(\getWSet) = \emptyset}
{\getPC := \pcCommitResp}
&
%
%
\action{\doCommitWriter{\txid}}
{\getPC = \pcDoCommit \land  \isValidIdx \txid {last(\mems)} }
{\getPC := \pcCommitResp\\&
  \mems := \mems \cat  {} \\
  & \quad ((\last{\mems} \shade{${}\oplus \{l \to 0 \mid l \in \getASet\}$}) \oplus \getWSet) } 
\\[\myrowsepb]
\action{\doReadFault{\txid}{l, n}}
{\getPC = \pcDoRead(l)\\&
 \shade{$l\nin \getASet \wedge{} $} l\nin \dom(\getWSet)\\&
 \isValidIdx \txid n \wedge \mems(n)(l) = \bot }
  {\shade{$pc := \lambda \txid \in \TXIDs .\ \pcChaos$} }
&
\action{\doCommitWriterFault{\txid}}
{\getPC = \pcDoCommit \land  \isValidIdx \txid {last(\mems)} \\ & 
\shade{$\dom(\getWSet) \not\suq \getASet \cup \setcomp{l}{\last{\mems}(l) \ne \bot}$}
}
{\shade{$pc := \lambda \txid \in \TXIDs .\ \pcChaos$}} 
\\[\myrowsepb]
\action{\doRead{\txid}{l, n}}
{\getPC = \pcDoRead(l)\\&
 \shade{$l\nin \getASet \wedge{} $} l\nin \dom(\getWSet)\\&
 \isValidIdx \txid n \wedge \mems(n)(l) = v \ne \bot }
  {\getRSet := \getRSet\oplus \{l \to v\}\\&
    \getPC := \pcReadResp(v) }
&
\action{\doOwnRead{\txid}{l}}
{\getPC = \pcDoRead(l)\\&
 \shade{$l\in \getASet \vee{}$}  l\in \dom(\getWSet)}
{\sif\ l \in \dom(\getWSet)\ \\&
 \qquad \getPC := \pcReadResp(\getWSet(l))\\&
 \selse\  \getPC := \pcReadResp(0)}
\\[\myrowsepb]
\shade{$\action{\doAlloc{\txid}{l}}
{\getPC = \pcDoAlloc\\&
l \not\in \getASet}
{\getASet := \getASet \cup \{l\}\\&
  \getPC := \pcAllocResp(l)}$}
&
\action{\doWrite{\txid}{l, v}}
{\getPC = \pcDoWrite(l, v)\\&
\shade{$l \in  \getASet \cup \setcomp{l}{\last{\mems}(l) \ne \bot}$}
}
{\getWSet := \getWSet\oplus \{l \to v\}\\&
\getPC := \pcWriteResp}
\end{tabular}
\egroup\smallskip

{\bf where}
$\begin{array}[t]{r@{~~}c@{~~}l}
  \isValidIdx \txid n &\sdef & \getBIdx \leq n < \size{\mems} \wedge
                          \getRSet \subseteq \mems(n) \shade{${} \wedge  \getASet \subseteq \set{ l \mid \mems(n)(l)=\bot$}} 
\end{array}
$
\end{minipage}


\clearpage
\newcommand{\pathends}{\ensuremath{\mathcal E}}
\newcommand{\pathbegins}{\ensuremath{\mathcal B}}
\section{\DDTMS is dynamically Opaque}
\subsection{Instrumented \DDTMS Transition System}
\[\small
\begin{array}{@{} l @{}}
\begin{array}{@{} r @{\hspace{2pt}} l @{}}
	\imems \in \IMSeqs \eqdef & \Seq{\IMems} 
	\quad \text{with} \quad
	\imem \in \IMems \eqdef \Locs \to (\Writes \cup \Allocs \cup \set{\bot})	
	\\
	\ismap \in \ISMaps \eqdef & \TXIDs \rightarrow \IStates \\
	\is \in \IStates \eqdef & \Nats \times \IRSets \times \IWSets \times \IASets \\
	\IRSets \eqdef & 
	\setcomp{
		f: \Locs \pfun \Writes \cup \Allocs
	}{
		\for{x, e} f(x) {=} e \Rightarrow \loc e {=} x
	} \\
	\IWSets \eqdef & 
	\setcomp{
		f: \Locs \pfun \Writes
	}{
		\for{x, w} f(x) {=} w \Rightarrow \loc w {=} x
	}  \\
%
%
%
%
%
	\IASets \eqdef & 
	\setcomp{
		f: \Locs \pfun \Allocs
	}{
		\for{x, a} f(x) {=} a \Rightarrow \loc a {=} x 
	}  \\
	\iact \in \IActions \eqdef & \INoEventActions \cup \IEventActions \cup \CrashMarkers \cup \{\segFault\} \\
	\INoEventActions \eqdef & 	
	\setcomp{
 	 	\ibeginInv, \ireadInv{l}, \iwriteInv{l, v}, \iallocInv
	}{
		l \!\in\! \Locs,
		v \in \Vals
	}\\
	\IEventActions \eqdef & 	
	\setcomp{
		\ibeginResp{b}, \ireadResp{r}, \iwriteResp{w}, \iallocResp{m}, \\
		\icommitInv c, \icommitResp{s}, \iabortResp{a}
	}{
		b \in \Begins, r \in \Reads, w \in \Writes,	m \in \Allocs,\\
		c \in \Commits, s \in \Succs, a \in \Aborts
	} \\
	\getE . : & \IEventActions \rightarrow \EventsType \\
	\getE \iact \eqdef & 
	\begin{cases}
		e & \text{if } \iact {=}  \mathit{res}(e) \\
	 	c & \text{if } \iact {=} \icommitInv c \\
	\end{cases}
\end{array} \\\\
	\isIValidIdx n {\is, \imems} {\is'} \iffdef 
	\begin{array}[t]{@{} l @{}}
		\is \equiv \is' \land 	\is'.\bidx \leq n < \size{\imems} \\
		\land\ \for{x \in \dom(\is'.\aset)} \imems(n)(x)=\bot \\
		\land\ \for{x, e} \is'\!.\rset(x) {=} e \Rightarrow \imems(n)(x) {=} e 
	\end{array}\\\\
	s \equiv s' \iffdef
	\begin{array}[t]{@{} l @{}}
		\is'.\bidx {=} \s.\bidx \land \is'\!.\aset {=} \is.\aset 
		\land \is'\!.\wset {=} \is.\wset \land \dom(\is'\!.\rset) {=} \dom(\is.\rset) \\
		\land\ \for{x, e} \is'\!.\rset(x) {=} e \Rightarrow \exsts{e'} \is.\rset(x) {=} e' \land \wval{e} {=} \wval{e'} \\
	\end{array}\\\\
	\ipath \!\in\! \IPaths \eqdef 
	\Seq{
	\begin{array}{@{} l @{}}
		\{\segFault\} \CrashMarkers \!\cup\! \Begins \cup \Aborts \cup \Writes \cup \Reads \cup \Allocs \cup \Commits \cup \Succs \\
		\cup 
		\setcomp{
			\rmap \txid {\rs} n
		}{
			\txid \in \TXIDs,
			\rs \in \IRSets, n \in \Nats
		} 
		\cup 
		\setcomp{
			\wsuc s n		
		}{
			s \in \Succs, 
			n \in \Nats		
		}\\
		\cup 
		\setcomp{
			\intr r e 
		}{
			r \!\in\! \Reads \land e \!\in\! (\Writes \cup \Allocs)
			\land \loc r {=} \loc e \land \rval r {=} \wval e
		} 
	\end{array}	
	}	\\
	\fresh \ipath e \iffdef e \nin \ipath \land (e \in \Reads \Rightarrow \for{w} \intr e w \nin \ipath) \land (e \in \Succs \Rightarrow \for{n} \wsuc e n \nin \ipath) \\
	\fresh \ipath {(\crash, n)} \iffdef (\crash, n) \nin \ipath
\end{array} 
\]

\subsubsection*{\textbf{Instrumented \DDTMS Global Transitions}}
\[
\infer{
	\pcmap, \ismap, \imems, \ipath \redConf{} \pcmap[\txid \mapsto \pc], \ismap[\txid \mapsto \is], \imems', \ipath'
}{
	\pcmap(\txid), \ismap(\txid), \imems, \ipath \redT{\iact} \pc, \is, \imems', \ipath'
	& \iact \in \INoEventActions
}
\]
\[
\infer{
	\pcmap, \ismap, \imems, \ipath \redConf{} \pcmap[\txid \mapsto \pc], \ismap[\txid \mapsto \is], \imems', \ipath'
}{
	\pcmap(\txid), \ismap(\txid), \imems, \ipath \redT{\iact} \pc, \is, \imems', \ipath'
	& \iact \in \IEventActions
	& \tx{\getE{\iact}} {=} \txid
	& \thrd{\getE{\iact}} {=} \txid
}
\]
\[
\infer{
	\pcmap, \ismap, \imems, \ipath \redConf{} \pcmap', \ismap[\txid \mapsto \is], \imems', \ipath'
}{
	\pcmap(\txid), \ismap(\txid), \imems, \ipath \redT{\segFault} \pcChaos, \is, \imems', \ipath'
	&& \pcmap' = \lambda \txid. \pcChaos
}
\]
\[
\infer{
	\pcmap, \ismap, \imems, \ipath \redConf{} \pcmap', \ismap, \langle last(\imems) \rangle, \ipath \cat [a_1, \cdots, a_n] \cat [(\crash, n)]
}{
\begin{array}{@{} c @{}}
	T {=} \setcomp{\txid}{\pcmap(\txid) \!\notin\! \{\pcNotStarted, \pcCommitted, \pcChaos \}} = \{\txid_1, \cdots, \txid_n\} \\
	\for{i \in \{1, \cdots, n\}} a_i \in \Aborts \land \lTX(a_i) {=} \txid_i \land \fresh \ipath {a_i} \\
	\pcmap' = \lambda \txid.\sif\ \txid \in T \;\;\sthen\;\; \pcAborted \;\;\selse\;\; \pcmap(\txid)
	\qquad	
	\fresh \ipath {(\crash, n)} 
\end{array}	
}
\]

\subsubsection*{\textbf{Instrumented \DDTMS Per-Transaction Transitions}}
\[
\infer{
	\pc, \is, \imems, \ipath \redT{\ibeginInv} \pcBeginPending, \is', \imems, \ipath
}{
	\pc = \pcNotStarted
	&&
	\is' = \is[\bidx \mapsto \size{\imems} {-} 1] 
}
\]
\[
\infer{
	\pc, \is, \imems, \ipath \redT{\ibeginResp b} \pcReady, \is, \imems, \ipath \cat [b]
}{
	\pc = \pcBeginPending
	& b \in \Begins
	& \fresh \ipath b
}
\]
\[
\infer{
	\pc, \is, \imems, \ipath \redT{\iabortResp a} \pcAborted, \is, \imems, \ipath \cat [a]
}{
	\pc \not\in \{\pcNotStarted, \pcReady, \pcCommitted, \pcAborted\}
	& a \in \Aborts
	& \fresh \ipath a
}	
\]
\[
\infer{
	\pc, \is, \imems, \ipath \redT{\ireadInv l} \pcDoRead(l), \is, \imems, \ipath
}{
	\pc = \pcReady
}
\]
%

%
%

%
\[
\infer{
	\pc, \is, \imems, \ipath \redT{\ireadResp{r}} \pcReady, \is'[\rset \mapsto \rs], \imems, \ipath \cat [r, \rmap {\lTX(r)} {\rs} n]
}{
	\begin{array}{@{} c @{}}
		\pc \!=\! \pcDoRead(l)
		\quad l \!\nin\! \setcomp{\loc e} {e \in \is.\aset} \cup \dom(\is.\wset) 
		\quad \isIValidIdx n {\is, \imems} {\is'} 
		\quad r \in \Reads \\
		\imems(n)(l) {=} w \!\neq\! \bot 
		\quad \loc r {=} l
		\quad \wval w {=} \rval r
		\quad \rs \!=\! \is'.\rset \!\oplus\! \{l \!\mapsto\! w\}
		\quad \fresh \ipath r
	\end{array}
}
\]
\[
\infer{
	\pc, \is, \imems, \ipath \redT{\segFault} \pcChaos, \is, \imems, \ipath \cat [\segFault]
}{
	\begin{array}{@{} c @{}}
		\pc \!=\! \pcDoRead(l)
		\quad l \!\nin\! \setcomp{\loc e} {e \in \is.\aset} \cup \dom(\is.\wset) 
		\quad \isIValidIdx n {\is, \imems} {\is'} 
		\quad \imems(n)(l) {=} \bot 
	\end{array}
}
\]
\[
\infer{
	\pc, \is, \imems, \ipath \redT{\ireadResp{r}} \pcReady, \is, \imems, \ipath \cat [\intr r w]
}{
	\pc = \pcDoRead(l)
	&  \is.\wset(l) {=} w
	& r \in \Reads
	& \loc r {=} \loc w
	& \wval w {=} \rval r
	& \fresh \ipath r
}
\]
\[
\infer{
	\pc, \is, \imems, \ipath \redT{\ireadResp{r}} \pcReady, \is, \imems, \ipath \cat [\intr r m]
}{
	\pc = \pcDoRead(l)
	& l \nin \dom(\is.\wset) 
	& \is.\aset(l) {=} m
	& r \in \Reads
	& \loc r {=} \loc m
	& \rval r {=} 0
	& \fresh \ipath r
}
\]
\[
\infer{
	\pc, \is, \imems, \ipath \redT{\iwriteInv{l, v}} \pcDoWrite(l, v), \is, \imems, \ipath
}{
	\pc = \pcReady
}
\]
%

%

%
\[
\infer{
	\pc, \is, \imems, \ipath \redT{\iwriteResp{w}} \pcReady, \is[\wset \mapsto \mathit{ws}], \imems, \ipath \cat [w]
}{
	\pc = \pcDoWrite(l, v)
	& w \in \Writes
	& \loc l {=} w
	& \mathit{ws} = \is.\wset \oplus \{l \!\mapsto\! w\}
	&  \fresh \ipath w
}
\]
\[
\infer{
	\pc, \is, \imems, \ipath \redT{\iallocInv} \pcDoAlloc, \is, \imems, \ipath
}{
	\pc = \pcReady
}
\]
%

%

%
\[
\infer{
	\pc, \is, \imems, \ipath \redT{\iallocResp{m}} \pcReady, \is[\aset \mapsto \mathit{as}], \imems, \ipath \cat [m]
}{
	\pc = \pcDoAlloc
	& l \!\nin\! \dom(\is.\aset)
	& \mathit{as} = \is.\aset \uplus \{l \mapsto m\}
	& m \in \Allocs
	& \loc m {=} l
	& \fresh \ipath m
}
\]
\[
\infer{
	\pc, \is, \imems, \ipath \redT{\icommitInv{c}} \pcDoCommit, \is, \imems, \ipath \cat [c]
}{
	\pc = \pcReady
	& c \in \Commits
	& \fresh \ipath c
}
\]
%

%
%

%
\[
\infer{
	\pc, \is, \imems, \ipath \redT{\icommitResp{s}} \pcCommitted, \is, \imems, \ipath \cat [s]
}{
	\pc = \pcDoCommit
	& \dom(\is.\aset) = \emptyset
	& \dom(\is.\wset) = \emptyset
	& s \in \Succs
	& \fresh \ipath s
}
\]
\[
\infer{
	\pc, \is, \imems, \ipath \redT{\icommitResp{s}} \pcCommitted, \is'\!, \imems'\!, \ipath \cat [\rmap{\lTX(s)} {\is'\!.\rset} {\size{\imems} {-} 1}] \cat [\wsuc s {\size{\imems}}]
}{
	\begin{array}{@{} c @{}}
		\pc = \pcDoCommit
		\qquad \isIValidIdx {\size{\imems} {-} 1} {\is, \imems} {\is'} 
		\qquad s \in \Succs 
		\qquad  \fresh \ipath s \\
		\dom(\is'.\wset) \suq \is'.\aset \cup \setcomp{l}{\last{\imems}(l) \ne \bot} \\
		\imems'\! = \imems \cat ((\last{\imems} \oplus \s'.\aset) \oplus \is'.\wset)
	\end{array}
}
\]
\[
\infer{
	\pc, \is, \imems, \ipath \redT{\segFault} \pcChaos, \is', \imems, \ipath \cat [\segFault]
}{
	\begin{array}{@{} c @{}}
		\pc = \pcDoCommit
		\qquad \isIValidIdx {\size{\imems} {-} 1} {\is, \imems} {\is'} \\
		\dom(\is'.\wset) \not\suq \is'.\aset \cup \setcomp{l}{\last{\imems}(l) \ne \bot} 
	\end{array}
}
\]

\newcommand{\erase}[1]{\ensuremath{\lfloor#1\rfloor}}
\begin{definition}
The \emph{erasure function for instrumented state maps}, $\erase .: \IStates \rightarrow \States$, is defined as $\erase{\ismap} \eqdef \lambda \txid. \erase{\ismap(\txid)}$, where $\erase{(n, \rs, \ws, \as)} \eqdef (n, \rs', \ws', \dom(\as)$ with 
$\rs' \eqdef \lambda \x. \wval{\rs(\x)}$ and $\ws' \eqdef \lambda \x. \wval{\ws(\x)}$.

The \emph{erasure function for instrumented memory sequences}, $\erase .: \IMems \rightarrow \Mems$, is defined as:
\[
\begin{array}{@{} r @{\hspace{2pt}} l @{}}
	\erase{[]} \eqdef & [] \\
	\erase{[\imem] \cat \imems} \eqdef & 
	[\erase{\imem}] \cat \erase{\imems}
\end{array}
\] 
where for all $\x \in \Locs$:
\[
	(\erase \imem) (\x) \eqdef
	\begin{cases}
		\wval e & \text{if } \imem(\x) = e \land e \in \EventsType \\
		\bot & \text{if } \imem(\x) = \bot
	\end{cases}
\] 
\end{definition}

%

%
%

\subsection{Paths Generated by \DDTMS Transitions Are Well-formed}

\begin{definition}
A path $\ipath$ is well-formed, written $\wfipath{\ipath}$, iff the following hold: 
\newcommand{\axlabel}[2]{\tag{\textsc{#1}}\label{#2}}
\begin{align}\small
	& \nodups{\ipath} 
	\axlabel{NoDups}{wf:nodups}\\
	& \segFault \not \plt{\ipath} \segFault 
	\axlabel{UniqueFault}{wf:unique_fault} \\
	& \for{\pent_1, \pent_2 \in \ipath} \tx{\getE{\pent_1}} {=} \tx{\getE{\pent_2}} \Rightarrow \thrd{\getE{\pent_1}} {=} \thrd{\getE{\pent_2}}
	\axlabel{SameTxSameTid}{wf:same_tx_same_thrd}\\
	& \for{b_1, b_2 \in \Begins \cap \ipath} b_1 \ne b_2 \Rightarrow \tx{b_1} \ne \tx{b_2}
	\axlabel{UniqueBegin}{wf:unique_begin}\\
	& \for{c_1, c_2 \in \Commits \cap \ipath}  c_1 \ne c_2 \Rightarrow \lTX(c_1) \ne \lTX(c_2) 
	\axlabel{UniqueCommit}{wf:unique_commit}\\
	& \for{a_1, a_2 \in \Aborts \cap \ipath} a_1 \ne a_2 \Rightarrow \lTX(a_1) \ne \lTX(a_2) 
	\axlabel{UniqueAbort}{wf:unique_abort}\\
	& \for{s_1, s_2 \in \Succs \cap \ipath} s_1 \ne s_2 \Rightarrow \lTX(s_1) \ne \lTX(s_2) 
	\axlabel{UniqueSucc1}{wf:unique_succ1} \\
	& \for{s_1 \in \Succs \cap \ipath} \for{\wsuc {s_2} n \in \ipath} s_1 \ne s_2 \Rightarrow \lTX(s_1) \ne \lTX(s_2) 
	\axlabel{UniqueSucc2}{wf:unique_succ2} \\
	& \for{\wsuc{s_1} {n_1}, \wsuc {s_2} {n_2} \in \ipath} s_1 \ne s_2 \Rightarrow \tx{s_1} \ne \tx{s_2} \land n_1 \ne n_2
	\axlabel{UniqueSucc3}{wf:unique_succ3} \\
	& \for{a \in \Aborts} a \in \ipath \Rightarrow \nexsts{s \in \Succs, n} \tx s {=} \tx a \land (s \in \ipath \lor \wsuc s n \in \ipath)
	\axlabel{AbortOrSucc1}{wf:abort_or_succ1} \\	
	& \for{s \in \Succs, n} (s \in \ipath \lor \wsuc s n \in \ipath) \Rightarrow \nexsts{a \in \Aborts} \tx s {=} \tx a \land a \in \ipath
	\axlabel{AbortOrSucc2}{wf:abort_or_succ2} \\	
	& \for{\intr r e \!\in\! \ipath} 
	\begin{array}[t]{@{} l @{}}
		e \plt{\ipath} \intr r e \land \tx r {=} \tx e \\
		\land \nexsts{w \!\in\! \Writes \!\cup \Allocs} \loc w {=} \loc r \land \tx w {=} \tx  r \land e \plt{\ipath} w \plt{\ipath} \intr r e
	\end{array} 
	\axlabel{InternalRead}{wf:internal_read} \\
	& \for{\txid, \rs, n} \ipath {=} \ipath_1 \cat [\rmap \txid {\rs} n] \cat - \Rightarrow \wfread \txid {\rs} n {\ipath_1} 
	\axlabel{ExternalRead}{wf:external_read} \\
	& \for{r \in \Reads \cap \ipath} \exsts{\rs, n} r \plt{\ipath} \rmap{\tx r} \rs n 
	\axlabel{TotalExternalRead}{wf:total_external_read} \\
	& \for{w \in \Writes, s} \ipath {=} \ipath_1 \cat [w] \cat - \cat [\wsuc s -] \cat - \land \tx w {=} \tx s \Rightarrow \wfwrite w {\ipath_1} 
	\axlabel{WFWrite}{wf:wf_write} \\
	& \begin{array}{@{} l @{}}
		\for{m_1, m_2 \in \Allocs \cap \ipath} \loc{m_1} {=} \loc{m_2} \Rightarrow \\
		\quad \lTX(m_1) \ne \lTX(m_2)
		\land \nexsts{s_1, s_2} \lTX(s_1) {=} \lTX(m_1) \land \wsuc {s_1} - \!\in\! \ipath \land \lTX(s_2) {=} \lTX(m_2) \land \wsuc {s_2} - \!\in\! \ipath \\
	\end{array} 
	\axlabel{UniqueAlloc}{wf:unique_alloc} \\
	& \for{e \in \Reads \cup \Writes \cup \Allocs \cup \Commits \cup \Succs \cup \Aborts} e \in \ipath \Rightarrow \exsts{b \in \Begins} b \plt{\ipath} e \land \lTX(e) {=} \lTX(b)
	\axlabel{BeginFirst1}{wf:begin_first1} \\
	& \for{\intr r w \in \ipath} \exsts{b \in \Begins} b \plt{\ipath} \intr r w \land \lTX(r) {=} \lTX(b) 
	\axlabel{BeginFirst2}{wf:begin_first2} \\
	& \for{\rmap \txid \rs n \in \ipath} \exsts{b \in \Begins} b \plt{\ipath} \rmap \txid \rs n \land \lTX(b) {=} \txid 
	\axlabel{BeginFirst3}{wf:begin_first3} \\
	& \for{\wsuc s n \in \ipath} \exsts{b \in \Begins} b \plt{\ipath} \wsuc s n \land \lTX(b) {=} \lTX(s) 
	\axlabel{BeginFirst4}{wf:begin_first4} \\
	& \for{e \in \Aborts \cup \Succs \cup \Commits, \txid} e \plt{\ipath} \rmap \txid - - \Rightarrow \lTX(e) \ne \txid 
	\axlabel{ReadBeforeEnd1}{wf:read_before_end1} \\
	& \for{s, \txid} \wsuc s - \plt{\ipath} \rmap \txid - - \Rightarrow \lTX(s) \ne \txid 
	\axlabel{ReadBeforeEnd2}{wf:read_before_end2} \\
	& \for{e \in \Aborts \cup \Succs \cup \Commits, r} e \plt{\ipath} \intr r - \Rightarrow \lTX(e) \ne  \lTX(r) 
	\axlabel{IReadBeforeEnd1}{wf:iread_before_end1} \\
	& \for{s, r} \wsuc s - \plt{\ipath} \intr r - \Rightarrow \lTX(s) \ne  \lTX(r) 
	\axlabel{IReadBeforeEnd2}{wf:iread_before_end2} \\
	& \for{e \in \Aborts \cup \Succs, e' \in \EventsType} e \plt{\ipath} e' \Rightarrow \lTX(e) \ne \lTX(e') 
	\axlabel{EventBeforeEnd1}{wf:event_before_end1} \\
	& \for{e' \in \EventsType, s, n} \wsuc s n \plt{\ipath} e' \Rightarrow \lTX(s) \ne \lTX(e') 
	\axlabel{EventBeforeEnd2}{wf:event_before_end2} \\
	& \for{c \in \Commits, e \in \EventsType} c \plt{\ipath} e \land \lTX(c) {=} \lTX(e)  \Rightarrow e \in \Aborts \cup \Succs 
	\axlabel{OnlyEndAfterCommit}{wf:only_end_after_commit}  \\
	& \for{s \!\in\! \Succs \cap \ipath} \exsts{c \!\in\! \Commits} 
	\begin{array}[t]{@{} l @{}}				
		c \plt{\ipath} s \land \tx c {=} \tx s \\
		\land \nexsts{e \!\in\! \Writes \cup \Allocs}\! \tx e {=} \tx c \land e \plt{\ipath} c 
	\end{array} 
	\axlabel{EndReadOnly}{wf:end_read_only} \\
	& \for{\wsuc s n \!\in\! \ipath} \exsts{c \!\in\! \Commits, w \!\in\! \Writes \cup \Allocs} 
	\begin{array}[t]{@{} l @{}}		
		w \plt{\ipath} c \plt{\ipath} \wsuc s n \land \tx w {=} \tx c {=} \tx s 
	\end{array} 
	\axlabel{EndWriteAlloc}{wf:end_write_alloc} \\
	& \for{s_1, n_1, s_2, n_2} \wsuc {s_1} {n_1} \plt{\ipath} \wsuc{s_2}{n_2} \Rightarrow n_1 < n_2 
	\axlabel{EndWriteAllocInOrder1}{wf:end_write_alloc_in_order1} \\
	& \for{n_1, n_2} \wsuc {-} {n_2} \in \ipath \land n_1 < n_2 \Rightarrow \wsuc - {n_1} \in \ipath 
	\axlabel{EndWriteAllocInOrder2}{wf:end_write_alloc_in_order2} \\
	& \for{n_1, n_2} \rmap - - {n_2} \in \ipath \land n_1 \leq n_2 \Rightarrow \wsuc - {n_1} \in \ipath 
	\axlabel{ExtReadContinuous}{wf:ext_read_continuous} \\
	& \for{n_1, n_2} \rmap - - {n_1} \plt \ipath \wsuc - {n_2} \Rightarrow n_1 < n_2 
	\axlabel{ExtReadInOrder1}{wf:ext_read_in_order1} \\
	& \for{n_1, n_2} \wsuc - {n_1} \plt \ipath \rmap - - {n_2} \Rightarrow n_1 \leq n_2 
	\axlabel{ExtReadInOrder2}{wf:ext_read_in_order2} \\
	& \for{s, n} \wsuc s n \in \ipath \Rightarrow \ipath = - \cat [\rmap {\tx s} - {n {-} 1}] \cat [\wsuc s n] \cat - 
	\axlabel{ExtReadBeforeSucc}{wf:ext_read_before_succ} 
\end{align}
where:
{\small
\begin{align*}
	\nodups \ipath \iffdef 
	& \for{e \in \EventsType, \ipath_1, \ipath_2} 
		\ipath {=} \ipath_1 \cat [e] \cat \ipath_2 \Rightarrow \fresh {\ipath_1 \cat \ipath_2} e \\*
	& \land\,\for{r \in \Reads, \ipath_1, \ipath_2} \ipath {=} \ipath_1 \cat [\intr r -] \cat \ipath_2 \Rightarrow \fresh {\ipath_1 \cat \ipath_2} r \\*	
	& \land\,\for{s \in \Succs, \ipath_1, \ipath_2} \ipath {=} \ipath_1 \cat [\wsuc s -] \cat \ipath_2 \Rightarrow \fresh {\ipath_1 \cat \ipath_2} s \\		
	\getE{\pent} \eqdef &
	\begin{cases}
		\pent & \text{if } \pent \in \EventsType \\
		e & \text{if } \pent {=} \intr e - \lor \pent {=} \wsuc e - \\
		\text{undefined} & \text{otherwise}
	\end{cases}\\
	A \plt{\ipath} B \iffdef & \ipath = - \cat [A] \cat - [B] \cat - \\
	\wfread \txid {\rs} n \ipath \iffdef & 
	\wsuc - n \in \ipath
	\land \for{r \!\in\! \Reads \cap \ipath} \lTX(r) {=} \txid \Rightarrow \loc r \!\in\! \dom(\rs) \\*
	& \land\ \for{l, e} {\rs}(l) {=} e \Rightarrow 
	\begin{array}[t]{@{} l @{}}
		\neg\exsts{e' \in \Writes \cup \Allocs} \loc{e'} {=} l \land \lTX(e') {=} \txid \land e' \in \ipath \\
		\land\ \exsts{s, m} \tx s {=} \tx e \ne \txid \land e \plt{\ipath} \wsuc s m \land m \leq n \\
		\land\ \neg\exsts{e' \in \Writes \cup \Allocs, s', m'} \\
			\qquad \loc{e'} {=} l \land \lTX(e') {=} \lTX(s') \land e' \plt{\ipath} \wsuc {s'} {m'} \land m  < m' 
	\end{array}	\\
	\wfwrite w \ipath \iffdef & 
	\exsts{m \in \Allocs \cap \ipath} \loc{m} {=} \loc w \land 
	(\tx m {\ne} \tx w \Rightarrow \exsts{s \in \Succs} \tx s {=} \tx m \land m \plt{\ipath} \wsuc s -)	
\end{align*}
}
\end{definition}


\begin{definition}
\[
\begin{array}{@{} r @{\hspace{2pt}} l @{}}
	\pcmap_0 \eqdef 
	& \lambda\txid. \pcNotStarted \\
%
	\ismap_0 \eqdef 
	& \lambda\txid. \is_0 
	\quad \text{with} \quad 
	\is_0 \eqdef (0, \emptyset, \emptyset, \emptyset) \\
	\imems_0 \eqdef 
	& [\imem_0]
	\quad \text{with} \quad
	\imem_0 \eqdef \lambda x. \bot \\
%
\end{array}
\]
\end{definition}

\begin{proposition}
\label{lem:wf_path}
For all $\imems, \pcmap, \smap, \ipath$, 
if $\pcmap_0, \ismap_0, \imems_0, \emptyipath \redConf{}^* \pcmap, \ismap, \imems, \ipath$, 
then $\wfipath \ipath$.
\end{proposition}

\begin{proposition}
\label{lem:wf_path_prefix}
For all paths $\ipath, \ipath'$, if $\wfipath \ipath$ holds and $\ipath'$ is a prefix of $\ipath$, then $\wfipath{\ipath'}$ also holds.
\end{proposition}

\begin{proposition}
\label{lem:wf_path_crashes_removed}
For all paths $\ipath, \ipath'$, if $\wfipath \ipath$ holds and $\ipath' = \ipath \setminus \CrashMarkers$, then $\wfipath{\ipath'}$ also holds, where:
\[
	[] \setminus \CrashMarkers \eqdef []
	\qqqquad
	(\pent:: \ipath) \setminus \CrashMarkers \eqdef
	\begin{cases}
		\ipath & \text{if } \pent = (\crash, -) \\
		\pent:: (\ipath \setminus \CrashMarkers) & \text{otherwise}
	\end{cases}
\]
\end{proposition}

\begin{lemma}
\label{lem:ddtms_refines_instrumented_ddtms}
Let $\smap_0 \eqdef \lambda \txid. (0, \emptyset, \emptyset, \emptyset)$ and $\mems_0 \eqdef [\mem_0]$ with $\mem_0 \eqdef \lambda \x. \bot$.
For all $\pcmap, \smap, \mems$ and executions $\sigma$,
if $\sigma$ denotes the execution $\pcmap_0, \smap_0, \mems_0 \redConf{}^* \pcmap, \smap, \mems $, 
then there exists $\ismap, \imems, \ipath$ such that $\pcmap_0, \ismap_0, \imems_0, [] \redConf{}^* \pcmap, \ismap, \imems, \ipath$ and $\mathit{OH}_\sigma = \getHist \ipath$, 
where $\getHist \ipath \eqdef (\Events, \teo)$ with
$\Events \eqdef \setcomp{e \in \EventsType}{\exsts{\pent \!\in\! \ipath} \getE \pent {=} e}$
and $\teo \eqdef \setcomp{(e_1, e_2)}{\exsts{\pent_1, \pent_2} \pent_1 \plt{\ipath} \pent_2 \land \getE{\pent_1} {=} e_1 \land \getE{\pent_2} {=} e_2}$.
\end{lemma}

\begin{proof}
Follows from straightforward induction on the structure of \DDTMS transitions and the definition of instrumented \DDTMS transitions.
\end{proof}


\subsection{Instrumented \DDTMS Transitions Refine Dynamic Durable Opacity}

\newcommand{\swriters}[1][\ipath]{\ensuremath{\mathsf{SWriters}_{#1}}}
\newcommand{\ereadonly}[1][\ipath]{\ensuremath{\mathsf{EROnly}_{#1}}}
\newcommand{\esreaders}[1][\ipath]{\ensuremath{\mathsf{ESReaders}_{#1}}}
\newcommand{\ereaders}[1][\ipath]{\ensuremath{\mathsf{EReaders}_{#1}}}
\newcommand{\pawriters}[1][\ipath]{\ensuremath{\mathsf{PAWriters}_{#1}}}
\begin{definition}
\[
	\peo(\ipath) \eqdef 
	\begin{array}[t]{@{} l @{}}
		\setcomp{
			(\txid_1,	\txid_2)
		}{
			(\txid_1, n_1), 	(\txid_2, n_2) \in \swriters \land n_1 < n_2
		} \\
		\cup
		\setcomp{
			(\txid_1,	\txid_2)
		}{
			(\txid_1, n_1) \in \swriters 
			\land	(\txid_2, n_2) \in \esreaders 
			\land n_1 \leq n_2
		} \\
		\cup
		\setcomp{
			(\txid_1,	\txid_2)
		}{
			(\txid_1, n_1) \in \esreaders 
			\land	(\txid_2, n_2) \in \swriters 
			\land n_1 < n_2
		} \\
		\cup
		\setcomp{
			(\txid_1,	\txid_2)
		}{
			(\txid_1, n_1) \in \swriters 
			\land	(\txid_2, n_2) \in \ereaders 
			\land n_1 \leq n_2
		} \\
		\cup
		\setcomp{
			(\txid_1,	\txid_2)
		}{
			(\txid_1, n_1) \in \ereaders 
			\land (\txid_2, n_2) \in \swriters 
			\land n_1 < n_2
		} 
	\end{array}
\]
with
\[\small
\begin{array}{@{} r @{\hspace{2pt}} l@{}}
	\swriters[\ipath] \eqdef & 
	\setcomp{
		(\lTX(s), n)	
	}{
		\wsuc s n \in \ipath
	} \\
	\esreaders[\ipath] \eqdef & 
	\setcomp{
		(\txid, n)
	}{
		\ipath {=} - \cat [\rmap \txid - n] \cat \ipath'
		\land \for{\rs, m} \rmap \txid \rs m \nin \ipath'
		\land \for m (\txid, m) \nin \swriters[\ipath] \\
		\land\ \exsts{s \in \Succs \cap \ipath'} \tx s = \txid
	} \\
	\ereaders[\ipath] \eqdef & 
	\setcomp{
		(\txid, n)
	}{
		\ipath {=} - \cat [\rmap \txid - n] \cat \ipath'
		\land \for{\rs, m} \rmap \txid \rs m \nin \ipath' \\
		\land \for m (\txid, m) \nin \swriters[\ipath] \cup \ereaders[\ipath] 
	} 
\end{array}
\]
\end{definition}

\begin{lemma}
\label{lem:aux}
For all paths $\ipath$ and executions $G$, if $\wfipath \ipath$ holds and $G {=} \makeG \ipath$,
then:
\begin{itemize}
	\item For all $w \in G.\Writes \cap \SSet$, there exists $n$ such that $(\tx w, n) \in \swriters$.
\end{itemize}
\end{lemma}

\begin{proof}
Pick arbitrary $w \in G.\Writes \cap \SSet$ and let $\tx w {=} \txid$. 
From $G {=} \makeG \ipath$ and the definition of $G.\Events$ we then know $w \!\in\! \ipath$. 
As $w  \in G.\SSet$, we know there exist $s \in G.\Succs$ such that $\tx s {=} \tx w$.
Consequently, given the definition of $G.\Events$, either i) $s \in \ipath$ or ii) there exists $n$ such that $\wsuc s n \in \ipath$. As such, since $w \in \ipath$, from $\wfipath \ipath$, \eqref{wf:event_before_end1}, \eqref{wf:event_before_end2} we know either i) $w \plt{\ipath} s$ or ii) there exists $n$ such that $w \plt{\ipath} \wsuc s n$. 
Therefore, from $\wfipath \ipath$ and \eqref{wf:end_read_only} we know case (i) does not arise and thus there exists $n$ such that $w \plt{\ipath} \wsuc s n$.
Hence by definition we have $(\txid, n) \in \swriters$.
\end{proof}

\begin{lemma}
\label{lem:peo_transitive}
For all paths, $\peo(\ipath)$ is transitive.
\end{lemma}

\begin{proof}
Follows from the definition of $\peo(\ipath)$ and the transitivity of $<$ and $\leq$ on natural numbers. 
\end{proof}

%

\begin{lemma}
\label{lem:peo_irreflexive}
For all paths $\ipath$, if $\wfipath \ipath$ holds, then $\peo(\ipath)$ is irreflexive.
\end{lemma}

\begin{proof}
Follows from $\wfipath \ipath$, the definition of $\peo(\ipath)$, the fact that the sets $\swriters[\ipath]$, $\esreaders[\ipath]$ and $\ereaders[\ipath]$ are pairwise disjoint, the irreflexivity of $\plt{\ipath}$ and the irreflexivity of $<$ on natural numbers. 
\end{proof}

\begin{lemma}
\label{lem:peo_strict_order}
For all paths $\ipath$, if $\wfipath \ipath$ holds, then $\peo(\ipath)$ is a strict order.
\end{lemma}

\begin{proof}
Follows from \cref{lem:peo_transitive} and \cref{lem:peo_irreflexive}.
\end{proof}

\newcommand{\beginevents}{\ensuremath{\mathcal B}}
\begin{definition}
%
The $\teo(\ipath)$ is defined as an extension of $\peo'(\ipath)$ to a strict total order on $\setcomp{\tx b}{b \in \Begins \cap \ipath}$, 
where $\peo'(\ipath) \eqdef \transC{(\peo(\ipath) \cup \clo(\ipath))}$ and 
\[
	\clo(\ipath) \eqdef 
	\setcomp{
		(\tx e, \tx b)	
	}{
		 e \in \Succs \cup \Aborts \land  b \in \Begins 
		\land (e \plt{\ipath} b \lor \wsuc e - \plt{\ipath} b)
	}
\]
\end{definition}

\begin{lemma}
Given a path $\ipath$, if $\wfipath \ipath$ holds then $\clo(\ipath)$ is a strict partial order. 
\end{lemma}

\begin{proof}

\end{proof}
\begin{lemma}
Given a path $\ipath$, if $\wfipath \ipath$ holds then $\peo'(\ipath)$ is a strict partial order.
\end{lemma}

\begin{lemma}
\label{lem:teo_strict_total_order}
Given a path $\ipath$, if $\wfipath \ipath$ holds then $\teo(\ipath)$ is a strict total order on $\setcomp{\tx b}{b \in \Begins \cap \ipath}$.  
\end{lemma}

\begin{proof}
Follows from t he fact that $\peo'(\ipath)$ is a strict partial order and that $\teo(\ipath)$ is an extension of $\peo'(\ipath)$ to a strict total order on $\setcomp{\tx b}{b \in \Begins \cap \ipath}$.
\end{proof}

\newcommand{\MO}{\ensuremath{\makerel{MO}}}
\newcommand{\NVO}{\ensuremath{\makerel{NVO}}}
\begin{definition}
\[
	\makeG{\ipath} \eqdef (\Events, \po, \clo, \rf, \mo)
\]
where
\[\small
\begin{array}{@{}r @{\hspace{2pt}} l@{}}
	\Events \eqdef & 
	\setcomp{e \in \EventsType}{\exsts{\pent \in \ipath} \getE \pent {=} e} \\
%
%
	\po \eqdef & 
	\setcomp{
		(e_1, e_2) \in \Events \times \Events
	}{
		\exsts{\pent_1, \pent_2} 
		\pent_1 \plt{\ipath} \pent_2
		\land \getE {\pent_1} {=} e_1
		\land \getE{\pent_2} {=} e_2
		\land \tx{e_1} {=} \tx{e_2} 
	} \\
	\clo \eqdef & 
	\setcomp{
		(c, d) \in \Events \times \Events
	}{
		\exsts{e \in \Succs \cup \Aborts, b \in \Begins} 
			\tx c {=} \tx e \land \tx d {=} \tx b \land (e \plt{\ipath} b \lor \wsuc e - \plt{\ipath} b)
	} \\
	\rf \eqdef & 
	\mathsf{IRF} \cup \bigcup_{\txid \in \TXIDs} \mathsf{ERF}_{\txid}  \\
	\mathsf{IRF} \eqdef & 
	\cup
	\setcomp{
		(w, r)	 \in \Events \times \Events
	}{
		\intr r w \in \ipath	
	} \\
	\mathsf{ERF}_{\txid} \eqdef & 
	\setcomp{
		(w, r)	 \in \Events \times \Events
	}{
		r \!\in\! \ipath \cap \Reads \land \tx r {=} \txid \land 
		\exsts{\rs, \ipath'} 
		\begin{array}[t]{@{} l @{}}
			\ipath {=} - \cat [\rmap \txid \rs -] \cat \ipath'\!
			\land \rs(\loc r) {=} w \\
			\land\ \for{\rs'\!, n} \rmap \txid {\rs'\!} n \nin \ipath' \!
		\end{array}
	} 
\end{array}
\]
and $\mo \eqdef \bigcup_{\x \in \Locs} \mo_\x$ 
with 
\[
	\mo_\x \eqdef 
	\begin{array}[t]{@{} l @{}}
		\setcomp{
			(e_1, e_2) \in (\Events \cap (\Writes \cup \Allocs))^2
		}{
			\tx{e_1} {=} \tx{e_2} \land \loc{e_1} {=} \loc{e_2}  {=} \x \land e_1 \plt{\ipath} e_2
		} \\
		\cup 
		\setcomp{
			(e_1, e_2) \in (\Events \cap (\Writes \cup \Allocs))^2
		}{
			\loc{e_1} {=} \loc{e_2} {=} \x \land \tx{e_1} \relarrow{\teo(\ipath)} \tx{e_2}
		} 
	\end{array}
\]

\end{definition}

\begin{lemma}
\label{lem:rf}
For all paths $\ipath, G, w, r$, if $\wfipath \ipath$ holds, $G {=} \makeG{\ipath}$ and $(w, r) \in G.\rf$, then:
\[
	\big(\tx w {=} \tx r \land (w, r) \in \mathsf{IRF}\big)
	\lor
	\big(\tx w {\ne} \tx r \land (w, r) \in \mathsf{ERF}_{\tx r}\big)
\]
\end{lemma}

\begin{proof}
Pick an arbitrary $\ipath, G, w, r$ such that $\wfipath \ipath$ holds, $G {=} \makeG{\ipath}$ and $(w, r) \in G.\rf$.
From the definition of $G.\rf$ there are two cases to consider: 
1)  $(w, r) \in \mathsf{IRF}$; or
2) there exists $\txid$ such that $(w, r) \in \mathsf{ERF}_\txid$. 
In case (1), from the definition of $\mathsf{IRF}$ we know $\intr r w \in \ipath$, and thus from $\wfipath \ipath$ and \eqref{wf:internal_read} we know $\tx w {=} \tx r$, as required.
In case (2), from the definition of $\mathsf{ERF}_\txid$ we know 
$\tx r {=} \txid $ and there exists $n, \rs, \ipath_1, \ipath_2$ such that $\ipath {=} \ipath_1 \cat [\rmap \txid \rs n] \cat \ipath_2$ and $\rs(\loc r) {=} w$.
As such, from $\wfipath \ipath$ and \eqref{wf:external_read} we have $\wfread \txid \rs n {\ipath_1}$. 
Consequently, since $\rs(\loc r) {=} w$, from $\wfread \txid \rs n {\ipath_1}$ we know $\tx w \ne \txid$. 
That is, as $\tx r {=} \txid$ and  $(w, r) \in \mathsf{ERF}_\txid$ we have  $(w, r) \in \mathsf{ERF}_{\tx r}$ and $\tx w \ne \tx r$, as required.
\end{proof}

\begin{lemma}
\label{lem:clo_external}
For all $\ipath$ and $G$, if $\wfipath \ipath$ holds and $G {=} \makeG{\ipath}$, then $G.\clo = G.\tex\clo$. 
\end{lemma}

\begin{proof}
Pick arbitrary $\ipath$ and $G$ such that $\wfipath \ipath$ and $G {=} \makeG{\ipath}$. 
It then suffices to show 
$G.\tin\clo = \emptyset$. 
Let us proceed by contradiction and assume that there exist $(c, d) \in G.\tin\clo$, \ie $\tx c = \tx d$. 
From the definition of $G.\clo$ we then know there exists $e \in \Succs \cup \Aborts, b \in \Begins, n$ such that $\tx e = \tx c$, $\tx b = \tx d$ and $(e \plt{\ipath} b \lor \wsuc e n \plt{\ipath} b)$ and thus (as $\tx c = \tx d$) $\tx e = \tx b$. 
However,  $(e \plt{\ipath} b \lor \wsuc e n \plt{\ipath} b)$ contradicts \ref{wf:event_before_end1} and \ref{wf:event_before_end2} as $\wfipath \ipath$ holds.
\end{proof}

\begin{lemma}
\label{lem:execution}
For all $\ipath$ and $G$, if $\wfipath \ipath$ holds and $G {=} \makeG{\ipath}$, then $G$ is an execution according to \cref{def:volatile_executions}.
\end{lemma}

\begin{proof}
Pick arbitrary $\ipath$ and $G = (\Events, \po, \clo, \rf, \mo)$ such that $\wfipath \ipath$ and $G {=} \makeG{\ipath}$.
We are required to show: 
\begin{enumerate}
	\item $\Events$ denotes a set of in $\EventsType$.
	\label{goal:ex_events}
	\item $\po \suq \Events \times \Events$ and is a disjoint 	union of strict total orders, each ordering the events of one thread.
	\label{goal:ex_po}
	\item $\clo \suq \Events \times \Events$, $\st; \clo; \st \suq \clo \setminus \st$ and $\po \setminus \st \suq \clo$
	\label{goal:ex_clo}
	\item $\rf \suq (\Allocs \cup \Writes) \times \Reads$.
	\label{goal:ex_rf_type}
	\item $\for{a, b} (a, b) \in \rf \Rightarrow \loc a {=} \loc b \land \wval a {=} \rval b$.
	\label{goal:ex_rf_vals}
	\item $\rf$ is functional on its range. 
	\label{goal:ex_rf_functional}
	\item $\rf$ is total on its range. 
	\label{goal:ex_rf_total}
	\item $\mo \suq \Events \times \Events$ and is the disjoint union of relations $\{\mo_\x\}_{\x \in \Locs}$.
	\label{goal:ex_mo_type}
	\item Each $\mo_\x$ is a strict total order on $\Allocs_\x \cup \Writes_\x$.
	\label{goal:ex_mo_order}
%
%
\end{enumerate}
Parts \eqref{goal:ex_events}, \eqref{goal:ex_po}, \eqref{goal:ex_rf_type}, \eqref{goal:ex_rf_vals} and \eqref{goal:ex_mo_type} 
follow immediately from the definition of $\makeG \ipath$ and the types of $\ipath$ entries and $\IRSets$. \\

\noindent\textbf{RTS. \eqref{goal:ex_clo}}\\*
Follows immediately from the definitions of $\po$, $\clo$ and \cref{lem:clo_external}.\\

\noindent\textbf{RTS. \eqref{goal:ex_rf_functional}}\\*
Pick an arbitrary $(w, r) \in G.\rf$. From the definition of $G.\Events$, $\wfipath \ipath$ and \eqref{wf:nodups} we know either 
i) $\intr r w \in \ipath$ and $r \nin \ipath$ and $\for{w' \ne w} \intr r {w'} \nin \ipath$; or
ii) $r \in \ipath$ and $\for{w} \intr r {w} \nin \ipath$. 
In case (i), from the definition of $G.\rf$ we know $\for{w' \ne w} \intr r {w'} \nin \ipath$, as required.
In case (ii), let us proceed by contradiction and assume there exist $w_1, w_2$ such that $(w_1, r), (w_2, r) \in \rf$ and $w_1 \ne w_2$. 
From the definition of $\rf$ and the assumption of the case we then know there exists $\rs, n, \ipath_2$ such that $\ipath {=}  \cat [\rmap {\tx r} \rs n] \cat \ipath_2$, $\for{\rs', m} \rmap {\tx r} {\rs'} m \nin \ipath_2$, $\rs(\loc r) {=} w_1$ and $\rs(\loc r) {=} w_2$. As $\rs \in \IRSets$ is a function, $\rs(\loc r) {=} w_1$ and $\rs(\loc r) {=} w_2$, we then have $w_1 {=} w_2$, leading to a contradiction since we assumed $w_1 \ne w_2$.\\

\noindent\textbf{RTS. \eqref{goal:ex_rf_total}}\\*
Pick an arbitrary $r \in G.\Reads$. From the definition of $G.\Events$, $\wfipath \ipath$ and \eqref{wf:nodups} we know either 
i) there exists $w$ such that $\intr r w \in \ipath$, $r \nin \ipath$ and $\for{w' \ne w} \intr r {w'} \nin \ipath$; or
ii) $r \in \ipath$ and $\for{w} \intr r {w} \nin \ipath$. 
In case (i), from the definition of $G.\rf$ we know $(w, r) \in \rf$, as required.
In case (ii), from $\wfipath \ipath$ and \eqref{wf:total_external_read} we know there exists $r \plt \ipath \rmap {\tx r} - -$. 
As such, since there is at least one $\rmap {\tx r} - -$ entry in $\ipath$, we know there exists $\rs, n, \ipath_1, \ipath_2$ such that $\ipath {=} \ipath_1 \cat [\rmap {\tx r} \rs n] \cat \ipath_2$, $\for{\rs', m} \rmap {\tx r} {\rs'} m \nin \ipath_2$ and $r \in \ipath_1$.
Consequently, from $\wfipath \ipath$ and \eqref{wf:external_read} we know $\wfread {\tx r} \rs n {\ipath_1}$ holds.
As such, from the definition of $\wfread {\tx r} \rs n {\ipath_1}$ and since $r \in \ipath_1$ we know $\loc r \in \dom(\rs)$, \ie there exists $w$ such that $\rs(\loc r) = w$. 
Consequently, since $r \in G.\Reads$, $r \in \ipath$, $\ipath {=} \ipath_1 \cat [\rmap {\tx r} \rs n] \cat \ipath_2$, $\for{\rs', m} \rmap {\tx r} {\rs'} m \nin \ipath_2$ and $\rs(\loc r) = w$, from the definition of $\rf$ we have $(w, r) \in \rf$, as required.\\

\noindent\textbf{RTS. \eqref{goal:ex_mo_order}}\\*
Follows from the definition of $\mo_\x$ in $\makeG \ipath$, and the fact that $\teo(\ipath)$ is a strict total order on $\setcomp{\tx b}{b \in \Begins \cap \ipath}$ (\cref{lem:teo_strict_total_order}).\\
%
%
%
\end{proof}

\begin{lemma}
\label{lem:wf_execution}
For all $\ipath$ and $G$, if $\wfipath \ipath$ holds and $G {=} \makeG{\ipath}$, then $G$ is a well-formed execution according to \cref{def:wf_execution}.
\end{lemma}

\begin{proof}
Pick arbitrary $\ipath$ and $G = (\Events, \po, \rf, \mo)$ such that $\wfipath \ipath$ and $G {=} \makeG{\ipath}$.
We are required to show: 
\begin{align}
	& \st \suq \po \cup \inv{\po}
	\label{goal:wf_same_tx_same_tid} \\
	& 	\st \cap (\tex{\po}; \tex{\po}) = \emptyset
	\label{goal:wf_tx_contiguous} \\
	& \size{\Events_{\txid} \cap \Begins} {=} 1 \land \be \not\in \rng(\poI)
	\label{goal:wf_begin} \\
	& \size{\Events_{\txid} \cap \Aborts} \leq 1 \land \size{\Events_{\txid} \cap \Commits} \leq 1 \land \size{\Events_{\txid} \cap \Succs} \leq 1 
	\label{goal:wf_single_abort_commit_success} \\
	& \abe \not\in \dom(\poI) \land \se \not\in \dom(\poI)
	\label{goal:wf_abort_succ_maximal} \\
	& \rng([C]; \po \cap \st) \suq \Aborts \cup \Succs 
	\label{goal:wf_only_success_or_abort_after_commit} \\
	& \se \in \Events \Rightarrow (\ce, \se) \in \imm\po
	\label{goal:wf_commit_succ_adjacent} \\
	& \rng([\PSet \cup \CPSet]; \tex{\po}) = \emptyset
	\label{goal:wf_at_most_one_live} \\
	& 	\for{\x} \size{\Allocs_\x \cap \SSet} \leq 1
	\label{goal:wf_single_alloc}
\end{align}

\noindent\textbf{RTS. \eqref{goal:wf_same_tx_same_tid}}\\*
Pick arbitrary $(a, b) \in \st$, \ie $\tx a {=} \tx b$. 
From the definition of $G.\Events$ we then know there exist $\pent_a, \pent_b \in \ipath$ such that $\getE {\pent_a} {=} a $ and $\getE{\pent_b} {=} b$. 
Moreover, from the totality of $\plt{\ipath}$ we know either $\pent_1 \plt{\ipath} \pent_2$ or $\pent_2 \plt{\ipath} \pent_1$. 
Consequently, since $\getE {\pent_a} {=} a $, $\getE{\pent_b} {=} b$ and  from the definition of $G.\po$ we know $(a, b) \in \po \cup \inv{\po}$, as required.\\

\noindent\textbf{RTS. \eqref{goal:wf_tx_contiguous}}\\*
It suffices to show that $\tex{\po}; \tex{\po} = \emptyset$.
Let us proceed by contradiction and assume there exist $a, b$ such that $(a, b) \in \tex{\po}; \tex{\po}$.
That is there exist $c$ such that $(a, c) \in \po$, $(c, b) \in \po$, $\tx a \ne \tx c$ and $\tx c \ne \tx b$. 
On the other hand as $(a, c) \in \po$, from the definition of $\po$ we have $\tx a = \tx c$, leading to a contradiction since we also have $\tx a \ne \tx c$.\\

\noindent\textbf{RTS. \eqref{goal:wf_begin}}\\*
As $\wfipath \ipath$ holds, the result follows from the definition of $\Events$, \eqref{wf:unique_begin}, \eqref{wf:begin_first1}, \eqref{wf:begin_first2} and \eqref{wf:begin_first4}.\\

\noindent\textbf{RTS. \eqref{goal:wf_single_abort_commit_success}}\\*
As $\wfipath \ipath$ holds, the result follows from the definition of $\Events$, \eqref{wf:unique_abort}, \eqref{wf:unique_commit}, \eqref{wf:unique_succ1},  \eqref{wf:unique_succ2} and \eqref{wf:unique_succ3}.\\

\noindent\textbf{RTS. \eqref{goal:wf_abort_succ_maximal}}\\*
As $\wfipath \ipath$ holds, the result follows from the definitions of $\Events$, $\po$, \eqref{wf:iread_before_end1}, \eqref{wf:iread_before_end2}, \eqref{wf:event_before_end1} and \eqref{wf:event_before_end2}.\\

\noindent\textbf{RTS. \eqref{goal:wf_only_success_or_abort_after_commit}}\\*
As $\wfipath \ipath$ holds, the result follows from the definitions of $\Events$, $\po$, \eqref{wf:iread_before_end1} and \eqref{wf:only_end_after_commit}.\\

\noindent\textbf{RTS. \eqref{goal:wf_commit_succ_adjacent}}\\*
As $\wfipath \ipath$ holds, the result follows from part \eqref{goal:wf_only_success_or_abort_after_commit}, the definitions of $\Events$, $\po$, \eqref{wf:nodups}, \eqref{wf:end_write_alloc}, \eqref{wf:abort_or_succ1} and \eqref{wf:abort_or_succ2}.\\

\noindent\textbf{RTS. \eqref{goal:wf_at_most_one_live}}\\*
Follows immediately from the fact that $\tex\po = \emptyset$ as shown in part \eqref{goal:wf_tx_contiguous}.\\

\noindent\textbf{RTS. \eqref{goal:wf_single_alloc}}\\*
As $\wfipath \ipath$ holds, the result follows from the definitions of $\Events$, $\po$, \eqref{wf:end_write_alloc}, \eqref{wf:end_read_only} and \eqref{wf:unique_alloc}.
\end{proof}

\begin{lemma}
\label{lem:cpendrf_empty}
For all paths $\ipath$ and $G$, if $\wfipath \ipath$ holds and $G {=} \makeG{\ipath}$, then $G.\CPRFSet = \emptyset$. 
\end{lemma}

\begin{proof}
Pick an arbitrary $\ipath$ and $G$ such that $\wfipath \ipath$ holds and $G {=} \makeG{\ipath}$.
Given the definition of $G.\CPRFSet$, it suffices to show that for all $(a, b) \in G.\rft$, $a \nin G.\CPSet$. 

Pick an arbitrary $(a, b) \in G.\rft$.
We then know there exist $w, r, \txid_w, \txid_r$ such that $\tx w = \tx a$, $\tx r = \tx b$, $\txid_w \ne \txid_r$ and $(w, r) \in G.\rf$. 
As such, since $(w, r) \in G.\rf$, from \cref{lem:aux} and the definition of $\mathsf{ERF}$ we know $(w, r) \in \mathsf{ERF}_{\txid_r}$. 
That is, there exists $\rs, n, \ipath_1, \ipath_2$ such that $\ipath {=} \ipath_1 \cat [\rmap \txid \rs n] \cat \ipath_2$ and $\rs(\loc r) {=} w$.
As such, since $\wfipath \ipath$ holds, from \eqref{wf:external_read} we know $\wfread {\txid_r} \rs n {\ipath_1}$ holds and thus from the definition of $\wfreadname$ we know there exists $s \in \Succs, k$ such that $\tx s {=} \tx w$, $k \leq n$ and $w \plt{\ipath_1} \wsuc s k$. 
Therefore, from the definition of $G.\Events$ we know $s \in G.\Events$, and consequently since $s \in \Succs$ and $\tx s {=} \tx w {=} \tx a$, we know $a \in G.\SSet$. 
Finally, since $a \in G.\SSet$, from the definition of $G.\CPSet$ we know $a \nin G.\CPSet$, as required.
\end{proof}

\begin{corollary}
\label{lem:vis_equals_succ}
For all paths $\ipath$ and $G$, if $\wfipath \ipath$ holds and $G {=} \makeG{\ipath}$, then $G.\VSet = G.\SSet$. 
\end{corollary}

\begin{proof}
Follows immediately from the definition of $G.\VSet$ and \cref{lem:cpendrf_empty}.
\end{proof}

\begin{lemma}
\label{lem:dynamic_abort_invisibility}
For all $\ipath$ and $G$, if $\wfipath \ipath$ holds and $G {=} \makeG{\ipath}$, then:
\begin{align}
	&\dom(G.\rft) \suq G.\VSet  \label{goal:dai_rf} \\
	& G.(\Writes \cap \VSet) \suq \rng\big([G.(\Allocs \cap \VSet)]; G.\mo\big)
	\label{goal:dai_mo} 
\end{align}
\end{lemma}

\begin{proof}
Pick arbitrary $\ipath$ and $G$ such that $\wfipath \ipath$ and $G {=} \makeG{\ipath}$. We then proceed as follows.\\ 

\noindent\textbf{RTS. \eqref{goal:dai_rf}}\\*
Note that given the definition of $G.\VSet$, it suffices to show $\dom(G.\rft) \suq G.\SSet$.
Pick an arbitrary $a, b \in G.\Events$ such that $(a, b) \in G.\rft$. 
That is, there exists $w, r \in G.\Events, \txid_w, \txid_r$ such that $\tx b = \tx r = \txid_r$, $\tx a {=} \tx w = \txid_w$, $\txid_w \ne \txid_r$ and $(w, r) \in G.\rf$.
As $\makeG{\ipath} {=} G$ and $\wfipath \ipath$, $\tx r = \txid_r$, $\tx w = \txid_w$ and $\txid_w \ne \txid_r$, from \cref{lem:rf} we know $(w, r) \in \mathsf{ERF}_{\txid_r}$. 
Let $\loc r = \loc w = \x$. From the definition of $\mathsf{ERF}_{\txid_r}$ we know there exist $n, \ipath_1, \ipath_2$ such that $r \!\in\! \ipath \cap \Reads$, $\ipath {=} \ipath_1 \cat [\rmap {\txid_r} \rs n] \cat \ipath_2$, $\rs(\x) {=} w$ and $\for{\rs'\!, m} \rmap {\txid_r} {\rs'\!} m \nin \ipath_2$.
As such, from $\wfipath \ipath$ and \ref{wf:external_read} we have $\wfread {\txid_r} {\rs} n {\ipath_1}$ and thus from the definition of $\wfread {\txid_r} {\rs} n {\ipath_1}$ we know there exists $s, m$ such that $m \leq n$, $\tx w {=} \tx s {=} \txid_w$ and $w \plt{\ipath_1} \wsuc{s} m$. 
That is, by definition we have $w \in \SSet$ and thus as $\tx a {=} \tx w = \txid_w$ we have $a \in \SSet$, as required.\\

\noindent\textbf{RTS. \eqref{goal:dai_mo}}\\*
From \cref{lem:vis_equals_succ} it suffices to show $G.(\Writes \cap \SSet) \suq \rng\big([G.(\Allocs \cap \SSet)]; G.\mo\big)$. 
Pick arbitrary $w \in G.(\Writes \cap \SSet)$ and let $\tx w {=} \txid$. 
By definition we know there exists $s \in G.\Succs$ such that $\tx s {=} \tx w {=} \txid$. 
Consequently, given the definition of $G.\Events$, either i) $s \in \ipath$ or ii) there exists $n$ such that $\wsuc s n \in \ipath$. As such, since $w \in \ipath$, from $\wfipath \ipath$, \eqref{wf:event_before_end1}, \eqref{wf:event_before_end2} we know either i) $w \plt{\ipath} s$ or ii) there exists $n$ such that $w \plt{\ipath} \wsuc s n$. 
Therefore, from $\wfipath \ipath$ and \eqref{wf:end_read_only} we know case (i) does not arise and thus there exists $n$ such that $w \plt{\ipath} \wsuc s n$.
As $w \plt{\ipath} \wsuc s n$, we know there exists $\ipath_1$ such that $\ipath {=} \ipath_1 \cat [w] \cat - \cat [\wsuc s n] \cat -$. 
Hence, since $\tx w {=} \tx s$ from $\wfipath \ipath$ and \eqref{wf:wf_write} we have $\wfwrite w {\ipath_1}$. 
Consequently, from the definition of $\wfwrite w {\ipath_1}$ we know there exists $m \in \Allocs \cap \ipath_1$ such that $\loc{m} {=} \loc w$ and $(\tx m {\ne} \tx w \Rightarrow \exsts{s_m \in \Succs, k} \tx {s_m} {=} \tx m \land m \plt{\ipath_1} \wsuc {s_m} k)$.
There are now two cases to consider: 
a) $\tx m {=} \tx w$; or
b) $\tx m {\ne} \tx w$.

In case (a), as $\ipath {=} \ipath_1 \cat [w] \cat - \cat [\wsuc s n] \cat -$ and $m \in \Allocs \cap \ipath_1$ we have $m \plt{\ipath} w$. 
As $m \in \Allocs \cap \ipath_1$ and thus $m \in \Allocs \cap \ipath$ by definition of $G.\Events$ we have $m \in G.\Allocs$. 
Moreover, since $\tx m {=} \tx w$ (assumption of case (i)), $\loc{m} {=} \loc w$, $w \in G.\Writes$, $m \in G.\Allocs$ and $m \plt{\ipath} w$, from the definition of $G.\mo$ we have $(m, w) \in G.\mo$.
Furthermore, since $\tx m {=} \tx w$ and $w \in G.\SSet$, by definition we also have $m \in G.\SSet$ and thus $m \in G.(\Allocs \cap \SSet)$. 
Consequently, since $m \in G.(\Allocs \cap \SSet)$ and $(m, w) \in G.\mo$, we have $w \in \rng\big([G.(\Allocs \cap \SSet)]; G.\mo\big)$, as required.

In case (b), from $(\tx m {\ne} \tx w \Rightarrow \exsts{s_m \in \Succs, k} \tx {s_m} {=} \tx m \land m \plt{\ipath_1} \wsuc {s_m} k)$ we know there exist $s_m \in \Succs, k$ such that $\tx {s_m} {=} \tx m$ and $m \plt{\ipath_1} \wsuc {s_m} k$. 
As such, by definition we have $s_m \in G.\Succs$ and thus $s_m \in G.\SSet$, and hence since $\tx {s_m} {=} \tx m$, by definition we also have $m \in G.\SSet$, \ie $m \in G.(\Allocs \cap \SSet)$. 
On the other hand, since $\ipath {=} \ipath_1 \cat [w] \cat - \cat [\wsuc s n] \cat -$ and $m \plt{\ipath_1} \wsuc {s_m} k$, we know $\wsuc {s_m} k \plt{\ipath} \wsuc s n$. 
As such, from $\wfipath \ipath$ and \eqref{wf:end_write_alloc_in_order1} we know $k < n$. 
Moreover, as $\wsuc {s_m} k, \wsuc s n \in \ipath$ and $\tx {s_m} {=} \tx m$, by definition we have $(\tx m, k), (\tx w, n) \in \swriters$. 
Therefore, since $k < n$, from the definition of $\peo(\ipath)$ we have $(\tx m, \tx w) \in \peo(\ipath) \suq \teo(\ipath)$. 
As such, since $\loc{m} {=} \loc w$, $w \in G.\Writes$, $m \in G.\Allocs$ and $(\tx m, \tx w) \in \teo(\ipath)$, from the definition of $G.\mo$ we have $(m, w) \in G.\mo$.
Consequently, since $m \in G.(\Allocs \cap \SSet)$ and $(m, w) \in G.\mo$, we have $w \in \rng\big([G.(\Allocs \cap \SSet)]; G.\mo\big)$, as required.
\end{proof}

\begin{corollary}
\label{lem:rft_from_succ}
For all $\ipath$ and $G$, if $\wfipath \ipath$ holds and $G {=} \makeG{\ipath}$, then $\dom(G.\rft) \suq G.\SSet$. 
\end{corollary}

\begin{proof}
Pick arbitrary $\ipath$ and $G$ such that $\wfipath \ipath$ holds and $G {=} \makeG{\ipath}$.
From part \eqref{goal:dai_rf} proved in \cref{lem:dynamic_abort_invisibility}, we know $\dom(G.\rft) \suq G.\VSet$. Consequently, from \cref{lem:vis_equals_succ} we have $\dom(G.\rft) \suq G.\SSet$, as required.
\end{proof}

\begin{lemma}
\label{lem:hb_acyclic}
For all paths $\ipath$ and $G$, if $\wfipath \ipath$ holds and $G {=} \makeG{\ipath}$, then:
\begin{align}
	& \for{a, b \in \Events} (a, b) \in G.\clo \Rightarrow (\tx a, \tx b) \in \teo(\ipath) \label{goal:hb_clo} \\
	& \for{a, b \in \Events} (a, b) \in G.\rft \Rightarrow (\tx a, \tx b) \in \teo(\ipath) \label{goal:hb_rft} \\
	& \for{a, b \in \Events} (a, b) \in G.\mot \Rightarrow (\tx a, \tx b) \in \teo(\ipath) \label{goal:hb_mot} \\
	& \for{a, b \in \Events} (a, b) \in G.(\rbt;[\VSet]) \Rightarrow (\tx a, \tx b) \in \teo(\ipath) \label{goal:hb_rbt} \\ 
	& \for{a, b \in \Events} (a, b) \in \transC{G.(\clo \cup \rft \cup \mot \cup \rbt;[\VSet])} \Rightarrow (\tx a, \tx b) \in \teo(\ipath) \label{goal:hb_trans} 
\end{align}
\end{lemma}

\begin{proof}
Pick arbitrary $\ipath$ and $G$ such that $\wfipath \ipath$ holds and $G {=} \makeG{\ipath}$.\\

\noindent\textbf{RTS. \eqref{goal:hb_clo}}\\*
Follows immediately from the definitions of $G.\clo$ and $\teo(\ipath)$.\\

\noindent\textbf{RTS. \eqref{goal:hb_rft}}\\*
Pick an arbitrary $a, b \in G.\Events$ such that $(a, b) \in G.\rft$. 
That is, there exists $w, r \in G.\Events, \txid_w, \txid_r$ such that $\tx b = \tx r = \txid_r$, $\tx a {=} \tx w = \txid_w$, $\txid_w \ne \txid_r$ and $(w, r) \in G.\rf$.
As $\makeG{\ipath} {=} G$ and $\wfipath \ipath$, $\tx r = \txid_r$, $\tx w = \txid_w$ and $\txid_w \ne \txid_r$, from \cref{lem:rf} we know $(w, r) \in \mathsf{ERF}_{\txid_r}$. 
Let $\loc r = \loc w = \x$. From the definition of $\mathsf{ERF}_{\txid_r}$ we know there exist $n, \ipath_1, \ipath_2$ such that $r \!\in\! \ipath \cap \Reads$, $\ipath {=} \ipath_1 \cat [\rmap {\txid_r} \rs n] \cat \ipath_2$, $\rs(\x) {=} w$ and $\for{\rs'\!, m} \rmap {\txid_r} {\rs'\!} m \nin \ipath_2$.
As such, from $\wfipath \ipath$ and \ref{wf:external_read} we have $\wfread {\txid_r} {\rs} n {\ipath_1}$ and thus from the definition of $\wfread {\txid_r} {\rs} n {\ipath_1}$ 
we know there exists $s', s, m$ such that $\wsuc {s'} n \in \ipath_1$, $m \leq n$, $\tx w {=} \tx{s} {=} \txid_w$ and $w \plt{\ipath_1} \wsuc{s} m$ 
and thus by definition we have $(\txid_w, m) \in \swriters$. 
There are now two cases to consider: 
i) there exists $k$ such that $(\txid_r, k) \in \swriters$; or 
ii) $\for k (\txid_r, k) \nin \swriters$.

In case (i) since $(\txid_r, k) \in \swriters$ we know there exists $s_r$ such that $\tx{s_r} {=} \tx r {=} \txid_r$ and $\wsuc {s_r} k \in \ipath$. As such, since $\ipath {=} \ipath_1 \cat [\rmap {\txid_r} \rs n] \cat \ipath_2$, from $\wfipath \ipath$ and \eqref{wf:read_before_end2} we know $\wsuc {s_r} k \in \ipath_2$. Consequently, since $\wsuc{s} m \in \ipath_1$ (from $w \plt{\ipath_1} \wsuc{s} m$), we have $\wsuc s m \plt{\ipath} \wsuc{s_r} k$. 
As such, from $\wfipath \ipath$ and \eqref{wf:end_write_alloc_in_order1} we have $m < k$. 
Finally, since $(\txid_w, m) \in \swriters$, $(\txid_r, k) \in \swriters$ and $m < n$, from the definition of $\peo(\ipath)$ we have $(\txid_w, \txid_r) \in \peo(\ipath)$, as required.

In case (ii), since $(w, r) \in G.\rf$, we know there exists $s_r \in \Succs \cap G.\Events$ such that $\tx {s_r} {=} \txid_r$. As such, from the definition of $G.\Events$ we have $s_r \in \Succs \cap \ipath$.
Consequently, from the assumption of the case and since $\ipath {=} \ipath_1 \cat [\rmap {\txid_r} \rs n] \cat \ipath_2$, $\rs(\x) {=} w$ and $\for{\rs'\!, m} \rmap {\txid_r} {\rs'\!} m \nin \ipath_2$ we have $(\txid_r, n) \in \esreaders$. 
As such, since $(\txid_w, m) \in \swriters$, $(\txid_r, n) \in \esreaders$ and $m \leq n$, from the definition of $\peo(\ipath)$ we have $(\txid_w, \txid_r) \in \peo(\ipath)$ and by definition $(\txid_w, \txid_r) \in \teo(\ipath)$, as required.\\

\noindent\textbf{RTS. \eqref{goal:hb_mot}}\\*
Follows immediately from the definition of $G.\mot$.\\

\noindent\textbf{RTS. \eqref{goal:hb_rbt}}\\*
Pick an arbitrary $a, b \!\in\! G.\Events$ such that $(a, b) \!\in\! G.\rbt; [\VSet]$. 
That is, there exists $w, r, w_r \!\in\! G.\Events, \txid, \txid', \x$ such that $(w_r, r) \!\in\! G.\rf$, $(w_r, w) \in \mo$, $w \in G.\Events \cap \VSet$, $\tx r {=} \tx{a} {=} \txid_r$, $\tx w {=} \tx b {=} \txid_w$, $\txid_r {\ne} \txid_w$ and $\loc{w_r} = \loc w {=} \loc r {=} \x$. 
As  $w \in G.\Events \cap \VSet$, from \cref{lem:vis_equals_succ} we know  $w \in G.\Events \cap \SSet$.
There are now two cases to consider: 
i) $\tx{w_r} {=} \txid_r$; or
ii) $\tx{w_r} {\ne} \txid_r$.
In case (i), by definition we have $(w_r, w) \!\in\! \mot$ and thus from part \eqref{goal:hb_mot} we have $(\txid_r, \txid_w) \!\in\! \teo(\ipath)$, as required. 

In case (ii) since $(w_r, r) \!\in\! G.\rf$, $\tx r {=} \tx{a} {=} \txid_r$  $\tx{w_r} {=} \txid$ and $\tx{w_r} {\ne} \txid_r$, by definition we have $(w_r, r) \!\in\! G.\rft$.
As such, from \cref{lem:rft_from_succ} we know $w_r \!\in\! G.\Events \cap \SSet$. 
Therefore, since  $w, w_r \!\in\! G.\Events \cap \SSet$, from \cref{lem:aux} and the definition of $\swriters$ we know there exist $s_1, s_2 \in \Succs, n_1, n_2$ such that $(\txid, n_1), (\txid_w, n_2) \in \swriters$, $\wsuc {s_1} {n_1}, \wsuc {s_2} {n_2} \in \ipath$, $\tx {s_1} {=} \txid$ and $\tx{s_2} {=} \txid_w$. 
As such, since $(w_r, w) \in \mo$ and $(\txid, n_1), (\txid_w, n_2) \in \swriters$, from the definitions of $\mo$, $\teo(\ipath)$ and $\peo(\ipath)$ we have $n_1 < n_2$ and thus from $\wfipath \ipath$ and \eqref{wf:end_write_alloc_in_order1} (and since $\plt \ipath$ is a total order), we have $\wsuc {s_1} {n_1} \plt{\ipath} \wsuc {s_2} {n_2}$.
Moreover, from $\wfipath \ipath$ and \eqref{wf:event_before_end2} we know $w_r \plt{\ipath} \wsuc {s_1} {n_1}$ and $w \plt{\ipath} \wsuc {s_2} {n_2}$. 

On the other hand, as $\wfipath \ipath$, $\makeG \ipath {=} G$, $\tx r {=} \txid_r$, $\tx{w_r} {\ne} \txid_r$ and $(w_r, r) \in G.\rf$, from \cref{lem:rf} we know $(w_r, r) \in \mathsf{ERF}_{\txid_r}$. 
That is, there exists $\ipath_1, \ipath_2, \rs, n$ such that 
$\ipath {=} \ipath_1 \cat [\rmap {\txid_r} \rs n] \cat \ipath_2$, $\rs(\x) {=} w_r$ and $\for{\rs', m} \rmap {\txid_r} {\rs'} m \nin \ipath_2$. 
As such, from $\wfipath \ipath$ and \eqref{wf:external_read} we have $\wfread {\txid_r} {\rs} n {\ipath_1}$. 
Consequently, since $w \plt{\ipath} \wsuc {s_1} {n_1}$, from \eqref{wf:nodups}, \eqref{wf:unique_succ3} and $\wfread {\txid_r} {\rs} n {\ipath_1}$ we know $w_r \plt{\ipath_1} \wsuc {s_1}{n_1}$
and $n_1 \leq n$. 

From the definitions of $\swriters$, $\esreaders$ and $\ereaders$ we know there exist $s \in \Succs, m$ such that either: 
a) $(\txid_r, m) \!\in\! \swriters$, $\tx s {=} \txid_r$, and $\wsuc s m \!\in\! \ipath$; or
b) $(\txid_r, n) \in \esreaders \cup \ereaders$ and $\for k (\txid_r, k) \!\nin\! \swriters$.
%
Since $\tx r {=} \txid_r$, from $\wfipath \ipath$, \eqref{wf:read_before_end2}, and \eqref{wf:ext_read_in_order1} we know either
a) $(\txid_r, m) \in \swriters$, $\wsuc s m \in \ipath$, $\rmap {\txid_r} \rs n \plt{\ipath_2} \wsuc s m$ and $n < m$; or
b) $(\txid_r, n) \in \esreaders \cup \ereaders$, $\for k (\txid_r, k) \nin \swriters$ and $\rmap {\txid_r} \rs n \plt{\ipath_2} s$.

In case (a) there are two cases to consider: 
A) $m < n_2$; or
B) $n_2 \leq m$. 
In case (A), as $(\txid_r, m) \in \swriters$, $(\txid_w, n_2) \in \swriters$ and $m < n_2$, by definition we have $(\txid_r, \txid_w) \in \peo(\ipath) \suq \teo(\ipath)$, as required.
In case (B), we arrive at a contradiction as follows.

As $\wsuc s m, \wsuc {s_2} {n_2} \in \ipath$ and $s \ne s_2$ (because $\tx s {=} \txid_r$ while $\tx{s_2} {=} \txid_w$ and $\txid_r {\ne} \txid_w$), from $\wfipath \ipath$ and \eqref{wf:unique_succ3} we have $n_2 {\ne} m$ and thus from the assumption of case (B) we have $n_2 < m$.   
As $\wsuc s m \!\in\! \ipath$, $\ipath {=} \ipath_1 \cat [\rmap {\txid_r} \rs n] \cat \ipath_2$, $\for{\rs', m} \rmap {\txid_r} {\rs'} m \nin \ipath_2$, $\rmap {\txid_r} \rs n \plt{\ipath_2} \wsuc s m$, from $\wfipath \ipath$, \eqref{wf:nodups}, \eqref{wf:unique_succ3}, \eqref{wf:ext_read_before_succ} we know $\ipath {=} \ipath_1 \cat [\rmap {\txid_r} \rs n] \cat [\wsuc s m] \cat -$ and $n {=}  m {-} 1$. 
As such, since $n_2 < m$ and $n {=}  m {-} 1$ we have $n_2 \leq n$. 
Consequently, as $\ipath {=} \ipath_1 \cat [\rmap {\txid_r} \rs n] \cat \ipath_2$, $n_2 \leq n$ and $\wsuc {s_2} {n_2} \in \ipath$ (from $w \plt{\ipath} \wsuc {s_2} {n_2}$), from totality of $\plt \ipath$, $\wfipath \ipath$ and \eqref{wf:ext_read_in_order1} we have $\wsuc {s_2} {n_2} \plt{\ipath} \rmap {\txid_r} \rs n$, \ie since $\ipath {=} \ipath_1 \cat [\rmap {\txid_r} \rs n] \cat \ipath_2$ and $w \plt{\ipath} \wsuc {s_2} {n_2}$ we have $w \plt{\ipath_1} \wsuc {s_2} {n_2}$. 
That is, we have $\rs(\x) {=} w_r$, $\tx{s_1} {=} \tx {w_r} {=} \txid$, $w_r \plt{\ipath_1} \wsuc {s_1} {n_1} $, $n_1 \leq n$, $w \in \Writes \cup \Allocs$, $\loc w {=} \x$, $\tx w {=} \tx{s_2} {=} \txid_w$, $w \plt{\ipath_1} \wsuc {s_2} {n_2}$ and $n_1 < n_2$, contradicting the last conjunct of $\wfread {\txid_r} \rs n {\ipath_1}$.

In case (b) there are similarly two cases to consider: 
A) $n < n_2$; or
B) $n_2 \leq n$. 
In case (A), as $(\txid_r, n) \in \esreaders \cup \ereaders$, $(\txid_w, n_2) \in \swriters$ and $n < n_2$, by definition we have $(\txid_r, \txid_w) \in \peo(\ipath) \suq \teo(\ipath)$, as required.
In case (B), we arrive at a contradiction as follows. 
As $\ipath {=} \ipath_1 \cat [\rmap {\txid_r} \rs n] \cat \ipath_2$, $n_2 \leq n$ and $\wsuc {s_2} {n_2} \in \ipath$ (from $w \plt{\ipath} \wsuc {s_2} {n_2}$), from totality of $\plt \ipath$, $\wfipath \ipath$ and \eqref{wf:ext_read_in_order1} we have $\wsuc {s_2} {n_2} \plt{\ipath} \rmap {\txid_r} \rs n$, \ie since $\ipath {=} \ipath_1 \cat [\rmap {\txid_r} \rs n] \cat \ipath_2$ and $w \plt{\ipath} \wsuc {s_2} {n_2}$ we have $w \plt{\ipath_1} \wsuc {s_2} {n_2}$. 
That is, we have $\rs(\x) {=} w_r$, $\tx{s_1} {=} \tx {w_r} {=} \txid$, $w_r \plt{\ipath_1} \wsuc {s_1} {n_1} $, $n_1 \leq n$, $w \in \Writes \cup \Allocs$, $\loc w {=} \x$, $\tx w {=} \tx{s_2} {=} \txid_w$, $w \plt{\ipath_1} \wsuc {s_2} {n_2}$ and $n_1 < n_2$, contradicting the last conjunct of $\wfread {\txid_r} \rs n {\ipath_1}$. 
\\

\noindent\textbf{RTS. \eqref{goal:hb_trans}}\\*
Follows from \eqref{goal:hb_clo}--\eqref{goal:hb_rbt} and the transitivity of $\teo(\ipath)$ (\cref{lem:teo_strict_total_order}).
\end{proof}

\begin{lemma}
\label{lem:dynamic_opacity}
For all $\ipath$ and $G$, if $\wfipath \ipath$ holds and $G {=} \makeG{\ipath}$, then $G$ is dynamically opaque. 
\end{lemma}

\begin{proof}
Pick arbitrary $\ipath$ and $G$ such that $\wfipath \ipath$ and $G {=} \makeG{\ipath}$. We are required to show: 
\begin{align}
	& G.\rfI \suq G.\po  \label{goal:opacity_rfI} \\
	& G.\moI \suq G.\po  \label{goal:opacity_moI} \\
	& G.\rbI \suq G.\po  \label{goal:opacity_rbI} \\
	& G.\tin{(\tex \mo; \tex \rf)\,} = \emptyset \label{goal:opacity_int2} \\
	& G.(\clo \cup \rft \cup \mot \cup \rbt; [\VSet]) \text{ is acyclic} \label{goal:opacity_hb} \\
	&\dom(G.\rft) \suq G.\VSet  \label{goal:opacity_rf} \\
	& G.(\Writes \cap \VSet) \suq \rng\big([G.(\Allocs \cap \VSet)]; G.\mo\big) \label{goal:opacity_mo} 
\end{align}

\noindent\textbf{RTS. \eqref{goal:opacity_rfI}}\\*
Pick arbitrary $w, r$ such that $(w, r) \in G.\rfI$. 
That is, there exists $\txid$ such that $\tx w {=} \tx r {=} \txid$ and $(w, r) \in G.\rf$. 
As such, since $\wfipath \ipath$, $G {=} \makeG{\ipath}$ and $(w, r) \in G.\rf$, from \cref{lem:rf} we know $\intr r w \in \ipath$. 
Therefore, from $\wfipath \ipath$ and \eqref{wf:internal_read} we know $w \plt{\ipath} \intr r w$. 
Consequently, since $\tx w {=} \tx r$, from the definition of $G.\po$ we have $(w, r) \in G.\po$, as required.\\

\noindent\textbf{RTS. \eqref{goal:opacity_moI}}\\*
Pick an arbitrary $w, w'$ such that $(w, w') \in G.\moI$. 
That is, there exists $\txid$ such that $\tx w {=} \tx {w'} {=} \txid$ and $(w, w') \in G.\mo$. 
Let $\loc w {=} \loc{w'} {=} \x$. 
From the definitions of $G.\mo$, $G.\mo_\x$, the totality of $\plt \ipath$ and since $\tx w {=} \tx {w'} {=} \txid$ we then know 
$w \plt{\ipath} w'$. 
Consequently, since $\tx w {=} \tx{w'}$, from the definition of $G.\po$ we have $(w, w') \in G.\po$, as required.\\

\noindent\textbf{RTS. \eqref{goal:opacity_rbI}}\\*
Pick arbitrary $r, w$ such that $(r, w) \in G.\rbI$. 
That is, there exists $w_r \in G.\Events, \txid$ such that $\tx r {=} \tx {w} {=} \txid$, $(w_r, r) \in G.\rf$ and $(w_r, w) \in  G.\mo$.
Let $\loc r {=} \loc w {=} \x$. 
From the definition of $G.\Events$ we know $w \in \ipath$. 
Moreover, from \cref{lem:rf} and the definitions of $\mathsf{IRF}$ and $\mathsf{ERF}_\txid$ we know either 
i) $\tx {w_r} {=} \tx r$ and $\intr r {w_r} \in \ipath$; or
ii) $\tx {w_r} {\ne} \tx r$, $r \!\in\! \Reads \cap \ipath$ and there exists $\exsts{\rs, \ipath_1, \ipath_2, n}$ such that 
$\ipath {=} \ipath_1 \cat [\rmap \txid \rs n] \cat \ipath_2$, 
$\land \rs(\x) {=} w_r$ and
$\for{\rs'\!, m} \rmap \txid {\rs'\!} m \nin \ipath_2$.

In case (i), since $w, \intr r {w_r} \in \ipath$, from the totality of $\plt \ipath$ we know either 
a) $\intr r {w_r} \plt{\ipath} w$; or 
b) $w \plt{\ipath} \intr r {w_r}$. 
In case (i.a), as $\tx r {=} \tx {w}$, from the definition of $G.\po$ we have $(r, w) \in G.\po$, as required.
Case (i.b) cannot arise as we arrive at a contradiction as follows. As $\tx {w_r} {=} \tx r$ (the assumption of case i) and $\tx r {=} \tx {w}$ we have $\tx {w_r} {=} \tx w$. 
Consequently, since $(w_r, w) \in G.\mo$, from the definition of $G.\mo$ we know $w_r \plt{\ipath} w$.
We then have $w \in \Writes \cup \Allocs$, $\loc r {=} \loc w {=} \x$,  $\tx r {=} \tx {w}$ and $w_r \plt{\ipath} w \plt{\ipath} \intr r {w_r}$, contradicting the \eqref{wf:internal_read} property of $\wfipath \ipath$.

In case (ii), as $w, r \in \ipath$, from the totality of $\plt \ipath$ we know either 
a) $r \plt{\ipath} w$; or 
b) $w \plt{\ipath} r$. 
In case (ii.a), as $\tx r {=} \tx {w}$, from the definition of $G.\po$ we have $(r, w) \in G.\po$, as required.
Case (ii.b) cannot arise as we arrive at a contradiction as follows. 
As $\tx r {=} \tx {w} {=} \txid$ and $\ipath {=} \ipath_1 \cat [\rmap \txid \rs n] \cat \ipath_2$, from $\wfipath \ipath$ and \eqref{wf:external_read} we have $\wfread \txid \rs n {\ipath_1}$. 
Moreover, as $\ipath {=} \ipath_1 \cat [\rmap \txid \rs n] \cat \ipath_2$, $\for{\rs'\!, m} \rmap \txid {\rs'\!} m \nin \ipath_2$ and $r \in \Reads \cap \ipath$, from $\wfipath \ipath$ and \eqref{wf:total_external_read} we know $r \plt \ipath \rmap \txid \rs n$. 
As such, from $w \plt{\ipath} r$ we also have $w \plt \ipath \rmap \txid \rs n$. 
Consequently, since $\ipath {=} \ipath_1 \cat [\rmap \txid \rs n] \cat \ipath_2$, we have $w \!\in\! \ipath_1$.
We then have $w \in \Writes \cup \Allocs$, $\loc w {=} \loc{w_r} {=} \x$, $\tx w {=} \tx r {=} \txid$ and $w \in \ipath_1$, contradicting $\wfread \txid \rs n {\ipath_1}$.\\ 


\noindent\textbf{RTS. \eqref{goal:opacity_int2}}\\*
Let us proceed by contradiction ans assume there exit $w, r$ such that $(w, r) \in G.\tin{(\tex \mo; \tex \rf)\,}$. 
That is, there exists $\txid, \txid'\!, w'\!$ such that $w, w'\! \!\in\! \Writes \cup \Allocs$, $\tx w {=} \tx r {=} \txid$, $\tx{w'} {=} \txid'\!$, $\txid {\ne} \txid'\!$, $(w, w') \in G.\mo$, $(w', r) \in G.\rf$ and $\loc w {=} \loc r {=} \loc{w'}$. 
By definition we Then have $(w', r) \in G.\rft$ and $(w, w') \in G.\mot$. 
As such, from \cref{lem:hb_acyclic} we have $(t', t) \in \teo(\ipath)$ and $(t, t') \in \teo(\ipath)$. Consequently, from the transitivity of $\teo(\ipath)$ (since $\teo(\ipath)$ is an order \cref{lem:teo_strict_total_order}) we have $(t, t) \in \teo(\ipath)$, leading to a contradiction as $\teo(\ipath)$ is a strict total order (\cref{lem:teo_strict_total_order}) and is thus irreflexive. 
\\

\noindent\textbf{RTS. \eqref{goal:opacity_hb}}\\*
We proceed by contradiction. 
Let us assume there exists $a$ such that $(a, a) \in \transC{G.(\clo \cup \rft \cup \mot \cup \rbt)}$. 
From \cref{lem:hb_acyclic} (part \eqref{goal:hb_trans}) we then know $(\tx a, \tx a) \in \teo(\ipath)$.
This, however, leads to a contradiction as from \cref{lem:teo_strict_total_order} we know $\teo(\ipath)$ is irreflexive and thus $(\tx a, \tx a) \nin \teo(\ipath)$.\\

\noindent\textbf{RTS. \eqref{goal:opacity_rf} and \eqref{goal:opacity_mo}}\\*
Follows immediately from \cref{lem:dynamic_abort_invisibility}.
%
%
%
\end{proof}

%

%

\begin{theorem}
For all $\pcmap, \smap, \mems, \sigma$, 
if $\sigma$ denotes the execution $\pcmap_0, \smap_0, \mems_0 \redConf{}^* \pcmap, \smap, \mems$, 
then the ordered history ${\it OH}_\sigma$ satisfies \DDO.
\end{theorem}

\begin{proof}
Pick arbitrary $\pcmap, \smap, \mems, \sigma$ such that $\sigma$ denotes the execution $\pcmap_0, \smap_0, \mems_0 \redConf{}^* \pcmap, \smap, \mems$.
From \cref{lem:ddtms_refines_instrumented_ddtms} we know there exist
$\ismap, \imems, \ipath, \hist$ such that $\pcmap_0, \ismap_0, \imems_0, [] \redConf{}^* \pcmap, \ismap, \allowbreak \imems, \ipath$
and ${\it OH}_\sigma = \getHist \ipath = (\Events, \pteo)$.
We must then show that the history $\hist$ is dynamically and durably opaque. 
Let $(\Events', \pteo') = (\Events \setminus \CrashMarkers, \coerce\pteo{\Events \setminus\CrashMarkers})$. 
From the definition of dynamic, durable opacity we must show $(\Events', \pteo')$ is dynamically opaque. 

Let $\ipath' {=} \ipath \setminus \CrashMarkers$. From \cref{lem:wf_path_crashes_removed} we know $\wfipath{\ipath'}$ holds. 
As $(\Events, \pteo) {=} \getHist \ipath$ and $(\Events', \pteo') {=} (\Events \setminus \CrashMarkers, \coerce\pteo{\Events \setminus\CrashMarkers})$, by definition of $\getHist .$ and $\ipath'$ we then know $(\Events', \pteo') = \getHist{\ipath'}$. 

Pick an arbitrary prefix $(\Events_p, \teo_p)$ of $(\Events', \pteo')$. 
As $(\Events', \pteo') = \getHist{\ipath'}$, from the definition of history prefixes we know there exists $\ipath_p$ such that $(\Events_p, \pteo_p) = \getHist{\ipath_p}$ and $\ipath_p$ is a prefix of $\ipath'$. 
As such, since $\wfipath{\ipath'}$, from \cref{lem:wf_path_prefix} we know $\wfipath{\ipath_p}$ holds. 
Let $G = \makeG{\ipath_p}$.
From the definition of $\makeG{\ipath_p}$ we then know that $\Events_p = G.\Events$, $\internal{(\pteo_p)} = G.\po$ and $\clo(\hist_p) = G.\clo$. 
That is, there exists $\rf, \mo$ such that $G = (\Events_p, \internal{(\pteo_p)}, \clo(\hist_p), \rf, \mo)$. 
Consequently, since $\wfipath{\ipath_p}$ and $G \eqdef \makeG{\ipath_p}$, from \cref{lem:dynamic_opacity} we know $G$ is dynamically opaque, as required.
\end{proof}
\section{Additional Definitions for Section \ref{sec:volatile_framework} and Relation to Original Opacity Definition}
\label{app:add_graphs}

\subsection{Well-formednes of Execution Graphs}
\label{app:wf}

\begin{definition}[Well-formedness]
\label{def:wf_execution}
An execution $(\Events, \po, \clo, \rf, \mo)$ is \emph{well-formed} iff $(\Events, \po)$ is well-formed. 
A pair $(\Events, \po)$ is well-formed iff: 
\begin{enumerate}[label=(\arabic*)]
	\item Events of the same transaction are by the same thread ($\st \suq \po \cup \inv{\po}$), and 
	are not interleaved by events of other transactions in that they form a contiguous block in $\po$ ($\st \cap (\tex{\po}; \tex{\po}) = \emptyset$).
	\label{wf:same_tx_same_tid}
%
	\item Each $\txid$ contains exactly one begin event (\ie $\size{\Events_{\txid} \cap \Begins} {=} 1$) denoted by $\be$, such that $\be$ is the $\po$-minimal event in $\txid$ (\ie $\be \not\in \rng(\poI)$).
	\label{wf:begin}
	\item Each $\txid$ contains at most a single abort event ($\size{\Events_{\txid} \cap \Aborts} \leq 1$), a single commit event ($\size{\Events_{\txid} \cap \Commits} \leq 1$), and a single success event ($\size{\Events_{\txid} \cap \Succs} \leq 1$), denoted by $\abe$, $\ce$ and $\se$, respectively.
	Moreover, $\abe$ and $\se$ (if they exists) are $\po$-maximal in $\txid$ ($\abe \not\in \dom(\poI)$ and $\se \not\in \dom(\poI)$). 
	\label{wf:single_abort_commit_success}
%
	\item A commit event may only be followed by an abort or success event (\ie $\rng([C]; \tin\po) \suq \Aborts \cup \Succs$), 
	and $\se$ (if it exists) is the immediate $\po$ successor of $\ce$ (\ie $\se \in \Events \Rightarrow (\ce, \se) \in \imm\po$).
	\label{wf:only_success_or_abort_after_commit}
%
	 \item 
	Each thread contains at most one pending or commit-pending transaction which is also the last transaction of the thread: 
	$\rng([\PSet \cup \CPSet]; \tex{\po}) = \emptyset$.
	 \label{wf:at_most_one_live}
	 \item Finally, each location $\x$ may be allocated at most once within a successfully committed transaction:
	$\for{\x} \size{\Events \cap \Allocs_\x \cap \SSet} \leq 1$.
	\label{wf:single_alloc}
\end{enumerate}
\end{definition}

\subsection{Relation to Original Opacity Definition}
\label{app:org_opacity}

The different sets of events induced by an execution $G$
are applied for a history $\hist = (\Events, \teo)$ as expected:
\begin{enumerate}
\item $\hist.\CSet=\set{e\in\Events \mid \exists a\in\Events\ldotp \lLAB(a)=\commitL \land \tup{a,e}\in\st}$.
\item $\hist.\ASet=\set{e\in\Events \mid \exists a\in\Events\ldotp \lLAB(a)=\abortL \land \tup{a,e}\in\st}$.
\item $\hist.\SSet=\set{e\in\Events \mid \exists a\in\Events\ldotp \lLAB(a)=\succL \land \tup{a,e}\in\st}$.
\item $\hist.\CPSet \eqdef \hist.\CSet \setminus (\hist.\ASet \cup \hist.\SSet)$.
\item $\hist.\PSet \eqdef \Events \setminus (\hist.\CPSet \cup \hist.\ASet \cup \hist.\SSet)$.
\end{enumerate}

\begin{definition}
A \emph{completion} of a history $\hist= (\Events, \teo)$ is any history $\hist' =(\Events', \teo')$
such that the following hold:
\begin{itemize}
\item $\hist$ is a prefix of $\hist'$.
\item For every $e\in \Events' \setminus \Events$,
we have $\lLAB(e)\in \set{\abortL,\succL}$.
\item $\hist'.\CPSet \cup \hist'.\PSet = \emptyset$.
\end{itemize}
\end{definition}

\begin{definition}
A history $\hist= (\Events, \teo)$ is \emph{sequential} if 
for every $a,b,c\in \Events$, 
if $\tup{a,c}\in \st$, $\tup{a,b}\in \teo$, and $\tup{b,c}\in \teo$,
then $\tup{a,b}\in \st$.
\end{definition}

\begin{definition}
A transaction $\txid \in \TXIDs$ is \emph{legal} in a sequential history 
$\hist= (\Events, \teo)$ if every read in $\hist$ reads either internally or from the value of the latest
successful write, that is:
For every $r\in \Events$ with $\lLAB(r)=\readL x v$ and $\lTX(r)=\txid$
there exists $w \in \hist.\SSet  \cup \set{e \mid \lTX(e) = \txid}$ with $\lLAB(w)=\writeL x v$
such that the following hold:
\begin{itemize}
\item $\tup{w,r}\in \teo$.
\item No $w'\in \hist.\SSet \cup \set{e \in \Events \mid \lTX(e) = \txid}$ with $\lLAB(w)=\writeL x \_$
has $\tup{w,w'}\in \teo$ and $\tup{w',r}\in \teo$.
\end{itemize}
\end{definition}

\begin{definition}
Two histories $\hist= (\Events, \teo)$ and $\hist'= (\Events', \teo')$
are \emph{equivalent} if $\Events=\Events'$ and
$\teo_i = \teo'_i$.
\end{definition}

\begin{definition}
A history $\hist= (\Events, \teo)$ is \emph{crashless} if 
$\Events \cap \CrashMarkers = \emptyset$.
\end{definition}

\begin{definition}
A crashless history $\hist$ is \emph{final-state opaque} if 
there exists a sequential history $\hist_S$ such that
$\hist_S$ is equivalent to some completion $\hist'$ of $\hist$,
$\clo(\hist) \suq \clo(\hist_S)$,
and every transaction is legal in $\hist_S$.
\end{definition}

\begin{definition}
A crashless history $\hist$ is \emph{``original'' opaque} if 
every prefix $\hist'$ of $\hist$ is final-state opaque.
\end{definition}

\begin{definition}
A history $\hist$ is \emph{``original'' durably opaque} if 
$(\Events \setminus \CrashMarkers, \coerce \teo {\Events \setminus \CrashMarkers})$ is opaque.
\end{definition}

\begin{theorem}
\label{thm:org_opacity}
A crashless history $\hist= (\Events, \teo)$ is final-state opaque iff
there exist $\rf$ and $\mo$ such that the execution $(\Events, \teo_i, \clo(\hist), \rf, \mo)$ is opaque (\cref{def:opacity}).
\end{theorem}
\begin{proof}
$(\Rightarrow)$ Suppose that $\hist= (\Events, \teo)$ is final-state opaque.
Let $\hist_S= (\Events_S, \teo_S)$ be a a sequential history such that
$\hist_S$ is equivalent to some completion $\hist'$ of $\hist$,
$\clo(\hist) \suq \clo(\hist_S)$,
and every transaction is legal in $\hist_S$.
Let:
\begin{itemize}
\item $\rf=\set{\tup{w_r,r} \mid r\in\Reads}$, where 
for every read event $r\in\Events$, $w_r$ is the $\teo_S$-maximal
write to $\lLOC(r)$ in $\hist_S.\SSet \cup  \set{e \in \Events_S \mid \lTX(e) = \lTX(r)}$ that is $\teo_S$-before $r$;
and \item $\mo$ be the restriction of $\teo_S$ to writes on the same location.
\end{itemize} 

Let $\CPRFSet \eqdef \dom([\hist.\CPSet]; \rft)$
and  $\VSet \eqdef \hist.\SSet \cup \hist.\CPRFSet$.
We show that the conditions of \cref{def:opacity} hold
for $(\Events, \teo_i, \clo(\hist), \rf, \mo)$:

\begin{itemize}
\item $\rfI \cup \moI \cup \rbI \subseteq \teo_i $:
Directly follows from our construction.


\item 
$(\clo(\hist) \cup \rft \cup \mot \cup (\rbt;[\VSet]))$ is acyclic:
Observe that
$\clo(\hist) \cup \rft \cup \mot \cup (\rbt;[\VSet]))\suq \teo_S$.
Indeed, $\clo(\hist) \cup \rft \cup \mot \suq \teo_S$
directly follows from our construction.
For $\rbt;[\VSet]$, let $\tup{e_1,e_2}\in \rbt$ such that $e_2\in\VSet$.
Let $\tup{r,w}\in \rb \setminus \st$ such that $\tup{e_1,r}\in \st$ and $\tup{e_2,w}\in \st$.
Let $w'$ such that $\tup{w',r}\in \rf$ and $\tup{w',w}\in \mo$.
Then, our constriction ensures that $w'$ is $\teo_S$-before $r$ and $w$.
Since  $e_2\in\VSet$, we also have $w\in\VSet$, and having
$w$ $\teo_S$-before $r$ would contradict the 
the $\teo_S$-maximality of $w'$ (the $\rf$ source of $r$).
Thus, we have $\tup{r,w}\in \teo_S$,
and since $\hist_S$ is sequential, it implies $\tup{e_1,e_2}\in \teo_S$ as well.
\item $\dom(\rft) \suq \VSet$:
By definition we only have $\rf$ edge intra-transaction or from events in $\hist_S.\SSet$.
Since $\hist_S$ is equivalent to a completion of $\hist$, the latter events must be also in 
$\VSet = \hist.\SSet \cup \hist.\CPRFSet$.
\end{itemize}

\medskip

$(\Leftarrow)$ 
Let $\hist= (\Events, \teo)$.
Let $\rf$ and $\mo$ such that the execution $G=(\Events, \teo_i, \clo(\hist), \rf, \mo)$ is opaque.
We show that $\hist$ is final-state opaque.

Let $\hist'=(\Events', \teo')$ be the completion of $\hist$ obtained by
successfully committing every commit-pending transactions in $\hist$ with 
some event $e\in \dom(\rft)$,
and aborting all other pending transactions.

Let $\CPRFSet = \dom([\CPSet]; \rft)$ and
$\VSet = \SSet \cup \CPRFSet$.

Let $R$ be any strict total order on $\Begins \cap \Events$ (begin events in $\Events$) extending 
$(\clo(\hist) \cup \rft \cup \mot \cup (\rbt;[\VSet]))$.
We define $\teo_S \defeq \st\cap \teo_i \cup (R \setminus \st)$.
Let $\hist_S=(\Events', \teo_S)$.

It is easy to see that $\hist_S$ is a sequential history equivalent to $\hist'$ and that $\clo(\hist) \suq \clo(\hist_S)$.
We show that every transaction is legal in $\hist_S$.

Let $r\in \Events'$ with $\lLAB(r)=\readL x v$ and $\lTX(r)=\txid$.
Let $w$ be the $\rf$-source of $r$ in $G$.
Then, $\lLAB(w)=\writeL x v$.
We consider two possible cases:
\begin{itemize}
\item $\tup{w,r}\in \st$:
In this case, $\tup{w,r}\in \rfI$.
The opacity of $G$ ensures that $\tup{w,r}\in \teo_i$.
Since $\st \cap \teo_i\suq \teo_S$, we also have $\tup{w,r}\in \teo_S$.
Now, suppose for contradiction that there exists 
$w'\in \hist.\SSet \cup \set{e \in \Events \mid \lTX(e) = \txid}$ with $\lLAB(w)=\writeL x \_$,
$\tup{w,w'}\in \teo_S$, and $\tup{w',r}\in \teo_S$.
Since $\hist_S$ is sequential, we must have $\tup{w,w'}\in\st$.
Since $G$ is opaque, we have $\rfI \cup \moI \cup \rbI \subseteq \teo_i$,
from which it follows that $\tup{w,w'}\in\moI$.
Thus, we have $\tup{r,w'}\in\rbI$, and so $\tup{r,w'}\in\teo_i$,
which contradicts the fact that $\tup{w',r}\in \teo_S$.
\item $\tup{w,r}\nin \st$:
In this case, $\tup{w,r}\in \rft$,
and since $\dom(\rft) \suq \VSet$, our construction ensures
that $w \in \hist_S.\SSet$.
Moreover, since $\rft\suq R$, we have $\tup{w,r}\in \teo_S$.
Now, suppose for contradiction that there exists 
$w'\in \hist.\SSet \cup \set{e \in \Events \mid \lTX(e) = \txid}$ with $\lLAB(w)=\writeL x \_$,
$\tup{w,w'}\in \teo_S$, and $\tup{w',r}\in \teo_S$.
Consider three cases:
\begin{itemize}
\item $\tup{w,w'}\in \st$: Since $\tup{w,w'}\in \teo_S$, we must have $\tup{w,w'}\in (\teo_S)_i$,
and since $\moI\suq (\teo_S)_i$, we have $\tup{w,w'}\in\moI$.
It follows that $\tup{w',r}\in \rbt ; [\VSet]$ and $\tup{w',r}\in \rft$.
This contradicts the acyclicity of $(\clo(\hist) \cup \rft \cup \mot \cup (\rbt;[\VSet]))$.
\item $\tup{w,w'}\nin \st$ and $\tup{w',r}\in \st$:
If $\tup{w',w}\in \mo$, we obtain an $\rft \cup \mot$ cycle.
Thus, we have $\tup{w,w'}\in \mo$, which implies $\tup{r,w'}\in\rbI$.
The opacity of $G$ ensures that  $\tup{r,w'}\in(\teo_S)_i$,
which again contradicts the fact that $\tup{w',r}\in \teo_S$.
\item $\tup{w,w'}\nin \st$ and $\tup{w',r}\nin \st$:
In this case we have $w'\in \hist.\SSet$,
and so $w' \in \VSet$. Since $\tup{w,w'}\in \teo_S$
and $R \setminus \st \suq \teo_S$,
it follows that $\tup{w,w'} \in \mo$.
Hence we have $\tup{r,w'}\in\rbt ;[\VSet]$,
which contradicts the fact $\tup{w',r}\in \teo_S$
(since $\rbt;[\VSet]) \suq R$ and $R \setminus \st \suq \teo_S$).
\end{itemize}
\end{itemize}
\end{proof}

The following follows directly from \cref{thm:org_opacity} since both the original definition
and \cref{def:history_dur_opacity} define their durable counterparts of opacity by simply ignoring crashes.

\begin{corollary}
  A history $\hist$ is ``original'' durably opaque iff it is durably
  opaque (\cref{def:history_dur_opacity}).
\end{corollary}


\end{document}